\pdfoutput=1

\newif\ifconference
\conferencefalse

\ifconference
\documentclass[sigconf,screen]{acmart}
\else
\documentclass[11pt]{article}
\fi
\usepackage{amsmath,amsfonts}
\usepackage[textwidth=2in]{todonotes}
\usepackage{mathtools}
\usepackage{amsthm}
\usepackage{moresize}
\usepackage{tikz}
\usepackage{booktabs}
\usepackage{pifont}
\usepackage[toc]{appendix}
\usepackage{autonum}
\usepackage{makecell}
\usepackage{float}
\usepackage[ruled,linesnumbered]{algorithm2e}
\usepackage{textcomp}
\usepackage{lineno}
\usepackage{microtype}
\ifconference \else
\usepackage[english]{babel}
\usepackage{amssymb}
\usepackage[utf8]{inputenc}
\usepackage[letterpaper, portrait, left=1in, right=1in, top=1in, bottom=1in]{geometry}
\usepackage[shortlabels]{enumitem}
\usepackage[numbers,sectionbib,longnamesfirst]{natbib}
\usepackage[hyperindex,breaklinks]{hyperref}
\usepackage{booktabs} 
\fi
\usepackage{cleveref}

\newcommand{\Ii}{\mathcal{I}}
\newcommand{\Ee}{\mathcal{E}}

\definecolor{ForestGreen}{gray}{0}

\usetikzlibrary{patterns}

\SetKwRepeat{Repeat}{repeat}{}

\makeatother

\newcommand{\jakub}[1]{}

\def\squareforqed{\leavevmode\hbox to.77778em{\hfil\vrule\vbox to.675em{\hrule width.6em\vfil\hrule}\vrule\hfil}} 

\SetKwFor{RepTimes}{repeat}{times}{end}

\newtheorem{theorem}{Theorem}
\newtheorem{corollary}[theorem]{Corollary}
\newtheorem{definition}[theorem]{Definition}
\newtheorem{lemma}[theorem]{Lemma}
\newtheorem{observation}[theorem]{Observation}
\newtheorem{fact}[theorem]{Fact}

\theoremstyle{remark}
\newtheorem*{remark}{Remark}

\usepackage{environ}
\NewEnviron{killcontents}{}

\ifconference
\setcopyright{acmlicensed}
\acmPrice{15.00}
\acmDOI{10.1145/3519935.3520059}
\acmYear{2022}
\copyrightyear{2022}
\acmSubmissionID{stoc22main-p748-p}
\acmISBN{978-1-4503-9264-8/22/06}
\acmConference[STOC '22]{Proceedings of the 54th Annual ACM SIGACT Symposium on Theory of Computing}{June 20--24, 2022}{Rome, Italy}
\acmBooktitle{Proceedings of the 54th Annual ACM SIGACT Symposium on Theory of Computing (STOC '22), June 20--24, 2022, Rome, Italy}


\fi

\allowdisplaybreaks

\date{}
\title{Edge Sampling and Graph Parameter Estimation via Vertex Neighborhood Accesses}
\ifconference
\author{Jakub Tětek}
\email{j.tetek@gmail.com}
\author{Mikkel Thorup}
\email{mikkel2thorup@gmail.com}
\affiliation{%
  \institution{University of Copenhagen}
  \streetaddress{Universitetsparken 3}
  \city{Copenhagen}
  \country{Denmark}
  \postcode{2400}
}
\else
\author{
	Jakub Tětek \qquad\quad Mikkel Thorup \vspace{5px}\\\texttt{\normalsize\{j.tetek,mikkel2thorup\}@gmail.com}\vspace{5px}\\Basic Algorithms Research Copenhagen\\University of Copenhagen 
}
\fi

\begin{document}
\def\eps{\varepsilon}
\def\epsilon{\varepsilon}
\ifconference \else
\maketitle
\fi
\begin{abstract}
In this paper, we consider the problems from the area of
sublinear-time algorithms of \emph{edge sampling}, \emph{edge
counting}, and \emph{triangle counting}. Part of our contribution is
that we consider three different settings, differing in the way in which one may access the neighborhood of a given vertex.
In previous work, people have considered \emph{indexed neighbor access}, with a query returning the $i$-th neighbor of a given vertex.
\emph{Full neighborhood access model}, which has a query that returns the entire neighborhood at
a unit cost, has recently been considered in the applied community. Between these, we propose \emph{hash-ordered} neighbor
access, inspired by coordinated sampling, where we have a global fully random\footnote{For counting edges, $2$-independence is sufficient.} hash function, and can access neighbors in order of their hash values, paying a constant for each accessed neighbor.

For edge sampling and counting, our new lower bounds are in the most
powerful full neighborhood access model. We provide matching upper
bounds in the weaker hash-ordered neighbor access model. Our
new faster algorithms can be provably implemented efficiently
on massive graphs in external memory and with the current
APIs for, e.g., Twitter or Wikipedia. For triangle counting, we
provide a separation: a better upper bound with full neighborhood
access than the known lower bounds with indexed neighbor access. 
The technical core of our paper is our edge-sampling algorithm on which the other results depend. We now describe our results on the classic problems of edge and triangle counting.

We give an algorithm that uses hash-ordered neighbor access to approximately count edges in time $\tilde{O}(\frac{n}{\epsilon \sqrt{m}} + \frac{1}{\epsilon^2})$ (compare to the state of the art without hash-ordered neighbor access of $\tilde{O}(\frac{n}{\epsilon^2 \sqrt{m}})$ by Eden, Ron, and Seshadhri [ICALP 2017]). We present an $\Omega(\frac{n}{\epsilon \sqrt{m}})$ 
lower bound for $\epsilon \geq\sqrt{m}/n$ in the full neighborhood access model. This improves
the lower bound of $\Omega(\frac{n}{\sqrt{\epsilon m}})$ by
Goldreich and Ron [Rand. Struct. Alg. 2008]) and it matches
our new upper bound for $\epsilon \geq \sqrt{m}/n$. 
We also show an algorithm that uses the more standard assumption of pair queries (``are the vertices $u$ and $v$ adjacent?''), with time complexity of $\tilde{O}(\frac{n}{\epsilon \sqrt{m}} + \frac{1}{\epsilon^4})$. This matches our lower bound for $\epsilon \geq m^{1/6}/n^{1/3}$.

Finally, we focus on triangle counting. For this, we use the full power of the full neighbor access. In the indexed neighbor model, an algorithm that makes $\tilde{O}(\frac{n}{\epsilon^{10/3} T^{1/3}} + \min(m,\frac{m^{3/2}}{\epsilon^{3} T}))$ queries for $T$ being the number of triangles, is known and this is known to be the best possible up to the dependency on $\epsilon$ (Eden, Levi, Ron, and Seshadhri [FOCS 2015]). We improve this significantly to $\tilde{O}(\min(n,\frac{n}{\epsilon T^{1/3}} + \frac{\sqrt{n m}}{\epsilon^2 \sqrt{T}}))$ full neighbor accesses, thus
showing that the full neighbor access is fundamentally stronger for triangle counting than the weaker indexed neighbor model. We also give a lower bound, showing that this is the best possible with full neighborhood access, in terms of $n,m,T$.
\end{abstract}

\ifconference \else
\thispagestyle{empty}
\clearpage
\setcounter{page}{1}
\fi

\section{Introduction}
In this paper, we consider three well-studied problems from the area of sublinear-time algorithms: edge sampling, edge counting, and triangle counting in a graph. All of the three problems we consider have been studied before (in their approximate versions); see \cite{Eden2018,Eden2020} for edge sampling, \cite{Feige2004,Goldreich2008,Seshadhri2015,Eden2017,dasgupta2014estimating} for edge counting, and \cite{Itai1977,Kolountzakis2012,Eden2015,bera2020count} for triangle counting.
We first give an algorithm for exact edge sampling. We then apply this algorithm to both edge and triangle counting. We consider these three problems both in the well-studied indexed neighbor model\footnote{In this setting, we have the following queries: given a vertex, return its degree; given a vertex $v$ and a number $i \leq d(v)$, return the $i$-th neighbor $v$ in an arbitrary ordering; return a random vertex.}, but also in two new models that we introduce.

The first of them is the \emph{full neighborhood access model}. This model has recently been considered in the applied community \cite{Eliezer2021}, and similar settings are commonly used in practice, as we discuss in \cref{sec:model}. In this model, upon querying a vertex, the algorithm receives the whole neighborhood. To the best of our knowledge, we are the first to formally define this model and to give an algorithm with provable guarantees on both correctness and its query complexity. In this model, we get an algorithm for triangle counting significantly more efficient than what is possible in the indexed neighbor model.

We also introduce a model we call \emph{hash-ordered neighbor access}. This is an intermediate model, stronger than the indexed neighbor model, but weaker than the full neighborhood access model. We show that for edge sampling and counting, this model is sufficient to get an algorithm that nearly matches lower bounds (which we also prove) that work in models even stronger than the full neighborhood access model. The queries provided by the hash-ordered neighbor access can be implemented efficiently (see \Cref{sec:model}). Interestingly, the same data structure can be used to implement pair queries as well as hash-ordered neighbor access. This model formalizes the fact that the algorithms for edge counting and sampling only use the full neighborhood access in a very limited way.\ifconference \else Moreover, our algorithms for edge sampling and counting are such that they may be efficiently implemented using queries implemented by, e.g., Twitter's, and Wikipedia's APIs, well as in some external memory setting; we discuss this in \Cref{sec:implement_with_api}. \fi

To appreciate our bounds, note that in a graph consisting of a clique of size\footnote{We use $n$ and $m$ to denote the number of vertices and edges in the graph, respectively.} $\approx \sqrt{m}$ and the rest being an independent set, we need $n/\sqrt{m}$ queries just to find one edge. This provides a lower bound for both edge sampling and counting.

\paragraph{Dependency on $\epsilon$.}
The focus of sublinear-time algorithms is usually on approximate solutions, as many problems cannot be solved exactly in sublinear time. We thus have some error parameter $\epsilon$, that controls how close to the exact solution the output of our algorithm should be.
Throughout this paper, we put emphasis not only on the dependence of the running time on $n$ and $m$ but also on $\epsilon$. After all, $\epsilon$ can be polynomial in $n$ or $m$ (that is, it may hold $\epsilon = \Omega(n^\delta)$ for some $\delta>0$). While the dependency on $\epsilon$ in sublinear-time algorithms has often been ignored, we believe it would be a mistake to disregard it. We are not the only ones with this opinion. For example, Goldreich says in his book \cite[page~200]{Goldreich2018} that he ``begs to disagree'' with the sentiment that the dependency on $\epsilon$ is not important and stipulates that ``the dependence of the complexity on the approximation parameter is a key issue".

\paragraph{Sampling edges.} The problem of sampling edges has been first systematically studied by \citet{Eden2018} (although it was previously considered in \cite{Kaufman2004}). They show how to sample an edge pointwise $\epsilon$-close to uniform (see \Cref{def:pointwise}) in time $O(\frac{n}{\sqrt{\epsilon m}})$ in the indexed neighbor setting, and they prove this is optimal in terms of $n,m$. We show that the power of this setting is sufficient to improve the time complexity exponentially in $\epsilon$, to $O(\frac{n \log \epsilon^{-1}}{\sqrt{m}}) = \tilde{O}(\frac{n}{\sqrt{m}})$. 

Sampling multiple edges has recently been considered by \citet{Eden2021}. They present an algorithm that runs in time\footnote{The authors claim complexity $\tilde{O}(\sqrt{s}\frac{n}{\epsilon \sqrt{m}})$, which is sublinear in $s$. However, one clearly has to spend at least $\Omega(s)$ time and perform $\Omega(\min(m,s))$ queries.} $\tilde{O}(\sqrt{s}\frac{n}{\epsilon \sqrt{m}} + s)$ and samples $s$ edges pointwise $\epsilon$-close to uniform with high probability, but they do not prove any lower bounds. We prove that their algorithm is optimal in terms of $n,m,s$.

We give more efficient algorithms with the hash-ordered neighbor access. Specifically, we give the first sublinear-time algorithm that w.h.p.\ returns a sample of edges from \emph{exactly} uniform distribution. It runs in expected time $\tilde{O}(\sqrt{s}\frac{n}{\sqrt{m}} + s)$, which is the same complexity as that from \cite{Eden2021} for approximate edge sampling for constant $\epsilon = \Omega(1)$ (we solve exact edge sampling, which is equivalent to the case $\epsilon = 0$). We give a near-matching lower bound for all choices of $n,m,s$. Apart from sampling with replacement, our methods also lead to algorithms for sampling without replacement and Bernoulli sampling.\footnote{{Bernoulli sampling is defined as} sampling each edge independently with some probability $p$.}

Apart from being of interest in its own right, the problem of sampling multiple edges is also interesting in that it sheds light on the relationship between two standard models used in the area of sublinear algorithms. While many algorithms only use vertex accesses, many also use random edge sampling. An algorithm for uniform edge sampling then can be used to simulate random edge queries. Our algorithm is not only more efficient, but also has an advantage over the algorithm in \cite{Eden2021} that it can be used as a black box. This is not possible with their algorithm as it only samples edges \emph{approximately} uniformly.

We use this reduction between the settings with and without random edge queries in our new algorithm for triangle counting. We prove that, perhaps surprisingly, this reduction results in near-optimal complexity in terms of $n,m,T$ (where $T$ is the number of triangles). We also use our edge sampling algorithm for counting edges\footnote{We do not directly use the result on sampling $s$ edges for fixed $s$. Instead, we use a variant which instead samples each edge independently with some given probability $p$.}, also resulting in a near-optimal complexity, this time even in terms of $\epsilon$, for $\epsilon \geq \sqrt{m}/n$.

Since we consider edge sampling to be the technical core of our paper, we focus on that in this extended abstract. We defer the rest of our results to the full version of the paper.

\paragraph{Counting edges.} The problem of counting edges in sublinear time was first considered by \citet{Feige2004}. In his paper, he proves a new concentration inequality and uses it to give a $2+\epsilon$ approximation algorithm for counting edges running in time $O(\frac{n}{\epsilon \sqrt{m}})$. This algorithm only uses random vertex and degree queries but no neighbor access. It is also proven in \cite{Feige2004} that in this setting, $\Omega(n)$ time is required for $2-\epsilon$ approximation for any $\epsilon > 0$.

Since we are dealing with a graph, it is natural to also consider a query that allows us to access the neighbors of a vertex. \citet{Goldreich2008} use indexed neighbor queries to break the barrier of $2$-approximation and show a $(1\pm\epsilon)$-approximation that runs in time $\tilde{O}(\frac{n}{\epsilon^{4.5}\sqrt{m}})$. They also present a lower bound of $\Omega(\frac{n}{\sqrt{\epsilon m}})$. To prove this lower bound, they take a graph with $m$ edges and add a clique containing $\epsilon m$ edges. To hit the clique with constant probability, $\Omega(\frac{n}{\sqrt{\epsilon m}})$ vertex samples are required. 

Using a clever trick based on orienting edges towards higher degrees, \citet{Seshadhri2015} shows a much simpler algorithm. This approach has been later incorporated into the journal version of \cite{Eden2015} and published in that paper. The trick of orienting edges also led to an algorithm for estimating moments of the degree distribution \cite{Eden2017}. The moment estimation algorithm can estimate the number of edges in time $\tilde{O}(\frac{n}{\epsilon^2 \sqrt{m}})$ by estimating the first moment -- the average degree. This is currently the fastest algorithm known for counting edges.

We show two more efficient algorithms that use either the pair queries or the hash-ordered neighbor access. Specifically, in this setting, we give an algorithm that approximately counts edges in time $\tilde{O}(\frac{n}{\epsilon \sqrt{m}} + \frac{1}{\epsilon^2})$. This bound is strictly better than the state of the art (assuming $\epsilon \ll 1$ and $m \ll n^2$). We also show that the (in some sense) slightly weaker setting\footnote{Any algorithm with pair queries with time/query complexity $Q$ can be simulated in $O(Q \log \log n)$ time/queries using hash-ordered neighbor access. We discuss this in \Cref{sec:model}} of indexed neighbor with pair queries (``are vertices $u$ and $v$ adjacent") is sufficient to get an algorithm with time complexity $\tilde{O}(\frac{n}{\epsilon \sqrt{m}} + \frac{1}{\epsilon^4})$. This improves upon the state of the art for $\epsilon$ being not too small. Our methods also lead to an algorithm in the indexed neighbor access setting that improves upon the state of the art for $\epsilon$ small enough.

We show lower bounds that are near-matching for a wide range of $\epsilon$. Specifically, we prove that $\Omega(\frac{n}{\epsilon \sqrt{m}})$ samples are needed for $\epsilon \geq\sqrt{m}/n$, improving in this range upon $\Omega(\frac{n}{\sqrt{\epsilon m}})$ from \cite{Goldreich2008}. This lower bound holds not only with full neighbor access, but also in some more general settings. For example, Twitter API implements a query that also returns the degrees of the neighbors. Our lower bound also applies to that setting.

\paragraph{Triangle counting.}
The number of triangles $T$ in a graph can be trivially counted in time $O(n^3)$. This has been improved by \citet{Itai1977} to $O(m^{3/2})$. This is a significant improvement for sparse graphs. The first improvement for \emph{approximate} triangle counting has been given by \citet{Kolountzakis2012}, who improved the time complexity to $\tilde{O}(m+\frac{m^{3/2}}{\epsilon^2 T})$ (recall that $T$ is the number of triangles). This has been later improved by \citet{Eden2017} to $\tilde{O}(\frac{n}{\epsilon^{10/3}T^{1/3}} + \min(m,\frac{m^{3/2}}{\epsilon^{3}T}))$. In that paper, the authors also prove that their algorithm is near-optimal in terms of $n,m,T$.

Variants of the full neighborhood access model are commonly used in practice, and the model has been recently used in the applied community \cite{Eliezer2021}. Perhaps surprisingly, no algorithm performing asymptotically fewer queries than the algorithm by \citet{Eden2017} is known in this setting. Since the number of neighborhood queries is often the bottleneck of computation (the rate at which one is allowed to make requests is often severely limited), more efficient algorithms in this model could significantly decrease the cost of counting triangles in many real-world networks. We fill this gap by showing an algorithm that performs $\tilde{O}(\min(n,\frac{n}{\epsilon T^{1/3}} + \frac{\sqrt{nm}}{\epsilon^2 \sqrt{T}}))$ queries. This is never worse than $\tilde{O}(\frac{n}{\epsilon^{10/3}T^{1/3}} + \frac{m^{3/2}}{\epsilon^{3}T})$ and it is strictly better when $T \ll m^{9/4}/n^{3/2}$ or $\epsilon \ll 1$. This also improves the complexity in terms of $n,m,T$. Our result also proves a separation between the two models, as the algorithm by \citet{Eden2017} is known to be near-optimal in terms of $n,m,T$ in the indexed neighbor model. Our triangle counting algorithm relies on our result for sampling edges, showcasing the utility of that result. Using the algorithm of \citet{Eden2018} to simulate the random edge queries would result in both worse dependency on $\epsilon$ and a more complicated analysis. 
We also prove near-matching lower bounds in terms of $n,m,T$.

%

%
%

\paragraph{Setting without random vertex/edge queries.}
If we are not storing the whole graph in memory, the problem of sampling vertices in itself is not easy. There has been work in the graph mining community that assumes a model where random vertex or edge queries are not available and we are only given a seed vertex. The complexity of the algorithms is then parameterized by an upper bound on the mixing time of the graph. The problem of sampling vertices in this setting has been considered by \citet{Chierichetti2018,Eliezer2021}. The problems of approximating the average degree has been considered by \citet{dasgupta2014estimating}. Counting triangles in this setting has been considered by \citet{bera2020count}.

\ifconference
\else
\subsection{What \emph{is} a vertex access?} \label{sec:model}
\fi
\paragraph{Motivation behind full neighborhood access.} 
All sublinear-time algorithms with asymptotic bounds on complexity published so far assume only a model which allows for indexed neighbor access. However, this model is usually too weak to model the most efficient ways of processing large graphs, as non-sequential access to the neighborhood is often not efficient or not possible at all. The full neighborhood access model attempts to capture this. For example, to access the $i$-th neighbor of a vertex in the Internet graph\footnote{Internet graph is a directed graph with vertices being webpages and a directed edge for each link.}, one has to, generally speaking, download the whole webpage corresponding to the vertex. Similarly, when accessing a real-world network through an API (this would usually be the case when accessing the Twitter graph, Wikipedia graph, etc.), while it is possible to get just the $i$-th neighbor of a vertex, one may often get the whole neighborhood at little additional cost. The reason is that the bottleneck is usually the limit on the allowed number of queries in some time period and in one query, one may usually get many neighbors. While there is usually a limit on this number of neighbors that can be fetched in one query, this limit is often large enough that for the vast majority of vertices, the whole neighborhood can be returned as a response to one query. 

\ifconference \else
As a sidenote, this limit is typically larger than the average degree. For example, in the case of Twitter, the average degree is $1414$ \cite{Twitter_study}\footnote{The study considers the average number of followers of \emph{an active account} but also count follows by inactive accounts. The actual average degree is thus likely somewhat lower. This number however suffices for our argument.} while the limit is $5000$ \cite{Twitter_api}. 
This fact can be used to formally prove that our edge sampling and counting algorithms can be efficiently implemented using standard API calls, as provided for example by Twitter or Wikipedia, as well as when the graph is stored on a hard drive (discussed below). We discuss the details of this in \ifconference the full version of this paper. \else \Cref{sec:implement_with_api}. \fi
\fi

Although the above-mentioned APIs do not support random vertex queries, there are methods that implement the random vertex query and are efficient in practice \cite{Eliezer2021,Chiericetti2016}. In the full neighborhood access model, we do not assume any specific vertex sampling method; any algorithm implementing the random vertex query may be used. This further strengthens the case for our model. If it takes several API calls to get one random vertex, one may perform multiple API calls on the neighborhood of each randomly sampled vertex without significantly increasing the complexity.

Another motivation for the full neighborhood access comes from graphs stored in external memory. When storing the graph on a hard drive, one may read $\approx 1$-$3$ MB in the same amount of time as the overhead caused by a non-linear access\footnote{Consider for example Seagate ST4000DM000 \cite{Seagate} and Toshiba MG07SCA14TA \cite{Toshiba}. These are two common server hard drives. The average seek (non-linear access) times of these hard drives are $12$ and $8.5$ ms, respectively. Their average read speeds are $146$ and $260$ MB/s, respectively. This means that in the time it takes to do one seek, one can read $\approx 1.75$ and $\approx 2.2$ MB, respectively.}. If each neighbor is stored in roughly 10-30 bytes, then when we access one neighbor, we may read on the order of $10^5$ neighbors while increasing the running time only by a small constant factor.

These considerations suggest that an algorithm in the full neighborhood access model with lower complexity may often be preferable to an asymptotically less efficient algorithm in the weaker indexed neighbor access setting. It is also for these reasons that this model has been recently used in the applied community \cite{Eliezer2021}.

\paragraph{Lower bounds and the full neighbor access.} 
While some lower-bounds in the past have been (implicitly) shown in the full neighborhood access model (such as the one for edge counting in \cite{Goldreich2008}), others do not apply to that setting (such as the one for triangle counting in \cite{Eden2017}; in fact, we prove that it does not hold in the full neighborhood access model). As we have argued, there are many settings where one can easily get many neighbors of a vertex at a cost similar to getting one vertex. Lower bounds proven in the indexed neighbor access model then do not, in general, carry over to these settings. This highlights the importance of proving lower bounds in the full neighborhood model, which are applicable to such situations.

\paragraph{Motivation behind hash-ordered neighbor access.}
We introduce a model suitable for locally stored graphs, inspired by coordinated sampling. This model is also suitable for graphs stored in external memory. We call this model the hash-ordered neighbor access. It is an intermediate model, stronger than the indexed neighbor access but weaker than the full neighborhood access model.

While in the indexed neighbor model, the neighbors can be ordered arbitrarily, this is not the case with hash-ordered neighbor access. 
In this setting, we have a global hash function $h: V \rightarrow (0,1]$ which we may evaluate. Moreover, we assume that the neighbors of a vertex are ordered with respect to their hash values.


The simplest way to implement hash-ordered neighbor access is to store for each vertex its neighborhood in an array, sorted by the hash values. This has no memory overhead as compared to storing the values in an array. One may also efficiently support hash-ordered access on dynamic graphs using standard binary search trees.

We believe that the hash-value-ordered array is also a good way to implement pair queries when storage space is scarce. We may tell whether two vertices $u$ and $v$ are adjacent as follows. We evaluate $h(u)$ and search the neighborhood of $v$ for a vertex with this hash value. We use that the hash values are random, thus allowing us to use interpolation search. This way, we implement the pair query in time $O(\log \log d(v)) \leq O(\log \log n)$ \cite{Peterson1957,Armenakis1985}.

We believe that the hash-ordered neighbor access can be useful for solving a variety of problems in sublinear time. Specifically, we show that it allows us to sample higher-degree vertices with higher probability --- something that could be useful in other sublinear-time problems.


\begin{table*}[t]
\resizebox{\ifconference\else \textwidth \fi}{!}{%
\begin{tabular}{|l|l|l|l|l|}
\hline
& \makecell[l]{Upper bound\\ previous work} & \makecell[l]{Upper bound\\ this paper} & \makecell[l]{Lower bound\\ previous work} & \makecell[l]{Lower bound\\ this paper} \\ \hline
 \multicolumn{5}{|c|}{\textbf{Sampling one edge}}\\ \hline
 Indexed neighbor & $\tilde{O}(\frac{n}{\sqrt{\epsilon m}})$ & $O(\frac{n\log \epsilon^{-1}}{\sqrt{m}})$ & $\Omega(\frac{n}{\sqrt{m}})$ &   \\ \hline
 \multicolumn{5}{|c|}{\textbf{Sampling $s$ edges}}\\ \hline
 \makecell[l]{Indexed neighbor} & $\tilde{O}(\sqrt{s} \frac{n}{\epsilon \sqrt{m}} + s)$ & $\tilde{O}(\sqrt{s n} + s)$ for $m\!\geq\!\Omega(n)$ & &  \\ \hline
 \makecell[l]{Hash-ordered neighbor access} & & $\tilde{O}(\sqrt{s} \frac{n}{\sqrt{m}} + s)$ & $\Omega(\frac{n}{\sqrt{m}} + s)$ & $\Omega(\sqrt{s}\frac{n}{\sqrt{m}} + s)$ \\ \hline
 
 \multicolumn{5}{|c|}{\textbf{Edge counting}}\\ \hline
 
 Indexed neighbor & $\tilde{O}(\frac{n}{\epsilon^2 \sqrt{m}})$ & $\tilde{O}(\frac{\sqrt{n}}{\epsilon} + \frac{1}{\epsilon^2})$ for $m\!\geq\!\Omega(n)$ &  &   \\ \hline
 \makecell[l]{Hash-ordered neighbor access} & & $\tilde{O}(\frac{n}{\epsilon \sqrt{m}} + \frac{1}{\epsilon^2})$ & &  \\ \hline
 \makecell[l]{Indexed neighbor \\+ pair queries} & & $\tilde{O}(\frac{n}{\epsilon \sqrt{m}} + \frac{1}{\epsilon^4})$ & & \\ \hline
 Full neighborhood access & & & $\Omega(\frac{n}{\sqrt{\epsilon m}})$ & \makecell[l]{$\Omega(\frac{n}{\epsilon \sqrt{m}})$ for $\epsilon\!\geq\!\Omega(\frac{\sqrt{m}}{n})$} \\ \hline
 
 \multicolumn{5}{|c|}{\textbf{Triangle counting}}\\ \hline
\makecell[l]{Indexed neighbor \\w/ random vertex query} & \makecell[l]{$\tilde{O}(\frac{n}{\epsilon^{10/3} T^{1/3}} + \frac{m^{3/2}}{\epsilon^{3}T})$} && $\Omega(\frac{n}{T^{1/3}} + \frac{m^{3/2}}{T})$ &\\\hline
\makecell[l]{Indexed neighbor \\w/ random edge query} & $\tilde{O}(m^{3/2}/(\epsilon^2 T))$ && $\Omega(m^{3/2}/T)$ &\\\hline
\makecell[l]{Full neighborhood access \\w/ random edge query} && $\tilde{O}(m/(\epsilon^2 T^{2/3}))$ && $\Omega(m/T^{2/3})$\\\hline
\makecell[l]{Full neighborhood access \\w/ random vertex query} && $\tilde{O}(\frac{n}{\epsilon T^{1/3}} + \frac{\sqrt{nm}}{\epsilon^2 \sqrt{T}})$ && $\Omega(n/T^{1/3} + \sqrt{nm/T})$\\\hline
\makecell[l]{Full neighborhood access \\w/ random vertex, edge queries} && $\tilde{O}(\epsilon^{-2} \min(\frac{m}{T^{2/3}},\sqrt{\frac{nm}{T}}))$ && $\Omega(\min(\frac{m}{T^{2/3}},\sqrt{\frac{nm}{T}}))$\\\hline
 
\end{tabular}
}

\ifconference \else
\caption{Comparison of our results with previous work. Note that the empty cells in the table do not imply that nothing is known about the problems --- an algorithm that works in some model can also be used in any stronger model and similarly a lower bound that holds in some model also holds in any weaker model. Any problem can be trivially solved in $O(n+m)$. We do not make this explicit in the bounds. Similarly, any stated lower bound is assumed to hold in the sublinear regime, unless specified otherwise. } \label{tbl:summary}
\fi
\end{table*}

\ifconference
\else
\subsection{Our techniques} \label{sec:techniques}
\fi
In the part of this paper where we consider edge sampling, we replace each (undirected) edge by two directed edges in opposite directions. We then assume the algorithm is executed on this directed graph.

\subsubsection{Sampling one edge by a random length random walk}
The algorithm for sampling one edge is essentially the same as that for sampling an edge in bounded arboricity graphs from \cite{Eden2019}. We use different parameters and a completely different analysis to get the logarithmic dependence on $\epsilon^{-1}$.

We call a vertex $v$ heavy if $d(v) \geq \theta$ and light otherwise, for some parameter $\theta$ that is to be chosen later. A (directed) edge $uv$ is heavy/light if $u$ is heavy/light. Instead of showing how to sample edges from a distribution close to uniform, we show an algorithm that samples each (directed) edge with probability in $[(1-\epsilon) c/(2m), c/(2m)]$ for some $c>0$ and fails otherwise (with probability $\approx 1-c$). One may then sample an edge $1+\epsilon$ pointwise close to random by doing in expectation $\approx 1/c$ repetitions.

It is easy to sample light edges with probability exactly $c/(2m) = (n \theta)^{-1}$ -- one may pick vertex $v$ at random, choose $j$ uniformly at random from $[\theta]$ and return the $j$-th outgoing edge incident to the picked vertex. Return ``failure" if $d(v) < j$. We now give an intuition on how we sample heavy edges.

We set $\theta$ such that at least one half of the neighbors of any heavy vertex are light (we need a constant factor approximation of $m$ for this; we use one of the standard algorithms to get it). Consider a heavy vertex $v$. We use the procedure described above to sample a (directed) light edge $uw$ and we consider the vertex $w$. Since at least one half of the incoming edges of any heavy vertex are light, the probability of picking $v$ is in $[c d(v) / (4m), c d(v)/(2m)]$. Sampling an incident edge, we thus get that each heavy edge is sampled with probability in $[c/(4m),c/(2m)]$. 

Combining these procedures for sampling light and heavy edges (we do not elaborate here on how to do this), we may sample edges such that for some $c'>0$ we sample each light edge with probability $c'/(2m)$ while sampling each heavy edge with probability in $[c'/(4m),c'/(2m)]$ (the value $c'$ is different from $c$ due to combining the procedures for sampling light and heavy edges).

We now show how to reduce the factor of $2$ to $1\pm\epsilon$. Consider a heavy vertex $v$. Pick a directed edge $uw$ from the distribution of the algorithm we just described and consider $w$. The probability that $w = v$ is in $[\frac{3 c' d(v)}{8 m}, \frac{c' d(v)}{2m}]$, as we now explain. Let $h_v$ be the fraction of neighbors of $v$ that are heavy. Light edges are picked with probability $c'/(2m)$ and at least half of the incoming edges are light. The remaining edges are picked with probability in $[c'/(4m), c'/(2m)]$. The probability of sampling $v$ is then $\geq (1-h_v) c'/(2m) + h_v c'/(4m) \geq \frac{3c'}{8m}$ because $h_v \leq 1/2$. The probability is also clearly $\leq c'/(2m)$, thus proving the claim. Combining light edge sampling and heavy edge sampling (again, we do not elaborate here on how to do this), we are now able to sample an edge such that each light edge is sampled with probability $c''/(2m)$ and each heavy edge with probability in $[\frac{3c''}{8m}, \frac{c''}{2m}]$ (again, the value $c''$ is different from $c,c'$ due to combining light and heavy edge sampling). Iterating this, the distribution converges pointwise to uniform at an exponential rate.

One can show that this leads to the following algorithm based on constrained random walks of random length. The length is chosen uniformly at random from $[k]$ for integer $k \approx \lg \epsilon^{-1}$. The algorithm returns the last edge of the walk. The random walk has constraints that, when not satisfied, cause the algorithm to fail and restart. These constraints are (1) the first vertex $v$ of the walk is light and all subsequent vertices are heavy except the last one (which may be either light or heavy) and (2) picking $X \sim Bern(d(v)/\sqrt{2m})$ \footnote{Bernoulli trial $Bern(p)$ is a random variable having value $1$ with probability $p$ and $0$ otherwise.}, the first step of the walk fails if $X = 0$ (note that this is equivalent to using rejection sampling to sample the first edge of the walk).
%
%

\subsubsection{Sampling multiple edges and edge counting using hash-ordered neighbor access}

The problems of edge counting and sampling share the property that in solving both these problems, it would be of benefit to be able to sample vertices in a way that is biased towards vertices with higher degree. This is clearly the case for edge sampling as it can be easily seen to be equivalent to sampling vertices with probabilities proportional to their degrees. Biased sampling is also useful for edge counting. If we let $v$ be chosen at random from some distribution and for a fixed vertex $u$ define $p_u = P(v = u)$, then $X = \frac{d(v)}{2 p_v}$ is an unbiased estimate of the number of edges, called the Horvitz-Thompson estimator \cite{Horvitz1952}. The variance is $Var(X) \leq E(X^2) = \sum_{v \in V}p_v \Big(\frac{d(v)}{p(v)}\Big)^2$ which decreases when high-degree vertices have larger probability $p_v$. The crux of this part of our paper, therefore, lies in how to perform this biased sampling.

\paragraph{Biased sampling procedure.} We first describe how to perform biased sampling in the case of the indexed neighbor access model and then describe a more efficient implementation in the hash-ordered neighbor access model. We say a vertex $v$ is heavy if $d(v) \geq \theta$ for some value $\theta$ and we say it is light otherwise. The goal is to find all heavy vertices in the graph as that will allow us to sample heavy vertices with higher probability. In each step, we sample a vertex and look at its neighborhood. This takes in expectation $O(m/n)$ time per step. After $\Theta(n \log (n)/\theta)$ steps, we find w.h.p. all vertices with degree at least $\theta$. We call this technique high-degree exploration.

How to exploit hash-ordered neighbor access? We do not actually need to find all heavy vertices in order to be able to sample from them when we have the vertex hash queries. We make a sample $S$ of vertices large enough such that, w.h.p., each heavy vertex has one of its neighbors in $S$ ($|S| = \Theta(n \log (n) /\theta$ suffices).
Pick all heavy vertices incident to the sampled vertices that have $h(v) \leq p$. We can find these vertices in constant time per vertex using hash-ordered access. Since each heavy vertex has at least one of its neighbors in $S$, these are in fact all heavy vertices $v$ with $h(v) \leq p$. Since the hash values are independent and uniform on $[0,1]$, each heavy vertex is (w.h.p.) picked independently with probability $p$.
This allows us to sample heavy vertices with larger probability than if sampling uniformly.

We apply this trick repeatedly for $k=1, \cdots, \log n$ with thresholds $\theta_k = 2^k \theta$ and $p_k = 2^k p$. This way, instead of having one threshold $\theta$ and one probability of being sampled for each $v$ with $d(v) \geq \theta$, we have logarithmically many thresholds and the same number of different probabilities. Since the ratio between $\theta_k$ and $p_k$ is constant, we see that the probability of each vertex $v$ with $d(v) \geq \theta$ being sampled is up to a factor of $2$ proportional to $d(v)$. Moreover, for each vertex, we know exactly its probability of being sampled.

\paragraph{Edge sampling.} We describe how to sample each edge independently with some probability $p$. This setting is called Bernoulli sampling. Light edges can be sampled in a way similar to that used for sampling light edges in the algorithm for sampling one edge. Using the biased sampling algorithm allows us to sample heavy vertices with higher probability. We then prove that for each vertex $v$, the probability that $v$ is sampled is at least the probability that one of the edges incident to $v$ is sampled when sampling each edge independently with probability $p$. This allows us to use rejection sampling to sample each vertex $v$ with a probability equal to the probability of at least one of its incident edges being sampled. By sampling $k$ incident edges of such vertex where $k$ is chosen from the right distribution, we get that each heavy edge is sampled independently with some fixed probability $p$. Sampling light and heavy edges separately and taking union of those samples gives us an algorithm that performs Bernoulli sampling from the set of all edges.

One can use Bernoulli sampling to sample edges without replacement. Specifically, when it is desired to sample $s$ edges without replacement, one may use Bernoulli sampling to sample in expectation, say, $2s$ samples and if at least $s$ are sampled, then return a random subset of the random edges, otherwise repeat.

To sample $s$ edges with replacement, we perform Bernoulli sampling $\Theta(s)$ times, setting the probability such that each time, we sample in expectation $O(1)$ edges. From each (non-empty) Bernoulli sample, we take one edge at random and add it into the sample.
While implementing this naively would result in a linear dependence on $s$, this can be prevented. The reason is that it is sufficient to perform the pre-processing for Bernoulli sampling only once. This allows us to spend more time in the ``high-degree exploration phase", making the Bernoulli sampling itself more efficient.

\paragraph{Edge counting by sampling.}
When using Bernoulli sampling with some inclusion probability $p$, the number of sampled edges $|S|$ is concentrated around $p m$. We estimate $m$ as $|S|/p$ and prove that for an appropriate choice of $p$, the approximation has error at most $\epsilon$ with high probability. The ``appropriate value'' of $p$ depends on $m$. We find it using a geometric search.

\paragraph{Counting edges directly.}
We also give an independent algorithm for edge counting based on a different idea. We use the biased sampling procedure to sample higher-degree vertices with higher probability. We then use the above-mentioned Horvitz-Thompson estimator; sampling higher-degree vertices with higher probability reduces the variance. We then take the average of an appropriate number of such estimators to sufficiently further reduce the variance. 

\subsubsection{Edge counting using pair queries}
\citet{Seshadhri2015} shows a bound on the variance of the following estimator: sample a vertex $v$, get a random neighbour $u$ of $v$, if $(d(v), id(v)) < (d(u), id(u))$ \footnote{We are assuming $id$ is a bijection between $V$ and $[n]$ and $<$ on the tuples is meant with respect to the lexicographic ordering.} then answer $n d(v)$, otherwise answer $0$. This is an unbiased estimate of $m$. We combine this idea with the technique of high-degree exploration. Direct all edges from the endpoint with lower degree to the one with higher degree. The biggest contributor to the variance of the estimator from \cite{Seshadhri2015} are the vertices that have out-degree roughly $\sqrt{m}$ (one can show that there are no higher-out-degree vertices). Our goal is to be able to sample these high out-degree ($d^+(v) \geq \theta$ for some parameter $\theta$) vertices with higher probability. Using the Horvitz-Thompson estimator for vertices of out-degree at least $\theta$ will decrease the variance. We use an estimator inspired by the one from \cite{Seshadhri2015} for vertices with out-degree $< \theta$.

We make a sample $S$ of vertices and let $S'$ be the subset of $S$ of vertices having degree $\geq \theta$. The intuition for why we consider $S'$ is the following.
If we pick $S$ to be large enough ($|S| = \Theta(n \log (n)/\theta)$ is sufficient), any vertex $v$ with $d^+(v) \geq \theta$ will have at least one of its out-neighbors sampled in $S$. Moreover, it holds $d(v) \geq d^+(v) \geq \theta$; since we direct edges towards higher-degree endpoints, these sampled out-neighbors also have degree $\geq \theta$. This means that they also lie in $S'$. Any high-out-degree vertex thus has, w.h.p., a neighbor in $S'$.
At the same time, $S'$ has the advantage of being significantly smaller than $S$ as there can only be few vertices with degree $\geq \theta$ in the graph (at most $2m/\theta$, to be specific), so each vertex of $S$ lies in $S'$ with probability $\leq 2m/(n \theta)$.

Now we pick each incident edge to $S'$ with a fixed probability $p$ and for each picked edge $uv$ for $u\in S'$, mark the vertex $v$. A vertex $v$ is then marked with probability $p_v = 1-(1-p)^{r(v)}$ where $r(v) = |N(v) \cap S'|$. It then holds $p_v \approx p r(v)$ (we ensure $r(v)$ is not too large, which is needed for this to hold). This is (w.h.p.) roughly proportional to the out-degree of $v$. Using the Horvitz-Thompson estimator, this reduces the variance. This suffices to get the improved complexity.

It remains to show how to efficiently compute $r(v)$ (we need to know the value in the Horvitz-Thompson estimator). We set the threshold $\theta$ so as to make sure that only a small fraction of all vertices can have degree greater than the threshold. Then, $S'$ will be much smaller than $S$ (as $S$ is a uniform sample). In fact, it will be so small that we can afford to use pair queries to check which of the vertices in $S'$ are adjacent to $v$. This is the main trick of our algorithm.

Another obstacle that we have to overcome is that when we are given a vertex, we cannot easily determine its out-degree. We need to know this to decide which of the above estimators to use for the vertex. Fortunately, the cost of estimating it is roughly inverse to the probability we will need the estimate, thus making the expected cost small.

\subsubsection{Triangle counting using full neighborhood access}
\paragraph{Warmup: algorithm with random edge queries.} We now show a warmup which assumes that we may sample edges uniformly at random. We then use this as a starting point for our algorithm. This warmup is inspired by and uses some of the techniques used in \cite{Kallaugher2019}.  

Consider an edge $e = uv$. By querying both endpoints, we can determine the number of triangles containing $e$ (because this number is equal to $|N(u) \cap N(v)|$). Let $t(e)$ be the number of triangles containing $e$. One of the basic ideas that we use is that we assign each triangle to its edge with the smallest value $t(e)$. This trick has been used before for edge counting in \cite{Seshadhri2015} and for triangle counting in \cite{Kallaugher2019}. We denote by $t^+(e)$ the number of edges assigned to $e$. Consider a uniformly random edge $e = uv$ and a uniform vertex $w \in N(u) \cap N(v)$. Let $X = t(e)$ if $uvw$ is assigned to $uv$ and $X = 0$ otherwise. The expectation is $E(X) = \sum_{e \in E} \frac{1}{m} \cdot \frac{t^+(e)}{t(e)}t(e) = T/m$ and we may thus give an unbiased estimator of $T$ as $mX$. The variance of $X$ is  $Var(X) \leq \sum_{e \in E} \frac{1}{m} \cdot \frac{t^+(e)}{t(e)}t(e)^2 = \frac{1}{m} \sum_{e \in E} t^+(e) t(e)$.
A bound used in \citet{Kallaugher2019} can be used to prove that this is $O(T^{4/3}/m)$. 
Taking $s = \Theta(m/(\epsilon^2 T^{2/3}))$ samples and taking the average then gives a good estimate with probability at least $2/3$ by the Chebyshev inequality.

We prove that this is optimal up to a constant factor in terms of $m$ and $T$ when only random edge queries (and not random vertex queries) are allowed. We will now consider this problem in the full neighborhood access model, which only allows for random vertex queries (and not random edge queries). Combining this algorithm with our edge sampling algorithm results in complexity $O(\sqrt{s}\frac{n}{\sqrt{m}} + s) = O(\frac{n}{\epsilon T^{1/3}} + \frac{m}{\epsilon^2 T^{2/3}})$ in the full neighborhood access model. This, however, is not optimal.

\paragraph{Sketch of our algorithm: algorithm with random vertex queries.} We will now describe a more efficient algorithm that uses both random edge and random vertex queries. We then remove the need for random edge queries by simulating them with our algorithm for edge sampling. Perhaps surprisingly, black-box application of our edge sampling algorithm results in near-optimal complexity. 

In order to find out the value $t(e)$ for as many edges as possible, one may sample each vertex independently with some probability $p$. For any edge $e$ whose both endpoints have been sampled, we can compute $t(e)$ with no additional queries. The sum of $t(e)$'s that we learn in this way is in expectation $3p^2T$. We may thus get an unbiased estimator of $T$. We may express the variance by the law of total variance and a standard identity for the variance of a sum. By doing this, we find out that there are two reasons the variance is high.
First, the variance of the number of triangles contributed to the estimate by one edge can be large\footnote{Formally, we are talking about the variance of a random variable $X_e$ equal to $t(e)$ if both endpoints of $e$ are sampled and $0$ otherwise.}. This is true for edges that are contained in relatively many triangles. The second reason is the correlations between edges: if we have edges $e,e'$ sharing one vertex and both endpoints of $e$ have been sampled, it is more likely that both endpoints of $e'$ are sampled, too. This introduces correlation between the contributions coming from different edges, thus increasing the variance. We now sketch a solution to both these issues.

The first issue could be solved by applying the above-described trick with assigning each triangle to its edge $e = uv$ with the smallest value $t(e)$. Pick $w$ uniformly from $N(u) \cap N(v)$ and let $X_e = t(e)$ if $uvw$ is assigned to $uv$ and let $X_e = 0$ otherwise. An analysis like the one described above would give good bounds on the variance of $X_e$. The issue with this is that for each edge with non-zero $t(e)$, we query the vertex $w$. If the number of edges in the subgraph induced by the sampled edges is much greater than the number of sampled vertices, this will significantly increase the query complexity. To solve this issue, we separately consider two situations. If an edge has many triangles assigned to it, we do the following. We sample a set of vertices $N$ (this set is shared for all edges) large enough such that, by the Chernoff bound, $\frac{|\{w \in N\cap N(u) \cap N(v),\text{ s.t. } uvw \text{ is assigned to }uv\}|}{|N\cap N(u) \cap N(v)|} \approx t^+(e)/t(e)$. Instead of sampling $w$ uniformly from $N(u) \cap N(v)$, we then sample from $N \cap N(u) \cap N(v)$. Since we do not need to query any vertex twice, we may bound the cost of this by $|N|$.
On the other hand, if the number of triangles assigned to $e$ is small, the variance of $X_e$ (which is proportional to $t^+(e) t(e)$) is relatively small (as $t^+(e)$ is small). We then may afford to only use the estimator $X_e$ with some probability $p'$ and if we do, we use $X_e/p'$ as the (unbiased) estimate of $t^+(e)$. This increases the variance contributed by the edge $e$, but we may afford this as it was small before applying this trick. This way, we have to only make an additional query with probability $p'$, thus decreasing the query complexity of this part of the algorithm.

The second problem is created by vertices whose incident edges have many tiangles (at least $\theta$) assigned to them (as becomes apparent in the analysis). A potential solution would be to not use this algorithm for the vertices with more than $\theta$ triangles assigned to incident edges and instead estimate the number of these triangles by the edge-sampling-based algorithm from the warm-up. Why is this better than just using that algorithm on its own? There cannot be many such problematic vertices. Namely, there can only be $\ell = T/\theta$ such vertices. This allows us to get a better bound on the variance. Specifically, instead of bounding the variance by $O(T^{4/3}/m)$, we prove a bound of $\tilde{O}(\ell T/m)$. An issue we then have to overcome is that we cannot easily tell apart the ``heavy" and ``light" vertices. Let us have a vertex $v$ and we want to know whether it is light (it has at most $\theta$ triangles assigned to its incident edges) or whether it is heavy. The basic idea is to sample some vertices (this set is common for all vertices) and then try to infer whether a vertex is with high probability light or whether it may potentially be heavy based on the number of triangles containing edges between $v$ and this set of sampled vertices.

These techniques lead to a bound of $\tilde{O}(\frac{n}{\epsilon T^{1/3}} + \frac{\sqrt{nm}}{\epsilon^2 \sqrt{T}})$. One can always read the whole graph in $n$ full neighborhood queries, leading to a bound of $\smash{\tilde{O}}(\min(n,\frac{n}{\epsilon T^{1/3}} + \frac{\smash{\sqrt{nm}}}{\epsilon^2 \sqrt{T}}))$.

\paragraph{Lower bounds.} Our lower bounds match our algorithms in terms of dependency on $n,m,T$. The lower bound of $\Omega(n/T^{1/3})$ is standard and follows from the difficulty of hitting a clique of size $T^{1/3}$. We thus need to prove a lower bound of $\Omega(\min(n,\sqrt{n m/T}))$. This amounts to proving $\Omega(\sqrt{n m/T})$ under the assumption $T \geq m/n$. We thus assume for now this inequality.

Our lower bound is by reduction from the OR problem (given booleans $x_1, \cdots, x_n$, compute $\bigvee_{i=1}^n x_i$) in a style similar to the reductions in \cite{Eden2018b}. The complexity of the OR problem is $\Omega(n)$. For an instance of the OR problem of size $\sqrt{nm/T}$, we define a graph $G$ with $\Theta(n)$ vertices and $\Theta(m)$ edges. The number of triangles is either $\geq T$ if $\bigvee_{i=1}^n x_i = 1$ or $0$ if $\bigvee_{i=1}^n x_i = 0$. Moreover, any query on $G$ can be answered by querying one $x_i$ for some $i \in [n]$. It follows that any algorithm that solves triangle counting in $G$ in $Q$ queries can be used to solve the OR problem of size $\Theta(\sqrt{n m / T})$ in $Q$ queries. This proves the desired lower bound.

We now describe the graph $G$. We define a few terms. A section consists of $4$ groups of $\sqrt{nT/m}$ vertices. The whole graph consists of sections and $m/n$ non-section vertices. There are $\sqrt{nm/T}$ sections, one for each $x_i$. In the $i$-th section, there is a complete bipartite graph between the first two groups of vertices if $x_i = 0$ and between the third and fourth if $x_i=1$. There is a complete bipartite graph between each third or fourth group of a section and the non-section vertices.
\ifconference \else See \Cref{fig:illustration_of_lb} on \cpageref{fig:illustration_of_lb} for an illustration of this construction.\fi
If $x_i = 0$ for all $i$, then $G$ is triangle-free. If $x_i=1$, then the $i$-th section together with the non-section vertices forms $\sqrt{nT/m} \cdot \sqrt{nT/m} \cdot m/n = T$ triangles. At the same time, a query ``within a section" only depends on the value $x_i$ corresponding to that section, so we can implement it by one query to the instance of the OR problem. A query that does not have both endpoints within the same section is independent of the OR problem instance. The number of vertices and edges is $\Theta(n)$ and $\Theta(m)$ as desired, and $G$ thus satisfies all conditions.

\section{Preliminaries} \label{sec:preliminaries}
\ifconference
\else
\subsection{Graph access models}
\fi
Since a sublinear-time algorithm does not have the time to pre-process the graph, it is important to specify what queries the algorithm may use to access the graph. We define the \textit{indexed neighbor access model} by the following queries:
\begin{itemize}
\item For $i \in [n]$, return the $i$-th vertex in the graph
\item For $v \in V$, return $d(v)$
\item For $v \in V$ and $i \in [d(v)]$, return the $i$-th neighbor of $v$
\item For a given vertex $v$, return $id(v)$ such that if $v$ is the $i$-th vertex, then $id(v) = i$
\end{itemize}
where the vertices in the graph as well as the neighbors of a vertex are assumed to be ordered adversarially. Moreover, the algorithm is assumed to know $n$. This definition is standard; see \cite{Goldreich2008} for more details. \emph{Pair queries} are often assumed to be available in addition to the queries described above:
\begin{itemize}
\item Given vertices $u,v$, return whether the two vertices are adjacent
\end{itemize}
This has been used, for example, in \cite{Eden2017,Eden2018,Assadi2020}.

In this paper, we introduce a natural extension of the above-described setting without pair queries, which we call the \emph{hash-ordered neighbor access model}. In this model, there is the following additional query
\begin{itemize}
\item For $v \in V$, return $h(v)$
\end{itemize}
where the hash of $v$, denoted $h(v)$, is a number picked independently uniformly at random from $[0,1]$. Moreover, neighborhoods of vertices are assumed to be ordered with respect to the hashes of the vertices. Our algorithms do not require the vertices to be ordered with respect to the hash values (in contrast to the neighborhoods), although that would also be a natural version of this model.

We define the \emph{full neighborhood access model} as follows. Each vertex $v$ has a unique $id(v) \in [n]$. We then have one query: return the $id$'s of all neighbors of the $i$-th vertex. We then measure the complexity of an algorithm by the number of queries performed, instead of the time complexity of the algorithm.

\ifconference
\else
\subsection{Pointwise \texorpdfstring{$\epsilon$}{epsilon}-approximate sampling}
\fi
\begin{definition} \label{def:pointwise}
A discrete probability distribution $\mathcal{P}$ is said to be pointwise $\epsilon$-close to $\mathcal{Q}$ where $\mathcal{P},\mathcal{Q}$ are assumed to have the same support, denoted $|\mathcal{P} - \mathcal{Q}|_P \leq \epsilon$, if
\[
|\mathcal{P}(x)-\mathcal{Q}(x)| \leq \epsilon \mathcal{Q}(x), \quad \text { or equivalently } \quad 1-\epsilon \leq \frac{\mathcal{P}(x)}{\mathcal{Q}(x)} \leq 1+\epsilon
\]
for all $x$ from the support.
\end{definition}
\noindent
In this paper, we consider distributions pointwise $\epsilon$-close to uniform. This measure of similarity of distributions is related to the total variational distance. Specifically, for any $\mathcal{P}, \mathcal{Q}$, it holds that $|\mathcal{P}- \mathcal{Q}|_{TV} \leq |\mathcal{P}- \mathcal{Q}|_{P}$ \cite{Eden2018}.

\ifconference
\else
\subsection{Conditioning principle}
\fi
Let $X,Y$ be two independent random variables taking values in a set $A$ and let $f$ be a function on $A$. If $Y \sim f(X)$, then $X \sim X | (f(X) = Y)$. In other words, if we want to generate a random variable $X$ from some distribution, it is sufficient to be able to generate (1) a random variable from the distribution conditional on some function of $X$ and (2) a random variable distributed as the function of $X$. We call this the conditioning principle. We often use this to generate a sample -- we first choose the sample size from the appropriate distribution and then sample the number of elements accordingly.

\ifconference
\else
\subsection{Notation}
\fi

We use relations $f(x) \lesssim g(x)$ with the meaning that $f(x)$ is smaller than $g(x)$ up to a constant factor. The relations $\simeq, \gtrsim$ are defined analogously. The notation $f(x) \sim g(x)$ has the usual meaning of $f(x)/g(x) \rightarrow 1$ for $x \rightarrow \infty$. We use $\lg x$ to denote the binary logarithm of $x$. We use $N(v)$ to denote the set of neighbors of $v$. Given a vertex $v$ and integer $i$, we let $v[i]$ to be the $i$-th neighbor of $v$. Given a (multi-)set of vertices $S$, we let $d(S) = \sum_{v \in S} d(v)$ and $d_S(v) = |N(v) \cap S|$. For an edge $e = uv$, we denote by $N(e)$ the set of edges incident to either $u$ or $v$. In addition to sampling with and without replacement, we use the less standard name of Bernoulli sampling. In this case, each element is included in the sample independently with some given probability $p$ which is the same for all elements. Given a priority queue $\mathcal{Q}$, the operation $\mathcal{Q}.top()$ returns the elements with the lowest priority. $\mathcal{Q}.pop()$ returns the element with the lowest priority and removes it from the queue.

We use $Bern(p)$ to denote a Bernoulli trial with bias $p$, $Unif(a,b)$ to be the uniform distribution on the interval $[a,b]$, $Bin(n,p)$ to be the binomial distribution with universe size $n$ and sample probability $p$. For distribution $\mathcal{D}$ and event $\mathcal{E}$, we use $X \sim (\mathcal{D} | \mathcal{E})$ to denote that $X$ is distribution according to the conditional distribution $\mathcal{D}$ given $\mathcal{E}$.

\ifconference \else
\ifconference
\subsection{Algorithms With Advice} \label{sec:removing_advice}
\else
\subsection{Algorithms with advice} \label{sec:removing_advice}
\fi
Many of our algorithms depend on the value of $m$ or $T$. We are however not assuming that we know this quantity (in fact, these are often the quantities we want to estimate). Fortunately, for our algorithms, it is sufficient to know this value only up to a constant factor. There are standard techniques that can be used to remove the need for this advice.


Specifically, we use the advice removal procedure from \cite{Tetek2021}. Similar advice removal procedures have been used before, for example in \cite{Eden2017,Goldreich2006,Aliakbarpour2018,Assadi2018,Eden2018}. This advice removal can be summarized as follows.
\begin{fact} \label{fact:advice_removal}
Let us have a graph parameter $\phi$ that is polynomial in $n$. Suppose there is an algorithm that takes as a parameter $\tilde{\phi}$ and has time complexity $T(n,\tilde{\phi},\epsilon)$ decreasing polynomially in $\tilde{\phi}$. Moreover, assume that for some $c > 1$, it outputs $\hat{\phi}$ such that $P(\hat{\phi} \geq c \phi) \leq 1/3$ and if moreover $\phi \leq \tilde{\phi} \leq c \phi$, then $P(|\hat{\phi} - \phi| > \epsilon \phi) \leq 1/3$. Then there exists an algorithm that has time complexity $O(T(n,\phi,\epsilon) \log \log n)$ and returns $\hat{\phi}$ such that $P(|\hat{\phi} - \phi| > \epsilon \phi)$ (and does not require advice $\tilde{\phi}$).
\end{fact}
We now give a sketch of the reduction. We start with $\tilde{\phi}$ with a polynomial upper-bound on $\phi$. For $m$, this may be $n^2$, for $T$, we can use $n^3$. We then geometrically decrease $\tilde{\phi}$. For each value of $\tilde{\phi}$, we perform probability amplification. Specifically, we have the algorithm run $\Theta(\log \log n)$ times and take the median estimate. We stop when this median estimate is $\geq c \tilde{\phi}$ and return the estimate. See \cite[Section~2.5]{Tetek2021} for details.

\ifconference
\subsection{Sampling Without Replacement}
\else
\subsection{Sampling without replacement}
\fi
Suppose we want to sample $k$ elements with replacement from the set $[n]$. This can be easily done in expected time $O(k)$. If $k > n/2$, we instead sample $n-k$ items without replacement and take the complement. We may, therefore, assume that $k \leq n/2$. We repeat the following $k$ times: we sample items with replacement until we get an element we have not yet seen before, which we then add into the sample. Since $k \leq n/2$, each repetition takes in expectation $O(1)$ time and the total time is in expectation $O(k)$.
\fi

\section{Edge sampling} \label{sec:sampling}
We start with some definitions which we will be using throughout this section. Given a threshold $\theta$ (the exact value is different in each algorithm), we say a vertex $v$ is heavy if $d(v) \geq \theta$ and light if $d(v) < \theta$. We denote the set of heavy (light) vertices by $V_H$ ($V_L$). In this section, we replace each unoriented edge in the graph by two oriented edges in opposite directions. We then assume the algorithm is executed on this oriented graph. We then call a (directed) edge $uv$ heavy (light) if $u$ is heavy (light). If we can sample edges from this oriented graph, we can also sample edges from the original graph by sampling an edge and forgetting its orientation. 

\ifconference
\else
\subsection{Sampling one edge in the indexed neighbor access model}
\fi
In this section, we show \Cref{alg:sample_edge} which samples an edge pointwise $\epsilon$-approximately in expected time $O(\frac{n}{\sqrt{m}}\log \frac{1}{\epsilon})$. This algorithm is, up to a change of parameters, the one used in \cite{Eden2019} but we provide a different analysis that is tighter in the case of general graphs (in \cite{Eden2019}, the authors focus on graphs with bounded arboricity). This algorithm works by repeated sampling attempts, each succeeding with probability $\approx \frac{\sqrt{m}}{n\log \epsilon^{-1}}$. We then show that upon successfully sampling an edge, the distribution is pointwise $\epsilon$-close to uniform.

\begin{algorithm}[h]
$\theta \leftarrow \lceil \sqrt{2 m} \rceil$

Sample a vertex $u_0 \in V$ uniformly at random\\
If $u_0$ is heavy, return ``failure"\\
Choose a number $j \in[\theta]$ uniformly at random\\
Let $u_1$ be the $j$'th neighbor of $u_0$; return ``failure" if $d(u_0) < j$\\
\For{$i$ from $2$ to $k$}{
If $u_{i-1}$ is light, return ``failure"\\
$u_i \leftarrow$ random neighbor of $u_{i-1}$\\
}
Return $(u_{k-1}, u_k)$.

\caption{\texttt{Sampling\_attempt($k$)} subroutine} \label{alg:sampling_attempt}
\end{algorithm}
\begin{algorithm}[h]
Pick $k$ from $\{1, \cdots, \ell = \lceil \lg \frac{1}{\epsilon} \rceil + 2\}$ uniformly at random\\
Call \texttt{Sampling\_attempt(k)}, if it fails, go back to line 1, otherwise return the result

\caption{Sample an edge from distribution pointwise $\epsilon$-close to uniform} \label{alg:sample_edge}
\end{algorithm}

We show separately for light and heavy edges that they are sampled almost uniformly. The case of light edges (\Cref{obs:light_uniform}) is analogous to the proof in \cite{Eden2018}. We include it here for completeness.

\begin{observation} \label{obs:light_uniform}
For $k$ chosen uniformly from $[\ell]$, any fixed light edge $e=uv$ is chosen by \Cref{alg:sampling_attempt} with probability $\frac{1}{\ell n \theta}$
\end{observation}
\begin{proof}
The edge $uv$ is chosen exactly when $k=1$ (this happens with probability $\frac{1}{\ell}$), $u_0 = u$ (happens with probability $\frac{1}{n}$), and $j$ is such that $v$ is the $j$-th neighbor of $u$ (happens with probability $\frac{1}{\theta}$). This gives total probability of $\frac{1}{\ell n \theta}$.
\end{proof}

We now analyze the case of heavy edges. Before that, we define for $v$ being a heavy vertex $h_{v,1} = \frac{d_H(v)}{d(v)}$ and for $i \geq 2$
\[
h_{v,i}= h_{v,1} \sum_{w \in n_H(v)} h_{w,i-1} / d_H(v)
\]
For light vertices, the $h$-values are not defined.

\begin{lemma} \label{lem:heavy_prob}
For $k$ chosen uniformly from $[\ell]$, for any heavy edge $vw$
\[
P(u_{k-1} = v, u_k = w | k \geq 2) = (1- h_{v,\ell-1})\frac{1}{(\ell-1) n \theta}
\]
\end{lemma}

\begin{proof}
Let $k$ be chosen uniformly at random from $\{2,\cdots,r\}$. We show by induction on $r$ that $P(u_{k-1} = v | r) = (1- h_{v,r-1})\frac{d(v)}{(r-1) n \theta}$ for any heavy vertex $v$. If we show this, the lemma follows by substituting $r = \ell$ and by uniformity of $u_k$ on the neighborhood of $u_{k-1}$.

For $r=2$, the claim holds because when $k=2$, there is probability $\frac{1}{n \theta}$ that we come to $v$ from any of the $(1-h_{v,1}) d(v)$ adjacent light vertices. 

We now show the induction step. In the following calculation, we denote by $P_r(\mathcal{E})$ the probability of event $\mathcal{E}$ when $k$ is chosen uniformly from $[r]$. Consider some vertex $v$ and take a vertex $w \in N(v)$. It now holds $P_r(u_{k-2} = w) = P_{r-1}(u_{k-1} = w)$. We have

\ifconference
\else
\begin{align}
P_r(u_{k-1} = v) =& \sum_{w \in N(v)} P_r(u_{k-2} = w) P(u_{k-1}=v | u_{k-2}=w) \\
=& \sum_{w \in N_L(v)} P(u_0 = w) P(u_1 = v | u_0 = w) P(k = 2) \\&+ \sum_{w \in N_H(v)} P_{r-1}(u_{k-1} = w | k > 1) P(u_{k} = v | u_{k-1} = w) P(k > 2) \\
=& \sum_{w \in N_L(v)} \frac{1}{n} \frac{1}{\theta} \frac{1}{r-1} + \sum_{w \in N_H(v)} (1- h_{w,r-2})\frac{d(w)}{(r-2) n \theta} \frac{1}{d(w)} \frac{r-2}{r-1} \\
=& (1-h_{v,2}) \frac{d(v)}{(r-1) n \theta} + \sum_{w \in N_H(v)} (1- h_{w,r-2})\frac{1}{(r-1) n \theta} \\
=& (1-h_{v,1}) \frac{d(v)}{(r-1) n \theta} + h_{v,1} d(v) \Big(1- \sum_{w \in N_H(v)}h_{w,r-2}/d_H(v)\Big)\frac{1}{(r-1) n \theta} \\
=& (1-h_{v,1}) \frac{d(v)}{(r-1) n \theta} + (h_{v,1} - h_{v,r-1})\frac{d(v)}{(r-1) n \theta} \\
=& (1 - h_{v,r})\frac{d(v)}{(r-1) n \theta}
\end{align}
\fi
\jakub{maybe explain more some equalities}
\end{proof}

\noindent
Before putting it all together, we will need the following bound on $h_{v,i}$.
\begin{lemma} \label{lem:h_bound}
For any $v \in V_H(G)$ and $k \geq 1$ it holds that
\[
h_{v,k} \leq 2^{-k}
\]
\end{lemma}
\begin{proof}
We first prove that for any $v \in V(G)$, it holds that $h_{v,1} \leq 1/2$. This has been shown in \cite{Eden2018} and we include this for completeness. We then argue by induction that this implies the lemma.

Since $v$ is heavy, it has more than $\theta$ neighbors. Moreover, there can be at most $\frac{m}{\theta}$ heavy vertices, meaning that the fraction of heavy neighbors of $v$ can be bounded as follows
\[
h_{v,1} \leq \frac{m}{\theta} \frac{1}{\theta} = \frac{m}{\lceil\sqrt{2m}\rceil} \frac{1}{\lceil\sqrt{2m}\rceil} \leq \frac{1}{2}
\]

We now show the claim by induction. We have shown the base case and it therefore remains to prove the induction step:
\[
h_{v,i}= h_{v,1} \sum_{w \in n_H(v)} h_{w,i-1} / d_H(v) \leq \frac{1}{2} \sum_{w \in n_H(v)} 2^{-(i-1)} / d_H(v) = 2^{-i}
\]
\end{proof}

\noindent
We can now prove the following theorem
\begin{theorem} \label{thm:sampling_one_edge}
For $\epsilon \leq \frac{1}{2}$, the \Cref{alg:sample_edge} runs in expected time $O(\frac{n}{\sqrt{m}} \log \frac{1}{\epsilon})$ and samples an edge from a distribution that is pointwise $\epsilon$-close to uniform.
\end{theorem}
\begin{proof}
We first prove the correctness and then focus on the time complexity.

\paragraph{Correctness.}
We first show that in an iteration of \Cref{alg:sample_edge}, each edge is sampled with probability in $[(1- \epsilon)\frac{1}{\ell n \theta}, \frac{1}{\ell n \theta}]$. For light edges, this is true by \Cref{obs:light_uniform}. We now prove the same for heavy edges.
Similarly, a heavy edge $(v,w)$ is chosen when $k \geq 2$, $u_{k-1} = v$ and $w = u_k$. Using \Cref{lem:heavy_prob},
\begin{align}
P(k \geq 2, u_{k-1} = v, u_k = w) &= P(k \geq 2)P(u_{k-1} = v, u_k = w | k \geq 2) \\
&= \frac{\ell - 1}{\ell} (1- h_{v,\ell-1})\frac{1}{(\ell-1) n \theta} \\
&= (1- h_{v,\ell-1})\frac{1}{\ell n \theta}
\intertext{We can now use \Cref{lem:h_bound} to get a lower bound of }
&\geq (1- 2^{-\ell+1})\frac{1}{\ell n \theta} \geq (1-\tfrac{1}{2}\epsilon) \frac{1}{\ell n \theta}
\end{align}
Similarly, since the value $h_{v,\ell}$ is always non-negative, it holds
\begin{align}
P(k \geq 2, u_{k-1} = v, u_k = w) \leq \frac{1}{\ell n \theta}
\end{align}

Consider one execution of \Cref{alg:sampling_attempt} with $k$ chosen uniformly from $[\ell]$ and let $\mathcal{S}$ denote the event that the execution does not end with failure. Let $e$ be the sampled edge and $e', e''$ some fixed edges. Then since $P(e = e' | \mathcal{S}) = \frac{P(e = e')}{P(\mathcal{S})}$ and for any fixed $e'$ it holds that
\[
(1-\tfrac{1}{2}\epsilon) \frac{1}{\ell n \theta} \leq P(e' = e) \leq \frac{1}{\ell n \theta}
\]
it follows that
\[
1-\tfrac{1}{2}\epsilon \leq \frac{P(e = e')}{P(e = e'')} \leq (1-\tfrac{1}{2}\epsilon)^{-1} \leq 1+\epsilon
\]
where the last inequality holds because $\epsilon \leq \frac{1}{2}$. \Cref{alg:sampling_attempt} perform sampling attempts until one succeeds. This means that the returned edge comes from the distribution conditional on $\mathcal{S}$. As we have noted, this scales the sampling probabilities of all edges by the same factor of $P(\mathcal{S})$ and the output distribution is, therefore, pointwise $\epsilon$-close to uniform.

\paragraph{Time complexity.}
Consider again one execution of \Cref{alg:sampling_attempt}. Since for every fixed edge $e'$, the probability that $e = e'$ is at least $(1-\epsilon) \frac{1}{\ell n \theta}$, the total success probability is
\[
P(\mathcal{S}) = P(\bigvee e = e') = \sum_{e' \in E} P(e = e') \geq m (1-\epsilon) \frac{1}{\ell n \theta}
\]
where the second equality holds by disjointness of the events. The expected number of calls of \Cref{alg:sampling_attempt} is then
\[
\frac{\ell n \theta}{(1-\epsilon) m} = \frac{\sqrt{2} (\lceil \lg \frac{1}{\epsilon} \rceil+1) n}{(1-\epsilon) \sqrt{m}} = O\left(\frac{n}{\sqrt{m}} \log \frac{1}{\epsilon}\right)
\]

Each call of \Cref{alg:sampling_attempt} takes in expectation $O(1)$ time because in each step of the random walk, we abort with probability at least $1/2$ (we are on a heavy vertex as otherwise we would have aborted, in order not to abort, the next vertex also must be heavy; the fraction of neighbors that are heavy is $h_{v,1} \leq 1/2$). Therefore, the complexity is as claimed.
\end{proof}

\ifconference
\else
\subsection{Biased vertex sampling using hash-ordered access}
\fi
We now describe a sampling procedure (\Cref{alg:biased_sampling}) which allows us to sample vertices such that vertices with high degree are sampled with higher probability. We will then use this for sampling multiple edges and \ifconference \else later in \Cref{sec:counting}\fi for approximate edge counting.
The procedure is broken up into two parts -- one for pre-processing and one for the sampling itself. 
We do not make it explicit in the pseudocode how the data structures built during pre-processing are passed around for the sake of brevity.

In the rest of this section, let $p_V = \frac{2\log n + \log \delta^{-1}}{\theta}$. $p_N$ is a parameter that determines the sampling probabilities; see \Cref{lem:biased_sampling} for the exact role this parameter plays.\jakub{should I re-name $p_V, p_N$?}
The algorithm works as follows. We make $\lg n$ vertex samples (line 2) which we call $S_k$'s. We will then have a priority queue for each $S_k$ called $\mathcal{Q}_k$ which allows us to iterate over $N(S_k)$ in the order of increasing value of $h(v)$. This could be done by inserting all vertices in $N(S_k)$ into the priority queue. We, however, use a more efficient way based on the hash-ordered neighbor access.
We start with a priority queue that has for each vertex $v$ from $S_k$ its first neighbor (as that is the one with the lowest hash value) represented as $(v,1)$ with the priority equal to its hash. Whenever we use the pop operation, popping the $i$-th neighbor of $v \in S_k$ represented $(v,i)$, we insert into $\mathcal{Q}_k$ the next neighbor of $v$, represented as $(v,i+1)$ with priority equal to its hash. This allows us to access $N(S_k)$ in the order of $h(v)$.

To be able to perform multiple independent runs of the biased vertex sampling algorithm, we will have to resample the hash value from an appropriate distribution for any vertex that the algorithm processes. We will call these resampled hash values virtual and will denote them $h'(v)$. At the beginning, $h'(v) = h(v)$ for all vertices $v$. When a vertex $v$ from $N(S_k)$ is processed, we put it, represented as $(v,``virtual")$, back into the priority queue $\mathcal{Q}_k$ with priority equal to the virtual hash. Each priority queue $\mathcal{Q}_k$ therefore contains vertices of the form $(v,i)$, representing $v[i]$, which have not been used by the algorithm yet and vertices of the form $(v,``virtual")$ that have been processed already but they have been re-inserted into the priority queue with a new priority $h'(v)$.

Note that that a vertex may be present in the priority queues multiple times (for example, if $v$ is the $i$-th neighbor of $u$ and $j$-th neighbour of $w$, then $v$ may be once inserted as $(u,i)$ and once as $(w,j)$ and once as $(v, ``virtual")$). We make sure that all copies of one vertex always have the same priority, namely $h'(v)$. We assume in the algorithm that the hash values of all vertices are different (this happens with probability 1). One may also use $(h'(v), id(v))$ as the priority of the vertex $v$ (with comparisons performed lexicographically) to make the algorithm also work on the event when two vertices have the same hash. This may be useful if implementing the algorithm in practice.

In the pre-processing phase, we initialize the priority queues $\mathcal{Q}_1, \cdots, \mathcal{Q}_{\lg n}$. These will be then used in subsequent calls of the biased vertex sampling algorithm. As we mentioned, we do not make it explicit how they are passed around (they can be thus though of as global variables).

In what follows, we let $v(a) = w[i]$ for for $a = (w,i)$ and $v(a) = w$ for $a = (w, ``virtual")$. We assume that the variables $T_k$ are sets (as opposed to multisets).

\begin{algorithm}[h]
\For{$k \in \{0, \cdots, \lg n\}$}{
	$S_k \leftarrow$ sample $n p_V/2^k = \frac{n  \log (2n/\delta)}{2^k \theta}$ vertices with replacement \label{line:sample_s}\\
	If $d(S_k) \geq \frac{4 m  \log (2n/\delta)}{2^k \theta}$, go back to line 2\\
	\For{$v \in S_k$}{ \label{line:iterate_over_s}
		Insert into $\mathcal{Q}_k$ the tuple $(v, 1)$ with priority $h'(v[1])$
	}
}

\caption{Biased vertex sampling -- preprocessing algorithm, given a tradeoff parameter $\theta$} \label{alg:biased_sampling_preproc}
\end{algorithm}

\begin{algorithm}[h]
Let $\ell$ be such that this is the $\ell$-th execution of the algorithm.\\
\For{$k \in \{0, \cdots, \log n\}$}{
	$T_{k} \leftarrow \emptyset$\\
	\While{either $h'(v(a)) \leq 1-(1-p_N 2^k)^{\ell}$ or $p_N 2^k \geq 1$ for $a \leftarrow \mathcal{Q}_k.top()$}{
	    \While{$v(b) = v(a)$ for $b \leftarrow \mathcal{Q}_k.pop()$}{ \label{line:take_all_copies_of_a_vertex}
		    \If{$b$ is of the form $(v,i)$}{
		    	Replace $(v,i)$ by $(v,i+1)$ in $\mathcal{Q}_k$, set priority to $h'(v[i+1])$\\
		    }
		}
		   \If{$d(v(a)) \not \in [2^{k} \theta, 2^{k+1} \theta)$}{
		   	Skip $a$, continue with next iteration of the loop
		   }
		$T_{k} \leftarrow T_{k} \cup \{v(a)\}$\\
		$h'(v(a)) \leftarrow Unif([1-(1-p_N 2^k)^{\ell}, 1])$ or $1$ if $1-(1-p_N 2^k)^{\ell} > 1$ \label{line:resample}\\
		Add to $\mathcal{Q}_k$ an item $(v(a), ``virtual")$ with priority $h'(v(a))$\\
	}
}
\Return{$T_{0} \cup \cdots \cup T_{\lg n}$}

\caption{Biased vertex sampling -- sampling algorithm, given parameters $\theta,p_N$} \label{alg:biased_sampling}
\end{algorithm}

\begin{remark}
In practice, implementing \Cref{alg:biased_sampling} poses the problem that the virtual hash values $h'(v)$ may be quickly converging to 1 as $\ell$ increases, making it necessary to use many bits to store them. This can be circumvented as follows. Whenever $1-(1-p_N 2^k)^{\ell} \geq 1/2$ during the execution of the algorithm, we access the whole neighborhood of the vertex, allowing us to resample the hashes (by setting virtual hashes) from scratch. In other words, when we have seen in expectation half of the vertices adjacent to a vertex, we get all of them, allowing us to resample them from the same distribution they had at the beginning. (Using this approach would require a minor modification to \Cref{alg:biased_sampling}.)
\end{remark}

In what follows, let $h'_\ell(v)$ be the virtual hash at the beginning of the $\ell$-th execution of \Cref{alg:biased_sampling} and let $k_v = \lfloor \lg(d(v)/\theta) \rfloor$.
\begin{lemma} \label{lem:uniform_virtual_hashes}
Condition on each heavy vertex $v$ having a neighbor in $S_{k_v}$. At the beginning of the $\ell$-th execution of \Cref{alg:biased_sampling}, for all heavy $v$ with $p_N 2^{k_v} <1$, $h'(v)$ is distributed uniformly on $[1-(1-p_N 2^{k_v})^{\ell-1},1]$ and 
the events $h'_\ell(v) \leq 1-(1-p_N 2^{k_v})^\ell$ for $v \in V$ and $\ell\in \mathbb{Z}^+$ are  jointly independent.
Moreover, any $a$ such that $v(a) = u$ in $\mathcal{Q}_k$ has priority $h'(u)$.
\end{lemma}
\ifconference
\begin{proof}[(Proof sketch.)]
\else
\begin{proof}
\fi
We first focus on the distribution of the virtual hashes; we prove that the priorities are equal to the hashes afterwards. The proof of both parts is by induction on $\ell$.

We define $X_{v,\ell}$ to be the indicator for $h'_\ell(v) \leq 1-(1-p_N 2^{k_v})^\ell$. We now focus on one vertex $v$ and drop it in the subscript. We prove by induction that, conditioned on $X_1, \cdots X_{\ell-1}$, the distribution of $h_\ell(v)$ is uniform on $1-(1-p_N 2^{k_v})^{\ell-1}$. This implies that the (unconditional) distribution of $h'_\ell(v)$ is as claimed. We will also use this below to prove independence.
The distribution of $h_1(v)$ is as claimed as the virtual hash values are initially equal to the original (non-virtual) hash values which are assumed to be uniformly distributed on $[0,1]$ and we are conditioning on an empty set of random variables. 
Consider $h_{\ell-1}(v)$ conditioned on $X_1, \cdots, X_{\ell-2}$. The distribution is uniform on $[1-(1-p_N 2^{k_v})^{\ell-2},1]$ by the inductive hypothesis. Conditioning on $X_{\ell-1} = 1$, we resample $h'(v)$ uniformly from $[1-(1-p_N 2^{k_v})^{\ell-1},1]$. Conditioning on $X_{\ell-1} = 0$, is equivalent to conditioning on $h'(v) > 1-(1-p_N 2^{k_v})^{\ell-1}$. This conditional distribution is uniform on $[1-(1-p_N 2^{k_v})^{\ell-1},1]$. Either way, the distribution is as claimed. This proves the inductive step.




\ifconference
\else
We now argue independence across vertices. The algorithm has the property that the virtual hash value of one vertex does not affect the values of other vertices (note that when processing vertex $v$, all conditions in the algorithm only depend on $h'(v)$ and independent randomness).
This implies that $X_{u,1}, X_{u,2}, \cdots$ and $X_{v,1}, X_{v,2}, \cdots$ are independent for $u \neq v$.
This together with what we have shown above implies joint independence.
To prove this formally, consider some finite subset $S \subseteq V \times \mathbb{Z}^+$ and let us have $x_a \in \{0,1\}$ for each $a \in S$. Let $V(S) = \{v | \exists \ell \in \mathbb{Z}^+, (v,\ell) \in S\}$
and let $\prec$ be arbitrary total ordering on $V(S)$.
Let $\mathbb{Z}_v(S) = \{\ell | (v,\ell) \in S\}$.
Consider now $P(\bigwedge_{a \in S} X_a = x_a)$. We can re-write this as
\begin{align}
\prod_{u \in V(S)} P(\bigwedge_{\substack{(v,\ell) \in S\\ v = u}} X_{v,\ell} = x_{v,\ell} | \{X_{w,\ell'}\}_{(w,\ell') \in S, w \prec u}) \ifconference \\ \fi =\prod_{u \in V(S)} P(\bigwedge_{\substack{(v,\ell) \in S\\ v = u}} X_{v,\ell} = x_{v,\ell})
\end{align}
where the equality holds by the independence of $X_{u,1}, X_{u,2}, \cdots$ and $X_{v,1}, X_{v,2}, \cdots$.
We can further rewrite
\ifconference
\begin{align}
P(\bigwedge_{\substack{(v,\ell) \in S\\ v = u}} X_{v,\ell} = x_{v,\ell}) &= \prod_{\ell \in \mathbb{Z}_v(S)} P(X_{u,\ell} = x_{u,\ell} | \{X_{u, \ell'}\}_{\ell' \in \mathbb{Z}_{u}, \ell' < \ell})  \\ &= \prod_{\ell \in \mathbb{Z}_v(S)} P(X_{u,\ell} = x_{u,\ell})
\end{align}
\else
\begin{align}
P(\bigwedge_{\substack{(v,\ell) \in S\\ v = u}} X_{v,\ell} = x_{v,\ell}) = \prod_{\ell \in \mathbb{Z}_v(S)} P(X_{u,\ell} = x_{u,\ell} | \{X_{u, \ell'}\}_{\ell' \in \mathbb{Z}_{u}, \ell' < \ell}) = \prod_{\ell \in \mathbb{Z}_v(S)} P(X_{u,\ell} = x_{u,\ell})
\end{align}
\fi
where the second equality holds because (as we have shown above) the virtual hash $h'_\ell(v)$ is independent of $X_{v,1}, X_{v,2}, \cdots, X_{v,\ell-1}$.
Putting this together, we have
\[
P(\bigwedge_{a \in S} X_a = x_a) = \prod_{a \in S} P(X_a = x_a)
\]
which means that the random variables $X_{v,\ell}$ (and thus the events $h_\ell(v) \leq 1-(1-p_N 2^{k_v})^\ell$) are jointly independent.
\fi

We now argue that the priorities are equal to the virtual hashes. Whenever a vertex is added to $\mathcal{Q}_k$, its priority is equal to its virtual hash. The only way it may happen that the priorities and virtual hashes are not equal is that the virtual hash of some vertex $u$ changes while there exists $a \in \mathcal{Q}_k$ such that $v(a) = u$. This never happens as the virtual hash of $v(a)$ changes only on \cref{line:resample} after all $a \in \mathcal{Q}_k$ such that $v(a)=u$ have been removed. Note that we are using the fact that different vertices have different virtual hashes (which we assume without loss of generality as we discussed above), which ensures that all $a$ with $v(a) = u$ are removed on \cref{line:take_all_copies_of_a_vertex}.
\ifconference
\end{proof}
\else
\end{proof}
\fi

\begin{lemma} \label{lem:biased_sampling}
Let us have integer parameters $\theta,t \geq 1$.  When executed $t$ times, \Cref{alg:biased_sampling} returns samples $T_1, \cdots, T_t$. Assume \Cref{alg:biased_sampling} is given the priority queues $\{\mathcal{Q}_k\}_{k=1}^{\lg n}$ produced by \Cref{alg:biased_sampling_preproc}.
Then there is an event $\mathcal{E}$ with probability at least $1-\delta$ such that, conditioning on this event, for any $i\in [t]$ and any vertex $v$ such that $d(v) \geq \theta$, it holds that 

{
\[
P(v \in T_i) = \min(1,p_N 2^{\lfloor\lg \frac{d(v)}{\theta}\rfloor}) \in [\min(1,\frac{d(v)}{2\theta}p_N), \min(1,\frac{d(v)}{\theta} p_N)]
\]
}
and, conditionally on $\mathcal{E}$, the events $\{e \in T_i\}_{e \in E, i \in [t]}$ are jointly independent. \Cref{alg:biased_sampling_preproc} has expected time complexity $O(\frac{n \log n \log (n/\delta)}{\theta})$. \Cref{alg:biased_sampling} has expected \ifconference \else total \fi time complexity $O\big(t\big(1+ \frac{p_N m \log^2 n \log (n/\delta)}{\theta}\big)\big)$.
\end{lemma}

\begin{proof}
We start by specifying the event $\mathcal{E}$ and bounding its probability. The probability that a vertex of degree at least $2^k \theta$ does not have a neighbor in $S_k$ after one sampling of $S_k$ (that is, not considering the repetitions) on \cref{line:sample_s}
is at most
\[
(1-\frac{2^k \theta}{n})^{\frac{n  \log( 2n/\delta)}{2^k \theta}} \leq \exp\big(-\log( 2n/\delta)\big) = \frac{\delta}{2n}
\]
Therefore, taking the union bound over all vertices, the probability that there exists an integer $k$ and a vertex with degree at between $2^k \theta$ and $2^{k+1} \theta$ that does not have at least one neighbour in $S_k$  after sampling $S_k$ on \cref{line:sample_s}, is at most $\delta/2$. We now bound the probability that this holds for some $S_k$ on line \cref{line:iterate_over_s}.

There are $n p_V/2^k = \frac{n  \log (2n/\delta)}{2^k \theta}$ vertices sampled on line 2 of \Cref{alg:biased_sampling_preproc}. The expected size of the neighborhood of a vertex picked uniformly at random is $\frac{2m}{n}$. The expectation of $d(S)$ is then $\frac{2 m  \log (2n/\delta)}{2^k \theta}$. By the Markov's inequality, the probability that $d(S) \geq \frac{4 m  \log (2n/\delta)}{2^k \theta}$ is at most $1/2$. We repeatedly sample $S_k$ until it satisfies $d(S) \leq \frac{4 m  \log (2n/\delta)}{2^k \theta}$. Its distribution is thus the same as if we conditioned on this being the case. It can be easily checked (by the Bayes theorem) that conditioned on the event $d(S) < \frac{4 n \log(2n/\delta)}{2^k \theta}$, any vertex with degree between $2^k \theta$ and $2^{k+1} \theta$ has at least one neighbour in $S_k$, with probability at least $1-\delta$. Therefore, on line 4, it holds that with probability at least $1-\delta$, there does not exist an integer $k$ and a vertex $v$ such that $2^k \leq d(v) \leq 2^{k+1}$ and $v$ has no neighbor in $S_k$. We call $\mathcal{E}$ the event that this is the case. In the rest of this proof, we condition on $\mathcal{E}$.

\ifconference\else \bigskip \fi \noindent
\emph{We now prove correctness.} Consider a vertex $v$ and let again $k = \lfloor \lg(d(v) / \theta) \rfloor$. Consider the case $p_N 2^k \geq 1$. Conditioned on $\mathcal{E}$, one of its neighbours is in $S_k$. Since $p_N 2^k \geq 1$, the condition on line 4 is satisfied in the $k$-th iteration of the loop on line 2. Therefore, the whole neighborhood of $S_k$ is added to $T_{k}$ and the returned sample thus contains $v$.

We now consider the case $p_N 2^k < 1$. Conditioned on $\mathcal{E}$, a vertex $v$ is included in $T_{k}$ when $2^{k} \theta \leq d(v) < 2^{k+1} \theta$ and $h'(v) \leq 1-(1-p_N 2^k)^{\ell}$. Since $h'(v)$ is uniformly distributed in $[1-(1-p_N 2^k)^{\ell-1},1]$, this happens with probability 
\[
\frac{1-(1-p_N 2^k)^{\ell} - (1-(1-p_N 2^k)^{\ell-1})}{1-(1-(1-p_N 2^k)^{\ell-1})} = p_N 2^k
\]
Let $h'_\ell(v)$ and $k_v$ be defined as in the statement of \Cref{lem:uniform_virtual_hashes}. We know from \ifconference\else that lemma \fi that the events $\{h'_\ell(v) \leq 1-(1-p_N2^{k_v})^\ell\}_{v \in V, \ell \in [t]}$ are jointly independent, conditioned on $\mathcal{E}$. $h'_\ell(v) \leq 1-(1-p_N2^{k_v})^\ell$ is equivalent to $v \in T_\ell$. This implies that the events $\{v \in T_\ell\}_{v \in V, \ell \in [t]}$ are also jointly independent, as we set out to prove.

\ifconference\else \bigskip \fi \noindent
\emph{We now prove the claimed query complexity.} We first show the complexity of \Cref{alg:biased_sampling_preproc}. As we argued, the probability of resampling $S_k$ because the condition on line $3$ is satisfied, is at most $1/2$. Therefore, the time spent sampling the set $S_k$ (including the repetitions) is $O(|S_k|)$.
For every $k$, $|S_k| = \frac{n \log( 2n/\delta)}{2^k \theta}$. Therefore, the total size of the sets $S_k$ is upper-bounded by
\[
\sum_{k=0}^\infty\frac{n \log( 2n/\delta)}{2^k \theta} = \frac{2 n \log (2n/\delta)}{\theta}
\]
Since all values that are to be inserted into $\mathcal{Q}_k$ in \Cref{alg:biased_sampling_preproc} are known in advance, we can build the priority queues in time linear with their size. This means that the preprocessing phase (\Cref{alg:biased_sampling_preproc}) takes $O(\frac{n  \log (n/\delta)}{\theta})$ time. 

We now focus on \Cref{alg:biased_sampling}. We now prove that an execution of the loop on line $2$ of \Cref{alg:biased_sampling} takes in expectation $O(1 + \frac{p_N m \log n \log (n/\delta)}{\theta})$ time for any $k \in \{0, \cdots, \log n\}$, from which the desired bound follows. Specifically, we prove that the number of executions of lines 6-8 is $O(\frac{p_N m \log (n/\delta)}{\theta})$ from which this bound follows as every iteration takes $O(\log n)$ time (as the time complexity of an iteration is dominated by the operations of the priority queue).

Lines 6-8 are executed once for each item $a$ in $\mathcal{Q}_k$ such that the priority of $a$ is $\leq 1-(1-p_N 2^k)^\ell$ if $p_N 2^k < 1$ or when $p_N2^k \geq 1$. As we have argued, this happens with probability $\min(1,p_N 2^k) \leq p_N 2^k$. There are $\frac{n \log( 2n/\delta)}{2^k \theta}$ vertices in $S_k$. Since these vertices are chosen at random, there is in expectation $\frac{2m \log( 2n/\delta)}{2^k \theta}$ incident edges. We consider an item for such incident edge with probability $\leq p_N 2^k$. This means that we consider on the mentioned lines in expectation $O(\frac{p_N m \log( 2n/\delta)}{\theta})$ edges, as we wanted to prove.
\end{proof}

\ifconference
\else
\subsection{Bernoulli sampling with hash-ordered neighbor access} \label{sec:sampling_with_bernoulli}
\fi
We now show how to sample each edge independently with some fixed probability $p$ in the hash-ordered neighbor access model. Our approach works by separately sampling light and heavy edges, then taking union of the samples.
In fact, we solve a more general problem of making $t$ Bernoulli samples with time complexity \emph{sublinear} in $t$ for some range of parameters. We will need this for sampling edges with replacement.
We first give an algorithm to sample edges, assuming we can sample separately light and heavy edges.

\begin{algorithm}[h]
\If{$p > 0.9$}{
    Perform the sampling by a standard linear-time algorithm in time $O(n+t p m)$.
}

$\theta \leftarrow \sqrt{\frac{\log (n) \log (n/\delta)}{p t}}$\\
\ifconference\else \bigskip \fi
Run \Cref{alg:biased_sampling_preproc} with parameter $\theta$ to prepare data structures for \Cref{alg:biased_sampling} (used within \Cref{alg:sampling_with_heavy})\\
\For{$i$ from $1$ to $t$}{
    $S_{L,i} \leftarrow$ sample each light edge with probability $p$ using \Cref{alg:sampling_with_light} with parameter $\theta$ \label{line:sample_light}\\
    $S_{H,i} \leftarrow$ sample each heavy edge with probability $p$ using \Cref{alg:sampling_with_heavy} with parameter $\theta$ \label{line:sample_heavy}\\
}
\Return{$(S_{L,1} \cup S_{H,1}, \cdots, S_{L,t} \cup S_{H,t})$}

\caption{Make $t$ Bernoulli samples from $E$ with inclusion probability $p$} \label{alg:sampling_with_bernoulli}
\end{algorithm}
\begin{theorem} \label{thm:sampling_with_bernoulli}
Given a parameter $t$, with probability at least $1-\delta$, \Cref{alg:sampling_with_bernoulli} returns $t$ samples $T_1, \cdots, T_t$. Each edge is included in $T_i$ with probability $p$. Furthemore, the events $\{e \in T_i\}_{e \in E, i \in [t]}$ are jointly independent. The expected time complexity is $O\big(\sqrt{t p}n\sqrt{\log (n)\log (n/\delta)} + t \big(1 +p m \log^2 n \log (n/\delta)\big)\big)$.
\end{theorem}
\begin{proof}
If $p > 0.9$, we read the entire graph in $O(n+m)$ and compute the sample in time $O(t m) = O(t p m)$. This is less than the claimed complexity. In the rest, we assume that $p \leq 0.9$.

Correctness follows from \cref{lem:sampling_with_light,lem:sampling_with_heavy} which imply that light and heavy edges, respectively, are separately sampled independently with the right probability. Moreover, there is no dependency between the samples of the light and heavy edges, as the sample of light edges does not depend on the hashes. The same lemmas give running times of $O(t p n \theta)$ and $O(t p m \log^2 n \log n/\delta)$ spent on \cref{line:sample_light,line:sample_heavy} respectively. \Cref{alg:biased_sampling_preproc} takes $\frac{n \log n \log (n/\delta)}{\theta}$ time. 
Substituting for $\theta$, the expected time complexity is as claimed.
\end{proof}

\subsubsection*{Sampling light edges} \label{sec:sampling_with_light}
Now we show how to sample from the set of light edges such that each light edge is sampled independently with some specified probability $p$.
\begin{algorithm}[h] 
$k \sim Bin(n\theta, p)$\\
$T \leftarrow \emptyset$\\
$M \leftarrow \emptyset$\\
\RepTimes{$k$}{
	$v \leftarrow$ pick vertex uniformly at random\\
	$\ell \leftarrow$ random number from $[\theta]$\\
	\If{$(v, \ell) \in M$}{
	    Go to line 5.
	}
	If $v$ is light and $d(v) \geq \ell$, add the $\ell$-th edge incident to $v$ to $T$\\
	$M \leftarrow M \cup \{(v,\ell)\}$\\
}
\Return{$T$}

\caption{Sample each light edge independently with probability $p$}\label{alg:sampling_with_light}
\end{algorithm}

As essentially the same algorithm already appeared in \cite{Eden2018}, we defer the proof of the following lemma to the full version.
\begin{lemma} \label{lem:sampling_with_light}
\Cref{alg:sampling_with_light} samples each light edge independently with probability $p \leq 0.9$ and its expected query complexity is $O(p n \theta)$.
\end{lemma}
\ifconference \else
\begin{proof}
Consider the set $B$ of pairs $(v,\ell)$ where $v \in V$ and $\ell \in [\theta]$. We sample $k\sim Bin(n \theta,p)$ such pairs without replacement. By the choice of $k$ and the conditioning principle, it holds that each pair $(v,\ell)$ has been sampled with probability $p$, independently of other pairs. Each light edge has exactly one corresponding pair in $B$, namely a light edge $uv$ where $v$ is the $i$-th neighbor of $u$ corresponds to $(u,i)$. The algorithm returns all light edges whose corresponding pair was sampled. This happened for each pair independently with probability $p$, thus implying the desired distribution of $T$.

By the assumption $p \leq 0.9$, we have $P(k \leq 0.95 n\theta) \geq 1-\exp(-\Omega(n \theta))$. On this event, sampling each edge (that is, sampling a pair $(v,\ell)$ that is not in $M$) takes in expectation $O(1)$ queries. Moreover $E(k) = p n \theta$. Therefore, it takes in expectation $O(p n \theta)$ queries to sample the $k$ pairs (combining the expectations using the Wald's equation). It always takes $O(n \theta \log (n \theta))$ time to sample $k$ edges due to coupon collector bounds, so the event $k > 0.95 n \theta$ only contributes $o(1)$ to the expectation since $P(k > 0.95n)$ is exponentially small.
\end{proof}
\fi

\subsubsection*{Sampling heavy edges}\label{sec:sampling_with_heavy}
We now show an algorithm for Bernoulli sampling from the set of heavy edges. The algorithm is based on \Cref{alg:biased_sampling}.
\begin{algorithm}
$S \leftarrow$ Use \Cref{alg:biased_sampling} with $p_N = \min(1,2 \theta p)$\\
$S' \leftarrow$ heavy vertices from $S$\\
$T \leftarrow \emptyset$\\
\For{$v \in S'$}{
	With probability $\frac{1-(1-p)^{d(v)}}{\min(1,p_N 2^{\lfloor\lg \frac{d(v)}{\theta}\rfloor})}$, skip $v$ and continue on line 3\\
	$k \sim (Bin(d(v), p) | k \geq 1)$\\
	Sample $k$ edges incident to $v$ without replacement, add them to $T$\\
}
\Return{$T$}
\caption{Sample each heavy edge independently with probability $p$, given parameter $\theta$}\label{alg:sampling_with_heavy}
\end{algorithm}
\begin{lemma} \label{lem:sampling_with_heavy}
Assume that \Cref{alg:sampling_with_heavy} is given $\{\mathcal{Q}_k\}_{k=1}^{\lg n}$ as set by \Cref{alg:biased_sampling_preproc}. With probability at least $1-\delta$, \Cref{alg:sampling_with_heavy} samples each heavy edge independently with probability $p$. Moreover, when executed multiple times, the outputs are independent. It has time complexity $O( p m \log^2 n  \log n/\delta)$ with high probability.
\end{lemma}
\begin{proof}
We first show that the probability on line 4 is between $0$ and $1$ (otherwise, the algorithm would not be valid). It is clearly non-negative, so it remains to show \ifconference \else that \fi $1-(1-p)^{d(v)} \leq \min(1,p_N 2^{\lfloor\lg \frac{d(v)}{2\theta}\rfloor})$. If $1 \leq p_N = 2 \theta p$ , then $1 \leq p_N 2^{\lfloor\lg \frac{d(v)}{2\theta}\rfloor}$, and the inequality then clearly holds. Otherwise, 
\[
p_N 2^{\lfloor\lg \frac{d(v)}{2\theta}\rfloor} \geq 2 \theta p 2^{\lg(\frac{d(v)}{\theta}) -1} \geq p d(v) \geq 1-(1-p)^{d(v)}
\]

\ifconference\else \bigskip \fi \noindent
\emph{We now show correctness.} That is, we show that edges are sampled independently with the desired probabilities. By \Cref{lem:biased_sampling}, for heavy $v$, it holds that $P(v \in S | \mathcal{E}) = \min(1,p_N 2^{\lfloor\lg \frac{d(v)}{2\theta}\rfloor})$. This means that the probability of $v$ being in $S$ and not being skipped is 
\[
\min(1,p_N 2^{\lfloor\lg \frac{d(v)}{2\theta}\rfloor}) \frac{1-(1-p)^{d(v)}}{\min(1,p_N 2^{\lfloor\lg \frac{d(v)}{\theta}\rfloor})} = 1-(1-p)^{d(v)}
\]

This is equal to the probability of $X \geq 1$ for $X \sim Bin(d(v), p)$. Since $k$ is picked from $(Bin(d(v), p) | k\geq 1)$, the distribution of number of edges incident to $v$ that the algorithm picks is distributed as $Bin(d(v),p)$. Let $X_v$ be the number of edges incident to $v$ that are picked. Since each vertex $v$ is picked independently by \Cref{lem:biased_sampling}, we have that these random variables are independent and $X_v \sim Bin(d(v),p)$. Consider the vector $(X_{v_1}, \cdots, X_{v_n})$. Consider experiment where we pick each edge independently with probability $p$ (this is the desired distribution) and let $Y_v$ be the number of edges incident to $v$ that are picked. Then $(Y_{v_1}, \cdots, Y_{v_n}) \sim (X_{v_1}, \cdots, X_{v_n})$. 
By the conditioning principle, we get that each (directed) edge is sampled independently with probability $p$.

Assume \Cref{alg:sampling_with_heavy} is executed $t$ times, outputting \ifconference \else sets \fi $H_1, \cdots, H_t$. By \Cref{lem:biased_sampling}, we know that the events $\{e \in H_i\}_{e \in E, i \in [t]}$ are jointly independent. \Cref{alg:sampling_with_heavy} only depends on the virtual hashes in the calls of \Cref{alg:biased_sampling} and the rest only depends on independent randomness. The samples $H_1, \cdots, H_t$ are thus independent. 

\ifconference\else \bigskip \fi \noindent
\emph{We now argue the time complexity.} \Cref{alg:biased_sampling} has expected time complexity $O(\frac{p_N m \log^2 n \log (n/\delta)}{\theta}) = O(p m \log^2 n \log (n/\delta))$. We now prove this dominates the complexity of the algorithm. The rest of the algorithm has time complexity linear in $|S| + |T|$. The complexity of \Cref{alg:biased_sampling} clearly dominates $|S|$ (as $S$ is the output of this algorithm). It holds $E(|T|) = p m$, so the expected size of $T$ is also dominated by the complexity of \Cref{alg:biased_sampling}. This completes the proof.  \looseness=-1
\end{proof}

\ifconference \else
\ifconference
\subsection{Sampling Edges Without Replacement with Hash-ordered Neighbor Access}
\else
\subsection{Sampling edges without replacement with hash-ordered neighbor access}
\fi
We now show how Bernoulli sampling can be used to sample vertices without replacement.
\begin{algorithm}[h]
$p \leftarrow 1/n^2$\\
$S \leftarrow$ sample each edges with probability $\min(1,p)$ using \Cref{alg:sampling_with_bernoulli} with failure probability $\frac{\delta}{3\lg n}$\\
\If{$|S| < s$}{
	$p \leftarrow 2 p$\\
	Repeat from line $1$\\
}
\Return{random subset of $S$ of size $s$}

\caption{Sample without replacement $s$ edges} \label{alg:sampling_with_no_replacement}
\end{algorithm}

\begin{theorem}\label{thm:sampling_with_no_replacement}
\Cref{alg:sampling_with_no_replacement} samples $s$ edges uniformly without replacement with probability at least $1-\delta$. Moreover, \Cref{alg:sampling_with_no_replacement} has expected query complexity $O(\sqrt{s}\frac{n \log(n) \log(n/\delta)}{\sqrt{m}} + s \log^2 n \log (n/\delta))$.
\end{theorem}
\begin{proof}
In each iteration, \Cref{alg:sampling_with_bernoulli} fails with probability at most $\frac{\delta}{3 \lg n}$. There are at most $\lg n^2$ iterations in which $p < 1$. After this number of iterations, each additional iteration can happen only if \Cref{alg:sampling_with_bernoulli} fails. This happens with probability $< 1/2$, so we get in expectation $< 2$ additional iterations. Putting this together by the Wald's equation\footnote{We use the Wald's equation on the indicators that in the $i$-th iteration fails. The bound is, in fact, on the expected number of failures, which is an upper bound on the probability of failure (by the Markov's inequality).}, we see that the failure probability is at most $(2 \lg + 2) \frac{\delta}{3 \lg n} < \delta$. In the rest of the analysis, we condition on no errors happening in any of the calls of \Cref{alg:sampling_with_bernoulli}.

Condition on the $\ell$-th iteration being the first to succeed. What is the conditional distribution of the sample? Note that repeating an experiment until the outcome satisfies some property $\phi$ results in the outcome of the experiment being distributed as if we conditioned on $\phi$.
Consider additionally conditioning on $|S| = k$. Then, by symmetry, $S$ is a sample without replacement of size $|k|$. If $k \geq s$, it holds that taking a sample without replacement of size $k$ and taking a random subset of size $s$, we get a sample with distribution of a sample without replacement of size $s$.
Since this distribution is the same for all values of $k \geq s$, the distribution is unchanged if we only condition on the union of the events $|S| = k$ for $k \geq s$ or, in other words, if we condition on $|S| \geq s$. Similarly, the distribution is the same independently of the value of $\ell$. Therefore, the unconditioned distribution is the same.
This proves that the algorithm gives a sample from the right distribution.
\jakub{this paragraph needs attention}.

Since the probability $p$ increases exponentially, so does the expected time complexity of each iteration. Therefore, the time complexity is dominated by the complexity of the last iteration. Consider an iteration with $p < \frac{2s}{m}$. Then the time complexity of this iteration is no greater than the desired bound. Consider now the case $p \geq \frac{2s}{m}$. By the Chernoff bound, the probability that $|S| < s$ is then at most $1/3$. The probability of performing each additional iteration thus decreases exponentially with base $1/3$. The expected time complexity of an iteration, on the other hand, increases exponentially with base $2$. This means that the expected time complexity contributed by each additional iteration decreases exponentially. The asymptotic time complexity is therefore equal to that of the iteration with $\frac{s}{m} \leq p < \frac{2s}{m}$. The complexity of an iteration is dominated by line 2. Therefore, by \Cref{thm:sampling_with_bernoulli}, the time complexity is as claimed.
\end{proof}
\fi

\ifconference
\else
\subsection{Sampling edges with replacement with hash-ordered neighbor access}
\fi

The algorithm for Bernoulli sampling can be used to sample multiple edges with replacement. The proof of the theorem below appears in the full version. We now sketch intuition of correctness of the algorithm. Suppose $S_1$ is non-empty. Picking from each non-empty $S_i$ one edge at random gives us a sample without replacement of size equal to the number of non-empty $S_i$'s. Since the number of $S_i$'s is $2s$, if we pick $p$ large enough, then with high probability at least $s$ of them will be non-empty. In that case, we have a sample without replacement of size $\geq s$ and taking a random subset of size $s$ gives a sample with the desired distribution.
\begin{algorithm}[h] 
$p \leftarrow 1/n^2$\\
$S_1, \cdots, S_{2s} \leftarrow$ Bernoulli samples with $\min(1,p)$ using \Cref{alg:sampling_with_bernoulli} with failure probability $\frac{\delta}{3\lg n}$\\
$T \leftarrow$ from each non-empty $S_i$ pick a random edge\\
\If{$|T| < s$}{
    $p \leftarrow 2p$\\
    Go to line 2\\
}
\Return{random subset of $T$ of size $s$}
\caption{Sample $s$ edges with replacement}\label{alg:sampling_with_with_replacement}
\end{algorithm}

\begin{theorem}\label{thm:sampling_with_with_replacement}
Given $s \leq n$, with probability at least $1-\delta$, \Cref{alg:sampling_with_with_replacement} returns $s$ edges sampled with replacement. Moreover, the algorithm runs in time $O(\sqrt{s} \frac{n}{\sqrt{m}} \log n \log n/\delta + s \log^2 n \log (n/\delta))$ with high probability.
\end{theorem}
\ifconference \else
\begin{proof}
By exactly the same argument as in \Cref{thm:sampling_with_no_replacement}, with probability at least $1-\delta$, there are no errors in the calls of \Cref{alg:sampling_with_bernoulli}. In the rest of the proof, we condition on this being the case.

We now argue that the returned sample has the correct distribution. The argument is again very similar to that in \Cref{thm:sampling_with_no_replacement}. Consider one of the $S_i$'s. If we condition on $|S_i| = k$, then by symmetry, the distribution is that of sampling $k$ edges without replacement. Picking at random one of those edges (assuming $k \geq 1$), we get one edge uniformly at random. Since this distribution is independent of $k$ for $k \geq 1$, this is also the distribution we get if we only condition on $k \geq 1$. Therefore, $T$ has a distribution of $|T|$ edges picked uniformly at random with replacement at random, where $|T|$ is a random variable. The effect on the distribution of $T$ of repeating the sampling until $|T| \geq s$ is the same as conditioning on $|T| \geq s$. Taking a random subset of size $s$, it has a distribution of $s$ edges being sampled with replacement.

The time complexity of the algorithm is dominated by line $2$. The expected time complexity of each iteration increases exponentially with base $2$ as this is the rate at which $p$ increases.
The probability that $S_i = \emptyset$ is $(1-p)^{m} \leq e^{-p m}$. Consider the case $p \geq 2/m$. Then the probability that $|T| < s$ can be upper-bounded by the Chernoff inequality by $1/3$. Therefore, after $p \geq 2/m$, the probability of each additional iteration decreases exponentially with base $3$. Therefore, the expected time complexity contributed by each additional iteration then decreases exponentially. The expected time complexity is thus dominated by the first iteration in which $p \geq 2/m$. The expected complexity is thus as claimed by \Cref{thm:sampling_with_bernoulli}. 
\end{proof}
\fi

\ifconference \else
\ifconference
\subsection{Implementing Our Algorithms with Batched Access} \label{sec:implement_with_api}
\else
\subsection{Implementing our algorithms with batched access} \label{sec:implement_with_api}
\fi
%
Let $d$ denote the average degree rounded up and consider the following setting. Suppose we have access to the following queries: (1) random vertex query and (2) query that, given a vertex $v$ and an index $i$, returns neighbors $(i -1)d+1, \cdots, i \, d$ of $v$ (where the neighborhoods are assumed to be ordered arbitrarily). As we argued in \Cref{sec:model}, this setting is relevant in many practical situations, such as when accessing a graph through a (commonly used) API interface or when the graph is stored on a hard drive.

Consider sampling a vertex uniformly and then querying it to learn its whole neighborhood. In expectation, this takes $\leq 2$ queries\footnote{It is $2$ and not $1$ due to rounding. For example, if $O(1)$ vertices have degree $0$ and the rest have degree just above $d$, then $\approx 2$ queries will be needed on average.} as the expected neighborhood size is $d$ and we may get $\geq d$ neighbors in one query. \Cref{alg:sampling_with_light}, has the property that it only accesses the neighborhoods of uniformly random vertices. In \Cref{alg:sampling_with_heavy}, we access neighborhoods on line 7 and within calls of \Cref{alg:biased_sampling_preproc,alg:biased_sampling} which also have this property. On line 7, we only access in expectation $O(p m)$ neighbors, adding a cost of at most $O(p m)$ queries, thus not increasing the complexity.
For the other vertices whose neighborhoods we will be accessing, we may learn their entire neighborhood and simulate the hash-ordered neighbor access at a cost of $O(1)$ queries per vertex. This allows us to implement \Cref{alg:sampling_with_bernoulli,alg:sampling_with_no_replacement,alg:sampling_with_with_replacement,alg:counting_by_sampling,alg:counting_with_direct} in this model without increasing their asymptotic query complexity.

\ifconference
\subsection{Sampling Multiple Edges without Hash-ordered Neighbor Access} \label{sec:sampling_without}
\else
\subsection{Sampling multiple edges without hash-ordered neighbor access} \label{sec:sampling_without}
\fi
We now show an algorithm that does not use the hash function, at the cost of a slightly worse running time. The only place where we have used the hash function in the above algorithms is when sampling heavy edges, specifically in \Cref{alg:biased_sampling}. We show how to simulate \Cref{alg:biased_sampling} in the indexed neighbor access model. We can then use this to get Bernoulli sampling as well as sampling with and without replacement.
We actually show how to simulate any algorithm from the hash-ordered neighbor access model in the indexed neighbor access model. We then analyze the running time of the simulation of \Cref{alg:biased_sampling}. Our algorithm improves upon the state of the art when $\epsilon = \tilde{o}(\sqrt{\frac{n}{m}})$.

\begin{theorem} \label{thm:sampling_without}
There are algorithms that, with probability at least $1-\delta$, return sample (1) of edges such that each such edge is sampled independently with probability $\frac{s}{m}$ (where $s$ does not have to be an integer), (2) of $s$ edges sampled without replacement, or (3) of $s$ edges sampled with replacement. Assuming $m = \Omega(n)$, these algorithms run in expected time $O(\sqrt{s n} \log n/\delta + s \log^2 n \log (n/\delta))$ with high probability.
\end{theorem}
\begin{proof}
We show a general way to simulate hash-ordered neighbor access. The first time the algorithm wants to use a neighborhood query on some vertex $v$, we look at the whole neighborhood of $v$, generate virtual hashes $h'(u)$ for all $u \in N(v)$ for which it has not been generated yet, and sort $N(v)$ with respect to the virtual hash values. This clearly allows us to run any algorithm from the hash-ordered neighbor access model in the indexed neighbor access model.

In \Cref{alg:biased_sampling}, we access neighborhoods of the vertices sampled in \Cref{alg:biased_sampling_preproc}, of total size at most $O(\frac{m \log n}{\theta})$. The time complexity is, therefore, by the same argument as in \cref{lem:biased_sampling,lem:sampling_with_light,lem:sampling_with_heavy}, at most $O(\frac{m \log n/\delta}{\theta} + s n \theta / m + s \log^2 n \log (n/\delta))$. We now set $\theta = \text{median}(1,\tilde{m} \sqrt{\frac{\log (n/\delta)}{n s}},n \sqrt{\log (n/\delta)})$.

We may use, for example, the algorithm from \cite{Goldreich2006} to get $\tilde{m}$ such that it holds $m \leq \tilde{m} \leq 2m$ with probability at least $1-\frac{1}{s n^2}$. We call this event $\mathcal{E}$.
Conditioning on $\mathcal{E}$, $\theta = \Theta(m \sqrt{\frac{\log (n/\delta)}{n s}})$ and the complexity is as claimed.
The time complexity is always $O(m \log n/\delta + s n^2 \log(n/\delta) / m + s \log^2 n \log (n/\delta))$ as $1 \leq \theta \leq n\sqrt{\log(n/\delta)}$. It holds $P(\mathcal{E}^C) \leq \frac{1}{s n^2}$. Therefore, this event contributes only to the expected time complexity only $O(\log^2 n \log(n/\delta)$. Thus, the event $\mathcal{E}^C$ does not increase the asymptotic time complexity.
\end{proof}

\subsection{Lower bound for sampling multiple edges}
As the last result on sampling edges, we prove that \cref{alg:sampling_with_bernoulli,alg:sampling_with_no_replacement,alg:sampling_with_with_replacement} are optimal up to logarithmic factors.
\begin{theorem}
Any algorithm in the hash-ordered neighbor access model that samples pointwise $0.9$-close to uniform (1) each edge independently with probability $\frac{s}{m}$, (2) $s$ edges without replacement, or (3) $s$ edges with replacement, has to use in expectation $\Omega(\sqrt{s}\frac{n}{\sqrt{m}} + s)$ queries.
\end{theorem} 
\begin{proof}
The term $s$ is dominant (up to a constant factor) for $s \gtrsim \frac{n^2}{m}$ and the lower bound holds on this interval as any algorithm that returns $s$ edges has to run in time $\Omega(s)$. Now we consider the case when $s \lesssim \frac{n^2}{m}$.

Let $G$ be a graph consisting of $s$ cliques, each having $m/s$ edges, and the remaining vertices forming an independent set. Due to the assumption on $s$, the total number of vertices used by the cliques is no more than $n$ and this graph, therefore, exists.

Consider the case of sampling $s$ edges pointwise $0.9$-close to uniform at random in either of the three settings. They hit in expectation $\Omega(s)$ distinct cliques. Since the algorithm has to hit each clique from which an edge is sampled, it has to hit in expectation $\Omega(s)$ cliques by uniformly sampling vertices. The probability that a uniformly picked vertex lies in one fixed clique is $O(\frac{s \sqrt{m/s}}{n}) = O(\frac{\sqrt{s m}}{n})$. To hit in expectation $s$ cliques, the number of samples the algorithm has to perform is then $\Omega(s \frac{n}{\sqrt{s m}}) = \Omega(\sqrt{s} \frac{n}{\sqrt{m}})$.
\end{proof}

\section{Estimating the Number of Edges by Sampling} \label{sec:counting_by_sampling}
The Bernoulli sampling from \Cref{sec:sampling_with_bernoulli} allows us to estimate the number of edges efficiently. The idea is that if we sample each edge independently with probability $p$, then the number of sampled edges is concentrated around $p m$, from which we can estimate $m$ (assuming we know $p$).

\begin{algorithm}
$p \leftarrow 1/n^2$\\
$S \leftarrow$ sample each edge with probability $p$ using \Cref{alg:sampling_with_bernoulli}\\
\If{$|S| < \frac{6 (\log \delta^{-1} + \log 8 \lg n)}{\epsilon^2}$}{
    $p \leftarrow 2p$\\
    Go to line 2\\
}
\Return{$|S|/p$}
\caption{Estimate the number of edges by edge sampling} \label{alg:counting_by_sampling}
\end{algorithm}

\begin{lemma} \label{lem:counting_using_sampling}
Given $\epsilon < 1$, \Cref{alg:counting_by_sampling} returns an estimate $\hat{m}$ of $m$  such that $P(|\hat{m} - m| > \epsilon m) < \delta$. It runs in time $O(\frac{n}{\epsilon \sqrt{m}} \log n \log n/\delta + \frac{\log n \log n/\delta}{\epsilon^2})$.
\end{lemma}
\begin{proof}
The proof works as follows. We first show that if $p < \frac{3 (\log \delta^{-1} + \log 8 \lg n)}{\epsilon^2 m}$, then with sufficiently high probability $|S| < \frac{6 (\log \delta^{-1} + \log 8 \lg n)}{\epsilon^2 m}$ and the algorithm will continue. We then show that if $p \geq \frac{3 (\log \delta^{-1} + \log 8 \lg n)}{\epsilon^2 m}$, then $|S|$ is sufficiently concentrated.
Taking the union bound over all iterations, we get that with good probability, in all iterations with $p < \frac{3 (\log \delta^{-1} + \log 8 \lg n)}{\epsilon^2 m}$, it holds that the algorithm continues and for all $p \geq \frac{3 (\log \delta^{-1} + \log 8 \lg n)}{\epsilon^2 m}$, the estimate is $\epsilon$-close to the true number of edges. This implies correctness.

\medskip \noindent
Consider the case $p < \frac{3 (\log \delta^{-1} + \log 8 \lg n)}{\epsilon^2 m}$. By the Chernoff bound,
\[
P(|S| \geq \frac{3 (\log \delta^{-1} + \log 8 \lg n)}{\epsilon^2 m}) 
\leq \exp(-\frac{p m}{3})
< \frac{\delta}{2 \lg n^2}
\]

\medskip \noindent
Consider now the case $p \geq \frac{3(\log \delta^{-1} + \log 8 \lg n)}{\epsilon^2 m}$ by the Chernoff bound,
\[
P(||S| - E(S)| > \epsilon E(|S|)) < 2\exp(-\frac{\epsilon^2 p m}{3}) \leq \frac{\delta}{2 \lg n^2}
\]

\end{proof}

\medskip \noindent
Using \cref{thm:sampling_with_bernoulli,thm:sampling_without} to sample the edges for $X$, we get the following theorem.
\begin{theorem} \label{thm:counting_using_sampling}
There is an algorithm that uses hash-ordered neighbor access and returns a $1+\epsilon$-approximate of the number of edges with probability at least $1-\delta$ in expected time $O(\frac{n}{\epsilon \sqrt{m}} \log n \log n/\delta + \frac{\log n \log n/\delta}{\epsilon^2})$.

There is an algorithm that \emph{does not use} hash-ordered neighbor access and returns a $1+\epsilon$-approximate of the number of edges with probability at least $1-\delta$ in expected time $O(\frac{\sqrt{n}}{\epsilon} \log n \log n/\delta)$.
\end{theorem}
In the algorithm without hash-ordered neighbor access, the second term does not have to be present, as it only becomes dominant when the expected time complexity is $\Omega(n)$, in which case we may use a trivial $O(n)$ algorithm.

\section{Directly Estimating the Number of Edges} \label{sec:counting}

We now give two algorithms for approximate edge counting. The first uses hash-ordered neighbor access and runs in time $\tilde{O}(\frac{n}{\epsilon \sqrt{m}} + \frac{1}{\epsilon^2})$. This is the same complexity \jakub{incl log factors?} as that of \Cref{alg:counting_by_sampling} which approximates the number of edges by Bernoulli sampling. The algorithm we give now is more straight-forward and solves the problem of approximating the number of edges directly. We then give a different algorithm which replaces the need for hash-ordered neighbor access by the more standard pair queries. It has time complexity of $\tilde{O}(\frac{n}{\epsilon \sqrt{m}} + \frac{1}{\epsilon^4})$.
We then show a lower bound of $\Omega(\frac{n}{\epsilon \sqrt{m}})$ for $\epsilon \geq \frac{\sqrt{m}}{n}$. This matches, up to logarithmic factors, the complexity of our algorithms for $\epsilon \geq \frac{\sqrt{m}}{n}$ and $\epsilon \geq \frac{m^{1/6}}{n^{1/3}}$, respectively.

\subsection{Algorithm with hash-ordered neighbor access} \label{sec:counting_with}
We combine our biased sampling procedure with the Horvitz-Thompson estimator. When we appropriately set the parameters of biased sampling, we get an estimator with lower variance than the estimator on which the algorithm in \cite{Seshadhri2015} is built.

It is also possible to simulate this algorithm without hash-ordered neighbor access. This algorithm has the same time complexity as the one of the algorithm from \Cref{thm:counting_using_sampling}. The simulation can be done by the same approach as in \Cref{sec:sampling_without}; we do not repeat the argument.

In our analysis, we assume we have an estimate $\tilde{m}$ such that $m \leq \tilde{m} \leq 2m$. We also prove that even when these inequalities do not hold, the algorithm is unlikely to return an estimate that is more than a constant factor greater than $m$. As a consequence of this guarantee, the advice of $\tilde{m}$ then can be removed by standard techniques. See \Cref{sec:removing_advice} for details, including a sketch of how the advice can be removed.


\begin{algorithm}
$\theta = \min(\epsilon \sqrt{\frac{1}{32}\tilde{m} (\log n + \log 12)}, \frac{1}{64} \epsilon^2 n / \log n)$\\
$S_L \leftarrow$ sample $k = \frac{2n (\log n + \log 12)}{\theta}$ vertices with replacement, keep those with degree $< \theta$\\
$S_H \leftarrow$ sample vertices using \Cref{alg:biased_sampling} with parameters $\theta$, $p_N = \frac{n}{\tilde{m} \log n}$, $t=1$, and $\delta = 1/12$\\
Define $P_v = \min(1, p_N 2^{\lfloor \lg \frac{d(v)}{\theta}\rfloor})$\\
\Return{$X = \frac{n}{2k} \sum_{v \in S_L} d(v)  + \frac{1}{2}\sum_{v \in S_H} \frac{d(v)}{P_v}$}

\caption{Estimate the number of edges in a graph} \label{alg:counting_with_direct}
\end{algorithm}

\begin{theorem} \label{thm:counting_with_direct}
\Cref{alg:counting_with_direct} returns an estimate $\hat{m}$ such that $P(\hat{m} \geq 8m) \leq 1/3$. It has expected time complexity $O(\frac{n \sqrt{\log n}}{\epsilon \sqrt{\tilde{m}}} + \frac{\log^2 n}{\epsilon^2})$. If, moreover, $\tilde{m} \leq m$, then with probability at least $2/3$, it holds $|X - m| \leq \epsilon m$. 
\end{theorem}

\begin{proof}
\emph{We first focus on the time complexity.} Line 2 runs in expected time $O(\frac{n \log n}{\theta})$. On line 3, we spend in expectation $O(\frac{n \log n}{\theta}) = O(\frac{n \sqrt{\log n}}{\epsilon \sqrt{m}} + \frac{\log^2 n}{\epsilon^2})$ time by \Cref{lem:biased_sampling} and our choice of $p_N$. This dominates the complexity of the rest of the algorithm.

\medskip \noindent
\emph{In the rest of the proof, we prove correctness.} We break up the estimator $X$ into three estimators $X_L,X_M,X_H$ such that $X_L+X_M+X_H=X$. We first prove separately bounds on each of the three estimators, and then put it together.

We say a vertex is light if $d(v) < \theta$, medium if $\theta \leq d(v) < \frac{2\theta}{p_N}$ and heavy if $d(v) \geq \frac{2\theta}{p_N}$ (note that this definition is different from that in \Cref{sec:sampling}). The threshold between medium and heavy vertices is set such that heavy vertices are sampled with probability $1$, whereas the probability that a fixed medium vertex is sampled is strictly less than $1$. We call the sets of light, medium, and heavy vertices $V_L,V_M,V_H$, respectively. We define $X_L = \frac{n}{2 k} \sum_{v \in S_L} d(v)$ and $X_\bullet$ for $\bullet \in \{M, H\}$ as $\frac{1}{2} \sum_{v \in S\cap V_\bullet} \frac{d(v)}{P_v}$ ($P_v$ is defined in the algorithm) and $X = X_L + X_M + X_H$.

\paragraph{Light vertices.}
Let $v_i$ be the $i$-th vertex sampled on line $1$. The expectation of $X_L$ is 
\[
E(X_L) = E\Big(\frac{n}{2 k} \sum_{v \in S_L} d(v)\Big) = \frac{n}{2 k} \sum_{i = 1}^{k} E\big(\mathbb{I}(d(v_i) \leq \theta) d(v_i)\big) = \frac{n}{2} E\big(\mathbb{I}(d(v_1) \leq \theta) d(v_1)\big) = \frac{1}{2} d(V_L)
\]
and the variance can be bounded as\footnote{We define $\sup(X)$ as the smallest $x$ such that $P(X > x) = 0$.}
\begin{align}
Var(X_L) &= Var\Big(\frac{n}{2 k} \sum_{i=1}^{k} d(v)\Big) \\ &= \frac{n^2}{4 k} Var\big(\mathbb{I}(d(v_1) \leq \theta) d(v_1)\big) \\ &\leq \frac{n^2}{4 k} \sup\big(\mathbb{I}(d(v_1) \leq \theta)d(v_1)\big) E\big(\mathbb{I}(d(v_1) \leq \theta) d(v_1)\big) \\ &\leq \frac{n^2 \theta d(V_L)}{4 k n}
\end{align}
The first inequality holds because $Var(X) \leq \sup(X) E(X)$ whenever $P(X \geq 0)=1$. The second holds because $\sup\big(\mathbb{I}(d(v_1) \leq \theta)d(v_1)\big) \leq \theta$ and $E\big(\mathbb{I}(d(v_1) \leq \theta) d(v_1)\big) = d(V_L)/\theta$. Because $\theta \leq \epsilon \sqrt{\frac{1}{32} \tilde{m}(\log n + \log 12)}$, $k = \frac{2n(\log n + \log 12)}{\theta}$ and $\tilde{m} \leq 2m$, we get the following upper bound:
\begin{align}
\leq \frac{\theta^2 d(V_L)}{4 (\log n + \log 12)} \leq 
\frac{1}{32} (\epsilon m)^2
\end{align}
where the last inequality holds because $d(V_L) \leq 2m$ and by substituting for the other variables. By the Chebyshev bound, it now holds that
\[
P(|X_L - \frac{1}{2}d(V_L)| \geq \epsilon m/2) \leq \frac{Var(X_L)}{(2 \epsilon m)^2} \leq \frac{1}{8}
\]

\paragraph{Medium vertices.}
For the medium vertices, we consider the conditional expectation and variance, conditioned on $\mathcal{E}$ where $\mathcal{E}$ is the event on which \Cref{alg:biased_sampling} succeeds. The expectation is
\[
E(X_M | \mathcal{E}) = \frac{1}{2} \sum_{v \in V_M} P_v \frac{d(v)}{P_v} = \frac{1}{2} d(V_M)
\]
and the variance is
\begin{align}
Var(X_M | \mathcal{E}) \leq \frac{1}{4}\sum_{v \in V_M}E\Big(\big(\frac{d(v)}{P_v}\big)^2 | \mathcal{E}\Big) \leq \frac{1}{4}\sum_{v \in V_M} P_v \big( \frac{d(v)}{P_v} \big)^2 \leq \frac{1}{2}\sum_{v \in V_M} \frac{d(v) \theta}{p_N}
\leq \frac{m \theta}{p_N} \leq \frac{2 m^2 \theta \log n^2}{n} \leq \frac{1}{32} (\epsilon m)^2
\end{align}
where the last inequality holds because $\theta \leq \frac{1}{64} \epsilon^2 n/ \log n$ and the one before that because $p_N = \frac{n}{\tilde{m} \log n} \geq \frac{n}{2 m \log n}$.
By the Chebyshev bound, it now holds that
\[
P(|X_M - m_M| \geq \epsilon m/2 | \mathcal{E}) \leq \frac{Var(X_M)}{(\epsilon m/2)^2} \leq \frac{1}{8}
\]

\paragraph{Heavy vertices.}
Since the heavy vertices are, conditioned on $\mathcal{E}$, sampled with probability 1, it is the case that $P(X_H = m_H | \mathcal{E}) = 1$.

\paragraph{Putting it all together.}
We first prove $P(\hat{m} \geq 8m) \leq 1/3$. By the Markov's inequality, we have that $P(X_L \geq 8 E(X_L)) \leq 1/8$ and $P(X_M + X_H \geq 8 E(X_M + X_H | \mathcal{E}) | \mathcal{E}) \leq 1/8$. We now have
\begin{align}
P(\hat{m} \geq 8 m) &= P(X_L + X_M + X_H \geq 8(E(X_L) + E(X_M + X_H|\mathcal{E}))) \\&\leq P(X_L \geq 8E(X_L)) + P(X_M + X_H \geq 8E(X_M + X_H|\mathcal{E}) | \mathcal{E}) + P(\mathcal{E}^C) \\&\leq \frac{1}{8} + \frac{1}{8} + \frac{1}{12} = \frac{1}{3}
\end{align}

\noindent
Finally, we prove the concentration. We use the union bound and the bounds on $X_L, X_M, X_H$ that we have proven above.
\begin{align}
P(|X - m| \geq \epsilon m) \leq& P(|X_L - m_L| \geq \epsilon m/2) + P(|X_M - m_M| \geq \epsilon m/2 \vee X_H \neq m_H) \\
\leq& P(|X_L - m_L| \geq \epsilon m/2) + P(|X_M - m_M| \geq \epsilon m/2 \vee X_H \neq m_H | \mathcal{E}) + P(\mathcal{E}^C) \\
\leq& P(|X_L - m_L| \geq \epsilon m/2) + P(|X_M - m_M| \geq \epsilon m/2 | \mathcal{E}) + ( X_H \neq m_H | \mathcal{E}) + P(\mathcal{E}^C) \\
\leq& \frac{1}{8} + \frac{1}{8} + 0 + \frac{1}{12} = \frac{1}{3}
\end{align}
\end{proof}

\noindent
By using \Cref{fact:advice_removal}, we may remove the need for advice, giving us
\begin{corollary} \label{cor:counting_with_direct}
There is an algorithm that runs in expected time $O(\frac{n \sqrt{\log n}}{\epsilon \sqrt{m}} + \frac{\log^2 n}{\epsilon^2})$ and returns $\hat{m}$ such that with probability at least $2/3$, it holds $|X - m| \leq \epsilon m$. 
\end{corollary}

\subsection{Algorithm with pair queries} \label{sec:counting_pairs}
We define $u \prec v$ if $d(u) < d(v)$ or $d(u) = d(v)$ and $id(u) \leq id(v)$. We consider an orientation of edges such that $uv$ is oriented from $u$ to $v$ such that $u \prec v$. We use $d^+(v)$ to denote the out-degree of $v$.

In this section, we use a notion of light and heavy vertices that is slightly different from the one we have used above. We divide the vertices into a set of \emph{light} $V_L$ and a set of \emph{heavy} vertices $V_H$. We assume that for every $v \in V_H, d^+(v) \geq \theta/2$ and for every $v \in V_L, d^+(v) \leq 2\theta$. The vertices with $d^+(v)$ between $\theta/2$ and $2\theta$ can be assigned to either $V_L$ or $V_H$ (but not both). 

We now describe two algorithms we use to classify vertices between being light or heavy. We need this classification to be consistent. When we first decide whether a vertex is light or heavy, we store this decision and use it if the same vertex is later queried. Assume we are given access to independent Bernoulli trials with bias $p$. The algorithm of \citet{Lipton1993} give $\hat{p}$ such that $P((1-\epsilon) p \leq \hat{p} \leq (1+\epsilon) p) \geq 1-\delta$ while running in time $O(\frac{\log \delta^{-1}}{p \epsilon^2})$ \footnote{In that paper, the authors in fact solve a more general problem. For presentation of this specific case, see \cite{Watanabe2005}}. Using this algorithm, we can classify $v$ with high probability in time $O(\frac{d(v) \log n}{d^+(v)})$. Alternatively, using standard Chernoff bounds, one may classify vertices w.h.p.\ in time $O(\frac{d(v) \log n}{\theta})$. We assume that, with probability $1-1/n$, all classifications are correct.

\begin{algorithm}
\DontPrintSemicolon

$m_\uparrow = \tilde{m}/2$\\
$m_\downarrow = 2\tilde{m}$\\
$\theta \leftarrow \epsilon \sqrt{m_\downarrow}$\\
$\tau \leftarrow \sqrt{m_\downarrow/(8\epsilon)}$\\

\ifconference \medskip \else \bigskip \fi
$A_1 \leftarrow 0$\\
\RepTimes{$k = 432 \frac{\theta  n}{\epsilon^2 m_{\uparrow}}$}{ \label{line:basic_algorithm}
    $v \leftarrow$ random vertex\\
    $w \leftarrow$ random neighbor of $v$\\ 
    \If{$v \prec w$ and $v$ is light (use the algorithm from \cite{Lipton1993} to classify vertices as light/heavy)}{
        $A_1 \leftarrow A_1+d(v)$
    }
}
$\hat{d}^+_L \leftarrow \frac{n A_1}{k}$\\

\ifconference \medskip \else \bigskip \fi
$S \leftarrow$ sample $\frac{48 n \log n}{\theta}$ vertices with replacement \tcp*{Note that $S,S'$ are multisets, not sets}
$S' \leftarrow$ vertices of $S$ with degree $\geq \theta$\\
\If{$|S'| > \frac{576 m_\downarrow \log n}{\theta^2}$ or $d(S') > \frac{1152 m_\downarrow \log n}{\theta}$}{ \label{line:many_sampled_vertices}
    \Return{``failure"} \label{line:fail}
}

\For{$i$ from $1$ to $k_2 = \frac{468}{\epsilon^2}$}{ \label{line:second_loop}
    $T \leftarrow$ Sample each edge $\vec{uv}$ incident to $S'$ with probability $p = \frac{\theta}{m_\uparrow}$\\
    $T' \leftarrow$ Set of all vertices $v$ such that $uv \in T$, $d(v) \leq \tau$ and $v$ is heavy (use the standard Chernoff-bound-based algorithm described above to classify vertices as light/heavy)\\
    For each vertex $v \in T'$, let $r(v) = |N(v) \cap S'|$ \tcp*{Compute by using a pair query for each pair $v,w$ for $w \in S'$}
    
    \smallskip
    $A_{2,i} \leftarrow 0$\\
    \For{$v \in T'$}{
        $w \leftarrow$ random neighbor of $v$\\ 
        \If{$v \prec w$}{
            $A_{2,i} \leftarrow A_{2,i}+d(v)/(1-(1-p)^{r(v)})$
        }
    }
}
$\hat{d}^+_H \leftarrow \frac{\sum_{i=1}^{k_2} A_{2,i}}{k_2}$\\

\Return{$\hat{d}^+_L + \hat{d}^+_h$}

\caption{Approximately count edges of $G$ given advice $\tilde{m}$} \label{alg:count_in_standard_model}
\end{algorithm}

We will again use the algorithm described in \Cref{sec:removing_advice} to remove the need for advice $\tilde{m}$. In the following theorem, we prove two deviation bounds. The second one is the one that will give us the approximation guarantee of the final algorithm, while the first one will allow us to remove the need for advice.

We are assuming in the algorithm that conditions in \textbf{if} statements are evaluated in order and the evaluation is stopped when the result is already known (e.g., in $\mathbf{if}(\phi \wedge \psi)$, if $\phi$ evaluates to false, $\psi$ would not be evaluated).

\begin{theorem}  \label{thm:counting_with_pairs}
Given $\tilde{m}$ and $\epsilon>0$, \Cref{alg:count_in_standard_model} returns $\hat{m}$ such that $P(\hat{m} > 7 m) \leq 1/3$ and runs in time $O(\frac{n \log n}{\epsilon \sqrt{\tilde{m}}} + \frac{\log^2 n}{\epsilon^4})$. If, moreover, $m \leq \tilde{m} \leq 2m$, then with probability at least $2/3$, $\hat{m} \in (1\pm \epsilon) m$.
\end{theorem}
\begin{proof}
We start by proving that when $m \leq \tilde{m} \leq 2m$, then with probability at least $2/3$, $\hat{m} \in (1\pm \epsilon) m$ (``correctness"). We then prove that $P(\hat{m} > 4 m) \leq 1/3$ (``bounds for advice removal"). We then finish by proving the time complexity (``time complexity").

\ifconference \medskip \else \bigskip \fi\noindent
\emph{We now prove correctness.} Let $d^+_l, d^+_h$ be the sum of out-degrees of the light and heavy vertices, respectively. It holds $m = d^+_L + d^+_h$.
Throughout the proof, we condition on all light vertices having out-degree at most $2\theta$ and all heavy vertices having out-degree at least $\theta/2$. As we said above, we are assuming this holds with probability at least $1-1/n$ for all vertices. We call this event $\mathcal{E}$.

We now prove that $\hat{d}^+_l$ is a good estimate of $d^+_l$. Let $A_{1,i}$ be the increment of $A_1$ in the $i$-th execution of the loop on line 6.

\begin{align}
E(A_{1,i}) &= \sum_{u \in V} \mathbb{I}(\text{u is light}) P(u = v) P(w \succ v) d(v) \\ &= \sum_{u \in V} \mathbb{I}(\text{u is light}) \frac{1}{n} \cdot \frac{d^+(v)}{d(v)} d(v) \\ &= \sum_{u \in V} \mathbb{I}(\text{u is light}) \frac{1}{n} d^+(v) = d^+_l/n
\end{align}

\begin{align}
Var(A_{1,i}) \leq E(A_{1,i}^2) &= \sum_{u \in V} \mathbb{I}(\text{u is light}) P(u = v) P(w \succ v) d(v)^2 \\ &= \sum_{u \in V} \mathbb{I}(\text{u is light}) \frac{1}{n} d^+(v) d(v) \\ &\leq \sum_{u \in V} \frac{1}{n} 2 \theta d(v) = \frac{4 \theta m}{n}
\end{align}
Therefore, $\hat{d}^+_l$ is an unbiased estimate of $d^+_l$ (conditioning on the correct classification of all light/heavy vertices). Its variance is $\frac{n^2}{k^2} \cdot k \cdot \frac{4\theta m}{n} \leq \tfrac{1}{108} \epsilon^2 m^2$. By the Chebyshev inequality, we have that
\begin{align}
P(|\hat{d}^+_L - d^+_l| > \epsilon m/3) \leq \frac{\epsilon^2 m^2 / 108}{(\epsilon m/3)^2} \leq 1/12
\end{align}

We now focus on the heavy vertices. There are at most $2m/\theta$ heavy vertices. Each one is sampled into $S$ in expectation 
$\frac{48 \log n}{\theta}$ times. Therefore, there are in expectation at most $\frac{2m}{\theta} \cdot \frac{48 \log n}{\theta} = \frac{96 m \log n}{\theta^2}$ vertices in $S'$. Similarly, because each vertex is sampled in expectation $\frac{48 \log n}{\theta}$ times, it holds that $E(d(S')) \leq E(d(S)) \leq 2m \frac{48 \log n}{\theta} = \frac{96 m \log n}{\theta}$.

The condition on \cref{line:many_sampled_vertices} is set such that the algorithm only fails on \cref{line:fail} when $|S'| > 12E(|S'|)$ or $|d(S')| > 12E(|d(S')|)$. By the Markov's inequality and the union bound, with probability at least $1/6$, neither of these inequalities is satisfied. In the rest of the algorithm, it holds $|S'| \leq \frac{576 m_\downarrow \log n}{\theta^2}$ and $|d(S')| \leq \frac{1152 m_\downarrow \log n}{\theta}$. We will use this when arguing the time complexity. When analyzing correctness, we do not condition on the condition on \cref{line:many_sampled_vertices} not being satisfied.

We now argue that $\hat{d}^+_h$ is a good estimate of $d^+_h$. Specifically, we prove that it holds with probability at least $11/12$ that $\hat{d}^+_H \in d^+_H \pm \tfrac{2}{3} \epsilon m$.  Let $u$ be a heavy vertex with $d(u) \leq \tau$. It holds $d^+(u) \geq \theta/2$.
Consider the value $r(u)$. Since each vertex is sampled into $S$ in expectation $\frac{48 \log n}{\theta}$ times, it holds that $E(r(u)) \geq d^+(u) \frac{48 \log n}{\theta} \geq \theta/2 \cdot \frac{48 \log n}{\theta} = 24 \log n$ for any heavy vertex $u$. It holds with probability at least $1-1/n^2$ that $r(u) \geq \frac{24 d^+(u) \log n}{\theta}$ because by the Chernoff bound, we have that
\begin{align}
P(r(u) < \frac{24 d^+(u) \log n}{\theta}) &\leq P(r(u) < E(r(u))/2) \\ &\leq \exp(-\frac{E(r(u))}{12}) \\ &\leq \exp(- 2 \log n) = 1/n^2
\end{align}
and by the union bound, this inequality holds for all heavy vertices simultaneously with probability at least $1-1/n$. 
We condition on this event in what follows, we call it $\mathcal{E}'$. In fact, we will condition on $S'$, and we assume that this inequality holds for $S'$. We now analyze the conditional expectation $E(A_{2,1} | S')$ (note that the expectation $E(A_{2,i} | S')$ is the same for all $i$ and we may thus focus on $i=1$).

\begin{align}
E(A_{2,1}|S') &= \sum_{u \in V} P(u \in T'|S') P(w \succ v) \frac{d(v)}{(1-(1-p)^{r(v)})}\\
&= \sum_{u \in V} \mathbb{I}(d(u) \leq \tau\text{ and }u\text{ is heavy})(1-(1-p)^{r(v)}) \cdot \frac{d^+(u)}{d(u)} \cdot \frac{d(v)}{(1-(1-p)^{r(v)})}\\
&= \sum_{u \in V}  \mathbb{I}(d(u) \leq \tau\text{ and }u\text{ is heavy}) d^+(u)
\end{align}
This counts all heavy edges whose lower-degree endpoint has degree at most $\tau$.
All the uncounted edges are therefore in the subgraph induced by these high-degree vertices. There are at most $\frac{2 m}{\tau} = \sqrt{\epsilon m/2}$ vertices with degree $> \tau$. This means that there can be at most $\binom{\sqrt{\epsilon m/2}}{2} < \epsilon m /3$ uncounted edges. Therefore, it follows that $|E(A_{2,1}|S') - d^+_h| \leq \epsilon m/3$.

We now analyze the conditional variance. Recall that we are assuming that for $S'$, it holds that all heavy vertices $v$ have $r(v) \geq \frac{24 d^+(v) \log n}{\theta}$.

\begin{align}
Var(A_{2,1}|S') \leq E(A_{2,1}^2|S') &= \sum_{u \in V} P(u \in T'|S') P(w \succ v) \Big(\frac{d(u)}{(1-(1-p)^{r(u)})}\Big)^2\\ 
&\leq \sum_{u \in V} \mathbb{I}(u\text{ is heavy}) (1-(1-p)^{r(u)}) \frac{d^+(u)}{d(u)} \cdot \frac{d(u)^2}{(1-(1-p)^{r(u)})^2}\\ 
&\leq \sum_{u \in V}  \mathbb{I}(u\text{ is heavy}) \frac{d^+(u)}{d(u)} \cdot \frac{d(u)^2}{(1-(1-p)^{\tfrac{d^+(u)}{\theta}})}\\ 
&\leq \sum_{u \in V} \frac{2 d^+(u)d(u)}{\tfrac{d^+(u)}{\theta} \cdot \tfrac{\theta}{m}} \\
&= 2 \sum_{u \in V} d(u) m = 4 m^2
\end{align}
where the second inequality holds because we are conditioning on $r(u) \geq \frac{24 d^+(v) \log n}{\theta} \geq \frac{d^+(v)}{\theta}$ \footnote{We intentionally do not use the tightest possible bound in order to allow us to use the next inequality. It is possible to slightly improve the constants by a more technical analysis.} and the third holds because for $x,y$ such that $x y <1, 1>x>0,y\geq 1$, it holds that $1-(1-x)^y \geq x y /2$ and $p=\theta/m_\uparrow \geq \theta/m$.
On the event $\mathcal{E}'$, the expectation $E(A_{2,1}|S')$ is independent of $S'$. By the law of total variance, $Var(A_{2,1} | \mathcal{E}') = E( Var(A_{2,1} | S', \mathcal{E}') | \mathcal{E}') \leq 4 m^2$.
Therefore, $Var(\hat{d}^+_H | \mathcal{E}') \leq \frac{\epsilon^2}{468} Var(A_{2,1} | \mathcal{E'}) \leq \epsilon^2 m^2 /117$. It now holds by the (conditional) Chebyshev inequality that
\begin{align}
P(|\hat{d}^+_H - E(\hat{d}^+_h)| > \epsilon m/3 | \mathcal{E}') \leq \frac{\epsilon^2 m^2 / 117}{(\epsilon m/3)^2} \leq 1/13
\end{align}
Putting this together with the union bound with probability bounds on the events of failure on \cref{line:fail} (probability $\leq 1/6$), event of $|\hat{d}^+_L - d^+_l| > \epsilon m/3$ (probability $\leq 1/12$) and the events $\mathcal{E}^C, {\mathcal{E}'}^C$ (probability $ \leq 1/n$), we get that with probability at least $2/3$, it holds $|\hat{d}^+_L - d^+_l| \leq \epsilon m/3$, $|\hat{d}^+_H - E(\hat{d}^+_h)| \leq \epsilon m/3$, and $|E(\hat{d}^+_h) - \hat{d}^+_h| \leq \epsilon m/3$. On this event, it holds by the triangle inequality that $|\hat{m} - m| \leq \epsilon m$. This proves correctness.

\ifconference \medskip \else \bigskip \fi \noindent
\emph{We now prove the bounds for advice removal.}
We have shown that $E(d^+_l) = d^+_l$ and $E(d^+_H | \mathcal{E}') \leq d^+_h$. By the Markov's inequality, $P(\hat{d}^+_L | \mathcal{E}' \geq 7d^+_k) \geq 1/7$ and $P(\hat{d}^+_H | \mathcal{E}' \geq 7d^+_h) \geq 1/7$. By the union bound, both hold with probability at least $2/7$. Adding the probability of $\mathcal{E}^C$ and ${\mathcal{E}'}^C$, upper bounded by $1/n$, we get that $P(\hat{m} \geq 7m) \leq 1/3$

\ifconference \medskip \else \bigskip \fi \noindent
\emph{We now prove the claimed time complexity bound.} We first focus on the first part (lines 5 - 13). There are $O(\frac{\theta  n}{\epsilon^2 m}) = O(\frac{n}{\epsilon \sqrt{m}})$ repetitions. It holds $P(v \prec w) = d^+(v)/d(v)$. Determining whether $v$ is light on line 9 takes $O(\frac{d(v) \log n}{d^+(v)})$. However, we only need to determine whether $v$ is light when $v \prec w$, which happens with probability $\frac{d^+(v)}{d(v)}$. This, therefore, takes in expectation $O(\log n)$ time. This dominates the expected cost of an iteration, leading to total expected running time of $O(\frac{n \log n}{\epsilon \sqrt{m}})$.

We now focus on the second part of the algorithm (lines 14 - 31). Computing $S$ clearly takes $O(\frac{n \log n}{\theta})$ time.
We now analyze one iteration of the loop on \cref{line:second_loop}. It holds that $d(S') \leq O(\frac{m \log n}{\theta})$. Each of the incident edges is sampled with probability $\leq \frac{2 \theta}{m}$. The expected size of $T$ is thus $O(\frac{m \log n}{\theta}) \cdot \frac{2 \theta}{m} = O(\log n)$. This is also an upper bound on the size of $T'$ as well as on 
the time it takes to compute $T$. 
In computing $T'$, we classify each endpoint $v$ of an edge $\vec{uv}$ of $T$ with $d(v) \leq \tau$. The running time of this is $O(\frac{d(v) \log n}{\theta}) \leq O(\frac{\log n}{\epsilon^{3/2}})$. Since $|T| = O(\log n)$, it takes $O(\frac{\log^2 n}{\epsilon^{3/2}})$ time to compute $T'$.
To calculate $r(v)$ for all $v$, we must query $|S'| |T'|$ vertex pairs. Since $|S'| \leq O(\frac{m \log n}{\theta^2})$, 
it holds that $E(|S'| |T'|) \leq O(\frac{m \log n}{\theta^2}) E(|T'|) \leq O(\frac{m \log^2 n}{\theta^2}) = O(\frac{\log^2 n }{\epsilon^2})$.
The rest of the iteration has time complexity $O(|T'|)$, thus not increasing the total time complexity. Since there are $O(\frac{1}{\epsilon^2})$ iterations, this gives a bound on the time complexity of the second part of $O(\frac{n \log n}{\theta} + \frac{\log^2 n}{\epsilon^4}) = O(\frac{n \log n}{\epsilon \sqrt{m}} + \frac{\log^2}{\epsilon^4})
$. This is then also the time complexity of the whole algorithm.
%
%
%
\end{proof}

By using the methods described in \Cref{sec:removing_advice}, we get the following:
\begin{corollary} \label{cor:counting_with_pairs}
There is an algorithm that, given $\epsilon > 0$, returns $\tilde{m}$ such that with probability at least $2/3$, $\hat{m} \in (1\pm \epsilon) m$ and has expected time complexity $O(\frac{n \log n}{\epsilon \sqrt{\tilde{m}}} + \frac{\log^2 n}{\epsilon^4})$.
\end{corollary}

\subsection{Lower bound}
In this section, we show a lower bound for approximate edge counting of $\Omega(\frac{n}{\epsilon \sqrt{m}})$ for $\epsilon \gtrsim \frac{\sqrt{m}}{n}$. This improves on this range over the previously known $\Omega(\frac{n}{\sqrt{\epsilon m}})$ of \citet{Goldreich2008}. Similarly to their result, our proof works even for the query complexity in the model where, for each accessed vertex, the algorithm gets the whole connected component of that vertex at unit cost. This is a considerably stronger model than our hash-ordered neighbor access model as well as the full neighborhood access model.

Our improvement comes from the fact that we are using anti-concentration results instead of relying just on the difficulty of hitting once a subset of vertices. Specifically, our proof relies on the following lemma:
\begin{lemma} \label{lem:slud_without_replacement}
Suppose we have a finite set $A$ of tuples $(i,t)$ where $i$ is a unique identifier and $t\in \{0,1\}$. Then any procedure $\mathcal{M}$ that can distinguish with probability at least $2/3$ between the case when of all $(i,t) \in A$, a $\frac{1}{2} - \epsilon$-fraction have $t=0$ and $\frac{1}{2} + \epsilon$-fraction have $t=1$ and the case when a $\frac{1}{2} + \epsilon$-fraction have $t=0$ and $\frac{1}{2} - \epsilon$-fraction have $t=1$, has to use in expectation $\Omega(\min(|A|, \frac{1}{\epsilon^2}))$ uniform samples from $A$.
\end{lemma}
\begin{proof}
Assume the existence of $\mathcal{M}$ that uses $o(\min(|A|, \frac{1}{\epsilon^2}))$ samples. We show that this entails a contradiction.

If we sample $k$ elements without replacement instead of with replacement, it is possible to simulate sampling $k$ elements with replacement, assuming we know $|A|$.  Existence of such $\mathcal{M}$ would imply the existence of algorithm $\mathcal{M}'$ with the same guarantees (distinguishing with probability at least $2/3$ between the cases when $\frac{1}{2} - \epsilon$-fraction have $t=0$ and $\frac{1}{2} + \epsilon$-fraction have $t=1$ and the case when a $\frac{1}{2} + \epsilon$-fraction have $t=0$ and $\frac{1}{2} - \epsilon$-fraction have $t=1$ using $o(\min(|A|, \frac{1}{\epsilon^2}))$ samples) which uses the same sample size but samples without replacement.

By symmetry, there is an optimal algorithm which does not use the knowledge of $i$ of the sampled elements. Therefore, if there is such a procedure $\mathcal{M}'$, there also has to be a procedure $\mathcal{M}''$ with the same guarantees that depends only on the number of sampled elements with $t=1$ and the number of sampled elements with $t=0$. That is, $\mathcal{M}''$ distinguishes the distributions $H_1 \sim HGeom(n,(\frac{1}{2} - \epsilon)n, k)$ and $H_2 \sim HGeom(n,(\frac{1}{2} + \epsilon)n, k)$ with probability at least $2/3$. Specifically, $P(\mathcal{M}''(H_i) = i) \geq 2/3$ for $i \in \{0,1\}$. We now show that such $\mathcal{M}''$ cannot exist, thus proving the lemma.

Let $B_1 \sim Bin(k, \frac{1}{2}- \epsilon)$ and $B_2 \sim Bin(k, \frac{1}{2}+ \epsilon)$. It is known (\cite[Theorem 3.2]{Diaconis2004}) that $\|Bin(k, p) - HGeom(n, p n, k)\|_{TV} \leq (k-1)/(n-1)$. Thus, there exists a coupling of $B_1,B_2$ and $H_1,H_2$ such that for fixed $i \in \{0,1\}$, it holds that $P(B_i \neq H_i) \leq (k-1)/(n-1)$. For $k \leq n/11$, it then holds that $P(B_1 \neq H_1) \leq 1/10$ for $n$ sufficiently large (specifically, larger than $10$). Now
\[
P(\mathcal{M}''(B_i) = i) \geq P(\mathcal{M}''(B_i) = i) - P(B_1 \neq H_1) \geq 2/3 - 1/10 \geq 0.55
\]

It is known (\cite[Lemma 5.1]{Anthony1999}) that $Bin(k, \frac{1}{2} - \epsilon)$ and $Bin(k, \frac{1}{2} + \epsilon)$ cannot be distinguished with probability at least $0.55$ when $k = o(\frac{1}{\epsilon^2})$. This implies that $\mathcal{M}''$ does not exist. Therefore $\mathcal{M}$ cannot use $o(\min(|A|, \frac{1}{\epsilon^2}))$ samples cannot exist (note the first branch of the $\min$ which comes from the assumption that $k \leq n/11$).
\end{proof}
Using this lemma, we can now prove the following theorem. In the proof, we use the notation $\sim f(x)$, for example saying that the graph has $\sim m$ edges. The meaning in this case is that the graph has $m'$ edges such that $m' \sim m$.
\begin{theorem}
For $\epsilon \geq \frac{4 \sqrt{m}}{n}$, any algorithm that with probability at least $2/3$ outputs $\tilde{m}$ such that $\tilde{m} = (1\pm \epsilon) m$, has to use $\Omega(\frac{n}{\epsilon \sqrt{m}})$ samples.
\end{theorem}
\begin{proof}
We construct two graphs $G_1,G_2$, one with $\sim (1-2\epsilon)m$ and the other with $\sim (1+2\epsilon) m$ edges and then show that it is hard to distinguish between them. Define dense chunk $S_d$ as a complete graph with $\beta$ vertices. Similarly, define a sparse chunk $S_s$ as an independent set on $\beta$ vertices.

We now describe the graphs $G_1,G_2$. They both consist of $\alpha$ chunks and an independent set of size $n - \alpha \beta$. The graph $G_1$ has a $(\frac{1}{2}+\epsilon)$-fraction of chunks being sparse and the rest being dense, whereas $G_2$ has a $(\frac{1}{2}- \epsilon)$-fraction of the chunks being sparse and the rest being dense.

Note that this means that the sparse chunks form, together with the vertices which are not part of any chunk, an independent set. Nevertheless, it will be useful to separately consider the sparse chunks and vertices not part of any chunk.

We set $\beta = \epsilon \sqrt{m}$ and $\alpha = \frac{4}{\epsilon^2}$. The graph $G_1$ now has 
\[
(\frac{1}{2}-\epsilon) \alpha {\beta \choose 2} \sim (1-2\epsilon) \frac{(\epsilon \sqrt{m})^2}{\epsilon^2} = (1-2\epsilon)m
\]
edges. By a similar calculation, $G_2$ has $\sim (1+2\epsilon)m$ edges. Note that the number of vetrtices in the chunks is 
\[
\alpha \beta = \epsilon \sqrt{m} \frac{4}{\epsilon^2} = \frac{4\sqrt{m}}{\epsilon} \leq n
\]
where the last inequality holds by the assumption on $\epsilon$. The described graph, therefore, does exist.

We now define graphs $G_1, G_2$, corresponding to the two cases from \Cref{lem:slud_without_replacement}. It then follows from the lemma that any algorithm that can tell apart with probability at least $2/3$ between $G_1$ and $G_2$ has to use $\Omega(\frac{n}{\epsilon \sqrt{m}})$ samples. In fact, we prove a stronger statement, namely that this is the case even if the algorithm receives for each sampled vertex information about which chunk the vertex is from as well as the type of the chunk (or that it does not belong to a chunk).

Assuming the algorithm makes $k$ samples, let $X_i$ for $i \in [k]$ be indicator for whether the $i$-th sample hit one of the chunks. It holds that $E(X_i) = \frac{\alpha \beta}{n}$ for any $i$. It follows from \Cref{lem:slud_without_replacement} that any algorithm satisfying the conditions from the statement has to hit the chunks in expectaton at least $\Omega(\frac{1}{\epsilon^2})$ times. If $k$ is the number of samples, this implies that
\[
E(k)\frac{\alpha \beta}{n} = E(\sum_{i =1}^k X_1) \gtrsim \frac{1}{\epsilon^2}
\]
which, solving for $E(k)$, gives
\[
E(k) \gtrsim \frac{n}{\epsilon \sqrt{m}}
\]
\end{proof}

\section{Triangle counting with full neighborhood access}
We start by giving several definitions that we use in this section. Given an edge $e$, we denote by $t(e)$ the number of triangles that contain $e$. Note that in the full neighborhood model, $t(e)$ may be computed in $O(1)$ queries. Specifically, for $e = uv$, it holds $t(e) = |N(u) \cap N(v)|$.
We now define order of edges $\prec$ so that $e_1 \prec e_2$ iff $t(e_1) < t(e_2)$ or $t(e_1) = t(e_2)$ and $id(e_1) \leq id(e_2)$.
We assign each triangle to its edge that is minimal with respect to $\prec$. Given an edge $e$, we let $t(e)$ be the number of triangles assigned to $e$. Note that, in contrast with $t(e)$, we may not in general compute $t^+(e)$ in $O(1)$ queries. Given a set $S \subseteq V$, 
we define $t_S^+(e)$ to be the number of triangles assigned to $e = uv$ that have non-empty intersection with $S \setminus \{u,v\}$. If we are given the endpoints of an edge $e$ and the set $S$, we may compute $t^+_S(e)$ without making any additional queries.
\subsection{Algorithm with edge sampling}
\begin{algorithm}
$A \leftarrow 0$\\
\RepTimes{$k=138 \frac{m}{\epsilon^2 \tilde{T}^{2/3}}$}{
    $uv = e \leftarrow$ pick an edge uniformly at random\\
    $w \leftarrow$ random vertex from $N(u) \cap N(v)$\\
    \If{$uv \prec uw$ and $uv \prec vw$}{
        $A \leftarrow A + t(e)$\\
    }
}
\Return{$\frac{m A}{k}$}

\caption{Count triangles approximately, given advice $\tilde{T}$} \label{alg:sampling_edges_with_advice}
\end{algorithm}

We now prove a bound which we will later need in order to bound variance. Essentially the same bound has been proven in \cite{Kallaugher2019}. We give a slightly different proof, which we later modify to prove \Cref{lem:my_inequality2}.

\begin{lemma} \label{lem:my_inequality}
It holds that
\[
\sum_{e \in E} t^+(e) t(e) \leq 46T^{4/3}
\]
\end{lemma}
\begin{proof}
We start by arguing several inequalities which we will use to bound the variance. We pick a permutation $\pi$ of the vertices such that $t(e_{\pi(1)}) \geq t(e_{\pi(2)}) \geq \cdots \geq t(e_{\pi(m)})$. Let us have $i \in \{2^{k-1}+1 , \cdots, 2^{k}\}$, it holds $(a)$ that  $t(e_{\pi(i)}) \leq \frac{3 T}{2^{k-1}}$. Otherwise, the first $i$ edges in the $\pi$-ordering would have total of $i \frac{3T}{2^{k-1}} > 3T$ edge-triangle incidencies, which is in contradiction with $T$ being the number of triangles. We now bound $\sum_{i=2^{k-1}+1}^{2^k} t^+(e_{\pi(i)})$. It clearly holds $(b)$ that $\sum_{i=2^{k-1}+1}^{2^k} t^+(e_{\pi(i)}) \leq T$. For any triangle $e_{\pi(i)}e_{\pi(j)}e_{\pi(\ell)}$ assigned to $e_{\pi(i)}$, it holds that $i > j$,$i>\ell$. Therefore, any triangle assigned to an edge $e_{\pi(i)}$ is formed by edges in $\{e_{\pi(j)}\}_{j=1}^{2^k}$. On $2^k$ edges, there can be at most $\sqrt{2} \, 2^{3k/2}$ triangles\footnote{The argument is as follows. Consider a graph on $m$ edges. Order vertices in the order of decreasing degrees. Each vertex has to its left $d^+(v) \leq \sqrt{2m}$ of its neighbors: this is clearly the case for the first $\sqrt{2m}$ vertices; it is also the case for any other vertices as otherwise the graph would necessarily have $>m$ edges. Then for $T\leq \sum_v d^+(v)^2 \leq \sqrt{2m} \sum_v d^+(v) = \sqrt{2} m^{3/2}$.}. Therefore, we have $(c)$ that $\sum_{i=2^{k-1}+1}^{2^k} t^+(e_{\pi(i)}) \leq \sqrt{2} \, 2^{3k/2}$.

In what follows, we use (in a slight abuse of notation) the convention $t(e_{\pi(i)}) = t^+(e_{\pi(i)}) = 0$ for $i > m$. We may now use the above bounds to prove the desired bound:
\begin{align}
\sum_{e \in E} t^+(e) t(e)
&= \sum_{i=1}^m t^+(e_{\pi(i)}) t(e_{\pi(i)})
\\&= \sum_{k=1}^{\lceil\log_2 m \rceil} \sum_{i=2^{k-1}+1}^{2^k} t^+(e_{\pi(i)}) t(e_{\pi(i)}) \tag{1}
\\ \tag{2} &\leq 6\sum_{k=1}^{\lceil\log_2 m \rceil} \frac{T}{2^k}\sum_{i=2^{k-1}+1}^{2^k} t^+(e_{\pi(i)}) 
\\ \tag{3} & \leq 6\sum_{k=1}^{\lceil\log_2 m\rceil} \frac{T}{2^k} \min(\sqrt{2} \, 2^{3k/2},\sqrt{2}T) 
\\ &\leq \sqrt{2} \, 6\Bigg(\sum_{k=1}^{\lfloor\frac{2}{3}\log_2 T\rfloor} 2^{k/2} T +  \sum_{k=\lceil\frac{2}{3}\log_2 T\rceil}^{\infty}\frac{T^2}{2^k}\Bigg)
\\&\leq \sqrt{2}\,6\Big( (2+\sqrt{2}) \cdot 2^{\frac{1}{3} \log_2 T}T + 2\frac{T^2}{2^{\frac{2}{3} \log_2 T}}\Big)
\\&< 46 T^{4/3}  
\end{align}
where $(1)$ is using the convention $t(e_{\pi(i)}) = t^+(e_{\pi(i)}) = 0$ for $i > m$, $(2)$ holds by inequality $(a)$ and $(3)$ holds by inequalities $(b)$ and $(c)$. Note that in the second branch of $\min$, we are using the upper-bound $\sqrt{2} T$ instead of $T$ for convenience.
\end{proof}

\begin{lemma} \label{lem:tc_edge_sampling}
Given $\tilde{T}$, \Cref{alg:sampling_edges_with_advice} returns $\hat{T}$ which is an unbiased estimator of $T$. It has query complexity $O(\frac{m}{\epsilon^2 \tilde{T}^{2/3}})$. If, moreover, $\tilde{T} \leq T$, then with probability at least $2/3$, it holds $\hat{T} \in (1\pm\epsilon)T$.
\end{lemma}
\begin{proof}
The query complexity is clearly as claimed. We now argue unbiasedness and the deviation bounds.
Let $A_i$ be the increment in $A$ in the $i$-th iteration of the loop. We now compute $E(A_1)$ and upper-bound $Var(A_1)$.
It holds
\begin{align} 
E(A_1 | e) = P(uv \prec uw \, \wedge \, uv \prec vw) t(e) = \frac{t^+(e)}{t(e)} t(e) = t^+(e)
\end{align} 
By the law of total expectation, $E(A_1) = E(t^+(e)) = T/m$ as each triangle is assigned to one edge, and the average number of triangles assigned per edge is thus $T/m$. We now analyze the variance. 
\begin{align} 
Var(A_1) \leq E(A_1^2) &= \frac{1}{m}\sum_{e \in E} \frac{t^+(e)}{t(e)} t(e)^2
\\&= \frac{1}{m}\sum_{e \in E} t^+(e) t(e)
\\&= \frac{46T^{4/3}}{m}
\end{align} 

We thus have $E(A) = k T/m$ and $Var(A) = k Var(A_1)$ and hence $E(\hat{T}) = T$ and 
\[
Var(\hat{T}) = m \frac{46T^{4/3}}{k} = \frac{1}{3} \epsilon^2 T^{4/3} \tilde{T}^{2/3} \leq \frac{1}{3} \epsilon^2 T^2
\]
where the last inequality holds for $\tilde{T} \leq T$. The expectation of $\hat{T}$ is thus as desired. By Chebyshev's inequality, we have that
\[
P(|\hat{T} - T| > \epsilon T) < \frac{1/3 \epsilon^2 T^2}{(\epsilon T)^2} = 1/3
\]
\end{proof}

\subsection{Algorithm with both vertex and edge sampling}
\subsubsection{Finding heavy subgraph}

\begin{definition}
Let us be given a parameter $\theta > 0$. A vertex $v$ is heavy with respect to a set $S$ if $\sum_{uv, u \in S} t(uv) \geq 8\sqrt{\theta} \log n$. Otherwise, it is light. Let $V_H(S), V_L(S)$ be the set of heavy and light vertices w.r.t. $S$, respectively.
\end{definition}

We use just $V_H,V_L$ when the set $S$ is clear from the context. When testing whether a vertex is heavy, we assume we are already given the set $S$ and the vertex $v$. We then do not make any additional queries to tell whether $v$ is heavy or light (as the sum in the definition of a heavy vertex only depends on the neighborhoods of $S,v$, which we know since we have already queried $v$ as well as all vertices of $S$). In our algorithm, we will set $S$ to be a subset of vertices such that each vertex is in $S$ independently with probability $16 \log n/\sqrt{\theta}$. 
We now prove some guarantees on which vertices will be heavy and how many heavy vertices there will be.

\begin{lemma} \label{lem:vertex_is_heavy}
Let us have a parameter $\theta$ and a vertex $v$. Assume $S$ includes each vertex independently with probability $16 \log n/ \sqrt{\theta}$. Assume that at least $\theta$ triangles are assigned to the edges incident to $v$. That is, assume $\sum_{e \ni v} t^+(e) \geq \theta$. Then, with probability at least $1-1/n^2$, the vertex $v$ is heavy.
\end{lemma}
\begin{proof}
In the whole proof, we treat $v$ as given and define values based on $v$ without explicitly specifying this throughout the proof. Let $t^+_{max} = \max_{e \ni v} t^+(e)$. For an edge $e$, we define $t^m(e) = \min(t^+_{max}, t(e))$. Let $T_v^m = \sum_{e \ni v} t^m(e)$.

We will now argue that for any edge $e$ incident to vertex $v$, it holds $t^m(e) \leq \sqrt{T_v^m}$. Consider the triangles assigned to $e$. Each of the $t^+(e)$ triangles assigned to $e$ consists of $e$, one other edge $e'$ incident to $v$ and one edge not incident to $v$. It holds $t(e') \geq t(e) \geq t^+(e)$ for otherwise the triangle would be assigned to $e'$ instead of $e$. Thus, it also holds $t^m(e') = \min(t^+_{max}, t(e'))\geq t^+(e)$. Since there are $t^+(e)$ such edges $e'$ (as there is a 1-to-1 correspondence between such edges and triangles assigned to $e$) they contribute at least $(t^+(e))^2$ to $T^m_v$ and thus $T^m_v \geq (t^+(e))^2$. Rearranging this, we get $t^+(e) \leq \sqrt{T^m_v}$. Since this holds for any $e \ni v$, it also holds $t^+_{max}  \leq \sqrt{T^m_v}$ and thus also $t^m(e) = \min(t^+_{max}, t(e)) \leq \sqrt{T^m_v}$.

This bound on $t^m(e)$ allows us to now use the Chernoff bound to prove concentration of $\sum_{uv, u \in S} t^m(uv)$ around $E(\sum_{uv, u \in S} t^m(uv)) = 16 T_v^m \log n/\sqrt{\theta}$. 
This gives us that $\sum_{uv, u \in S} t(uv)$ is not too small and $v$ is thus heavy. Specifically, we get that 
%
\[
P(\sum_{uv, u \in S} t(uv) < 8\sqrt{\theta} \log n) \leq P(\sum_{uv, u \in S} t^m(uv) < 8 T_v^m \log n/\sqrt{\theta}) < \exp(-\frac{16 T^m_v \log n/\sqrt{\theta}}{8 \sqrt{T^m_v}}) \leq 1/n^2
\]
where the first inequality holds because $T_v^m \geq \sum_{e \ni v} t^+(uv) \geq \theta$ and thus $\sqrt{\theta} \leq T^m_v / \sqrt{\theta}$, the second inequality holds by the Chernoff bound, and the third holds again because $T_v^m \geq \theta$.
\end{proof}

\begin{lemma} \label{lem:not_many_heavy_vertices}
Assume $S$ includes each vertex independently with probability $16 \log n / \sqrt{\theta}$. It holds $E(|V_H(S)|) \leq 6 T/\theta$. 
\end{lemma}
\begin{proof}
Let $T_v$ be the number of triangles containing vertex $v$. We can now bound $E(|V_H|)$ as follows
\begin{align} 
E(|V_H|) &= \sum_{v \in V} E(\Ii(v \in V_H)) \\
&= \sum_{v \in V} P(\sum_{u \in N(v) \cap S} t(uv) \geq 8\sqrt{\theta} \log n) \\
&\leq \sum_{v \in V} \frac{E(\sum_{u \in N(v) \cap S} t(uv))}{8 \sqrt{\theta} \log n} \\
&= \sum_{v \in V} \frac{16 T_v \log n / \sqrt{\theta}}{8 \sqrt{\theta} \log n}
= 6T/\theta
\end{align} 
where the inequality holds by the Markov's inequality.
\end{proof}

\subsubsection{Putting it together}
\begin{algorithm}
\If{$\tilde{T} \geq \frac{m^3}{n^3 \log^6 n}$}{
    Use \Cref{alg:sampling_edges_with_advice} instead
}
\medskip

$\theta \leftarrow \sqrt{\frac{m \tilde{T}}{n}}$\\
$S \leftarrow$ sample each vertex with probability $p_1 = 16 \log n/\sqrt{\theta}$ \label{line:sample_S}\\
%
$S_v \leftarrow$ sample each vertex with probability $p_2 = 100 \sqrt{\log n} \theta /(\epsilon^2 \tilde{T})$ \label{line:sample_Sv}\\
$S_t \leftarrow$ sample each vertex with probability $p_3 = \theta \log n/\tilde{T}$ \label{line:sample_St}\\
$N \leftarrow$ sample each vertex with probability $p_4 = \theta \log n/(\epsilon^2 \tilde{T})$ \label{line:sample_N}\\
$A_{v,1} \leftarrow 0$\\
$A_{v,2} \leftarrow 0$\\
\For{$uv \in E(G[S_v])$}{
    \If{both $u$ and $v$ are light w.r.t. $S$}{
        \eIf{$t^+_{S_t}(uv) < 162 \log n / \epsilon^2$}{
            $B \sim Bern(\epsilon^2 \tilde{T}^2/\theta^3)$\\
            \If{B = 1}{
                $w \leftarrow$ random vertex from $N(u) \cap N(v)$ \label{line:estimating_edge_from_incuced_subgraph_few_triangles}\\
                \If{$uv \prec uw$ and $uv \prec vw$}{ \label{line:few_triangles_is_the_triangle_assigned}
                    $A_{v,1} \leftarrow A_{v,1} + \theta^3 t(uv)/(\epsilon^2 \tilde{T}^2)$\\
                }
            }
        }{
            $w \leftarrow$ random vertex from $N \cap N(u) \cap N(v)$ \label{line:estimating_edge_from_incuced_subgraph_many_triangles}\\
            \If{$uv \prec uw$ and $uv \prec vw$}{ \label{line:many_triangles_is_the_triangle_assigned}
                $A_{v,2} \leftarrow A_{v,2} + t(uv)$\\
            }
        }
    }
}
$A_v \leftarrow A_{v,1} + A_{v,2}$\\
$\hat{T}_L \leftarrow A_v / p^2$\\

\medskip
$S_e \leftarrow$ sample $k = \frac{432 m (\log(n)+2)}{\epsilon^2 \theta}$ edges with replacement \label{line:sample_edges_incident_to_heavy}\\
$A_e \leftarrow 0$
\For{$uv \in S_e$}{
    \If{either $u$ or $v$ is heavy w.r.t. $S$}{
        $w \leftarrow$ random vertex from $N(u) \cap N(v)$ \label{line:sampling_third_vertex}\\
        \If{$uv \prec uw$ and $uv \prec vw$}{
            $A_e \leftarrow A_e + t(uv)$\\
        }
    }
}
$\hat{T}_H \leftarrow \frac{m A_e}{k}$\\

\Return{$\hat{T}_L + \hat{T}_H$}
    
\caption{\textsc{CountTrianglesVertexSampling}$(\tilde{T}, \epsilon)$} \label{alg:tc_vertex_sampling}
\end{algorithm}
In the analysis of this algorithms, we will need the following inequality. It is similar to \Cref{lem:my_inequality} but is tighter when we are only considering edges with non-empty intersection with some small given set of vertices. We will be using this lemma with the ``small set of vertices" being the set of heavy vertices $V_H$.
\begin{lemma} \label{lem:my_inequality2}
Let us have a set $V' \subseteq V$ and let $\ell = |V'|$. It holds that
\[
\sum_{\substack{e \in E\\e \cap V' \neq \emptyset}} t^+(e) t(e) \leq 6(2 + \log m) \ell T
\]
\end{lemma}
\begin{proof}
Just like in the proof of \Cref{lem:my_inequality}, we start by arguing several inequalities. We pick a permutation $\pi$ of the vertices such that $t(e_{\pi_1}) \geq t(e_{\pi_2 }) \geq \cdots \geq t(e_{\pi_m})$. Let us have $i \in \{2^{k-1}+1 , \cdots, 2^{k}\}$, it holds $(a)$ that $t(e_{\pi(i)}) \leq \frac{3 T}{2^{k-1}}$. Otherwise, the first $i$ edges in the $\pi$-ordering would have total of $i \frac{3T}{2^{k-1}} > 3T$ edge-triangle incidencies, which is in contradiction with $T$ being the number of triangles. We now bound $\sum_{i=2^{k-1}+1}^{2^k} t^+(e_{\pi(i)})$. It clearly holds $(b)$ that $\sum_{i=2^{k-1}+1}^{2^k} t^+(e_{\pi(i)}) \leq T$. For any triangle $e_{\pi(i)}e_{\pi(j)}e_{\pi(\ell)}$ assigned to $e_{\pi(i)}$, it holds that $i > j$,$i>\ell$. Therefore, any triangle assigned to an edge $e_{\pi(i)}$ is formed by edges in $\{e_{\pi(j)}\}_{j=1}^{2^k}$. In a graph on $2^k$ edges, the number of triangles with non-empty intersection with some given subset $V'$ of vertices of size at most $\ell$ is at most $\ell 2^k$. Therefore, we have $(c)$ that $\sum_{i=2^{k-1}+1}^{2^k} \Ii(e_{\pi(i)} \cap V' \neq \emptyset) t^+(e_{\pi(i)}) \leq l 2^k$. 
Like in the previous proof, in what follows, we use (in a slight abuse of notation) the convention $t(e_{\pi(i)}) = t^+(e_{\pi(i)}) = 0$ for $i > m$. We may now use these inequalities to bound the variance:
\begin{align}
\sum_{\substack{e \in E\\e \cap V' \neq \emptyset}} t^+(e) t(e)
&= \sum_{i=1}^m \Ii(e_{\pi(i)} \cap V' \neq \emptyset) t^+(e_{\pi(i)}) t(e_{\pi(i)}) 
\\&= \sum_{k=1}^{\lceil\log_2 m\rceil} \sum_{i=2^{k-1}+1}^{2^k} \Ii(e_{\pi(i)} \cap V' \neq \emptyset) t^+(e_{\pi(i)}) t(e_{\pi(i)})\tag{4}
\\\tag{5}&\leq 6 \sum_{k=1}^{\lceil\log_2 m\rceil} \frac{T}{2^k}\sum_{i=2^{k-1}+1}^{2^k} \Ii(e_{\pi(i)} \cap V' \neq \emptyset) t^+(e_{\pi(i)}) 
\\ \tag{6}&\leq 6 \sum_{k=1}^{\lceil\log_2 m\rceil} \frac{T}{2^k} \min(l 2^k,T)
\\&= \Bigg(\sum_{k=1}^{\log_2 T/\ell} \ell T +  \sum_{k=\log_2 T/\ell}^{\infty}\frac{T^2}{2^k}\Bigg)
\\&\leq 6 \Big( \ell T \log m + 2 \ell T\Big)
\\&= 6 (2 + \log m) \ell T   
\end{align}
where $(4)$ is using the convention $t(e_{\pi(i)}) = t^+(e_{\pi(i)}) = 0$ for $i > m$, $(5)$ holds by inequality $(a)$ and $(6)$ holds by inequalities $(b)$ and $(c)$.
\end{proof}

We now prove three lemmas, one on (conditional) expectation and variance of $A_{v,1}$, one on $A_{v,2}$ and one on (conditional) expectation of $\hat{T}_H$. Before we can state the lemmas, we will need several definitions.

Let $E_L$ be the set of edges whose both endpoints are light and $E_H = E \setminus E_L$ be the set of edges that have at least one heavy endpoint (note the asymmetry in the definitions). Let $E_{L,1}$ be the subset of $E_L$ of edges $e$ that have $t^+_{S_t}(e)< 162 \log n$ and $E_{L,2} = E_L \setminus E_{S,2}$. Let $T_L$ be the number of triangles assigned to edges in $G[V_L(S)]$ and let $T_H = T - T_L$ be the number of triangles assigned to edges having at least one vertex in $V_H(S)$. Let
\begin{align} 
T_{L,1} &= \sum_{\substack{u,v \in V_L\\ t^+_{S_t}(uv) < 162 \log n}} t^+(uv)
&T_{L,2} = \sum_{\substack{u,v \in V_L\\ t^+_{S_t}(uv) \geq 162 \log n}} t^+(uv)\\
\end{align} 
Note that $T_L = T_{L,1} + T_{L,2}$. Let us define for an edge $uv$ such that $t_{S_t}^+(e) \geq 162 \log n$
\[
F_{N, uv} = \frac{|\{w \in N\cap N(u) \cap N(v),\text{ s.t. } uv \prec uw, uv \prec vw\}|}{|N\cap N(u) \cap N(v)|}
\]
and we define $F_{N,uv} = t^+(e) / t(e)$ otherwise (note that in this case, the value does not depend on $N$). We now argue concentration of $F_{N,e}$ around $t^+(e)/t(e)$.
\begin{lemma}\label{lem:concentration_of_F}
With probability at least $1-3/n$, it holds for all edges $e$ that $F_{N,e} = (1\pm \epsilon/3) t^+(e)/t(e)$.
\end{lemma}
\begin{proof}
For an edge $e$ with $t_{S_t}^+(e) < 162 \log n$, the claim holds by the way we define $F_{N,e}$ for such edge.
Consider now the case $t_{S_t}^+(e) \geq 162 \log n$. By a standard argument based on the Chernoff bound over the choice of $S_t$, if $t^+(e) < 81 \tilde{T} \log n/\theta$, then $t_{S_t}^+(e) \geq 162 \log n$ only on some event $\Ee$ with probability $\leq 1/n^3$. Conditioned on $\bar{\Ee}$ (that is, assuming $t^+(e) \geq 81 \tilde{T} \log n/\theta$), another application of the Chernoff bound gives that $F_{N, uv} = (1\pm \epsilon/3) t^+(e) /t(e)$ with probability at least $1-2/n^3$. It then holds by the union bound that with probability at least $1-3/n$, we have that $F_{N,e} = (1\pm \epsilon/3) t^+(e)$ for all edges $e$.
\end{proof}

\begin{lemma} \label{lem:light_subgraph_exp}
It holds $E(\hat{T}_L |S) = T_L$ and, with high probability, $E(\hat{T}_L |S,S_t,N) = (1\pm \epsilon/3) T_L$ (where the high probability is over the choice of $S,S_t,N$)\footnote{In fact, it is sufficient to consider the probability over $S_t,N$ but we will not need this.}.
\end{lemma}
\begin{proof}
We start by analyzing $E(A_{v,1} | S,S_t)$. For any edge $e \in E_{L,1}$, the probability that it is in the induced subgraph $G[S_v]$ is $p^2_2$. The probability that $B = 1$ is $\epsilon^2 \hat{T}^2/\theta^3$. Conditioned on both of this happening, the probability that $w$ satisfies the condition on \cref{line:few_triangles_is_the_triangle_assigned} (that is, that the triangle $uvw$ is assigned to $e$) is $t^+(e)/t(e)$ in which case we increment $A_{v,1}$ by $\theta^3 t(e)/(\hat{T}^2 \epsilon^2)$. In expectation (conditioned on $S$ and $S_t$), $e$ contributes to $A_{v,1}$ exactly $p_2^2 \cdot \frac{T^2\epsilon^2}{\theta^3} \cdot \frac{t^+(e)}{t(e)} \cdot \frac{t(e)\theta^3}{T^2\epsilon^2} = p_2^2 t^+(e)$. Any edge not in $E_{L,1}$ contributes $0$ (conditioned on $S, S_t$). By the linearity of expectation, it holds that $E(A_{v,1} |S,S_t) = p_2^2 T_{L,1}$. Moreover, $A_{v,1}$ is clearly independent of $N$ and it thus also holds $E(A_{v,1} |S,S_t,N) = p_2^2 T_{L,1}$.

We now analyze $E(A_{v,2} | S,S_t,N)$. For any edge $e \in E_{L,2}$, the probability that it is in the induced subgraph $G[S_v]$ is $p_2^2$. Conditioned on this happening, the probability that $w$ satisfies the condition on \cref{line:many_triangles_is_the_triangle_assigned} (that is, that the triangle $uvw$ is assigned to $e$) is $F_{N, uv}$ in which case we increment $A_{v,2}$ by $t(e)$. In expectation (conditioned on $S,S_t,N$), $e$ contributes to $A_{v,2}$ exactly $C_e = p^2 F_{N,uv} t(e)$. By \Cref{lem:concentration_of_F}, we have with probability at least $1-3/n$ that for all edges $e$, $C_e = (1\pm \epsilon/3) p_2^2 t^+(e)$.
By the linearity of expectation, it then holds that $E(A_{v,2} |S,S_t,N) = (1\pm \epsilon/3)p_2^2 T_{L,2}$.

Putting this together, we have $E(A_{v} |S,S_t,N) = p_2^2 T_{L,1} + (1\pm\epsilon/3) p_2^2 T_{L,2} = (1\pm\epsilon/3) p_2^2 T_L$ and thus $E(\hat{T}_L | S,S_t,N) = E(A_v | S,S_t,N)/p_2^2 = (1\pm\epsilon/3) T_L$.

It holds $E(C_e | S, S_t) = E(F_{N,e} | S,S_t) t(e) = t^+(e)$. It thus holds $E(A_{v,2} |S,S_t) = p_2^2 T_{L,2}$. Putting this together, we have $E(A_{v} |S,S_t) = p_2^2 T_{L,1} +  p_2^2 T_{L,2} = p_2^2 T_L$ and thus $E(\hat{T}_L | S,S_t) = T_L$.
\end{proof}

\begin{lemma} \label{lem:light_subgraph_var}
It holds with high probability that $Var(\hat{T}_L | S,S_t,N) = \frac{1}{8} (\epsilon T)^2$. (where the high probability is over the choice of $S,S_t,N$).
\end{lemma}
\begin{proof}
We now analyze the variance of $A_v = A_{v,1} + A_{v,2}$, conditional on $S,S_t,N$ and assuming $F_{N,e} = (1\pm\epsilon/3) \frac{t^+(e)}{t(e)}$ (and thus also $C_e = (1\pm \epsilon/3) p_2^2 t^+(e)$); recall that this holds for all edges $e$ with probability at least $1-3/n$. We analyze the variance using the law of total variance, conditioning on $S_v$. We start by upper-bounding the variance of the conditional expectation.

Let $X_{uv} = F_{N,uv} t(uv)$ if $u,v \in S_v$ and $u,v \in V_L$, and let $X_{uv} = 0$ otherwise.
If $t^+_{S_t}(uv) < 162 \log n / \epsilon^2$, then it holds $X_{uv} = t(uv) P(uv \prec uw \wedge uv \prec vw) = \theta^3/(\epsilon^2 \tilde{T}^2)t(uv) P(B=1) P(uv \prec uw \wedge uv \prec vw)$ for $w \sim N(u) \cap N(v)$. For $t^+_{S_t}(uv) \geq 162 \log n / \epsilon^2$, it holds $X_{uv} = t(uv) P(uv \prec uw \wedge uv \prec vw)$ for $w \sim N \cap N(u) \cap N(v)$ (note the different distribution of $w$).
We then have
\begin{align} 
E(A_v | S,S_v,S_t,N) = &\sum_{\substack{u,v \in S_v \\ u,v \in V_L\\ t^+_{S_t}(e)  < 162 \log n /\epsilon^2}} t(uv) P_{w \sim N(u) \cap N(v)}(uv \prec uw \wedge uv \prec vw) +\\& \sum_{\substack{u,v \in S_v \\ u,v \in V_L\\ t^+_{S_t}(e)  \geq 162 \log n /\epsilon^2}} t(uv) P_{w \sim N \cap N(u) \cap N(v)}(uv \prec uw \wedge uv \prec vw)
= \sum_{e \in E_L} X_e
\end{align} 
\jakub{tohle je špatně, musím tam zohlednit to, že koukám na $B$.}
Calculating the conditional variance of this expectation, we get
\begin{align} 
Var(E(A_v | S,S_v,S_t,N)|S,S_t,N) &= Var(\sum_{e \in E_L}X_e | S,S_t,N) \\&\leq \sum_{e \in E_L} Var(X_e|S,S_t,N) + \sum_{\substack{e_1, e_2 \in E_L\\e_1 \cap e_2 \neq \emptyset}} E(X_{e_1} X_{e_2}|S,S_t,N) 
\end{align} 
%
%
%
%
where we are using that $X_{e_1}$ and $X_{e_2}$ are independent when $e_1 \cap e_2 = \emptyset$, conditionally on $S_t,S,N$. We can further bound
\[
\sum_{e \in E_L} Var(X_e|S,S_t,N) \leq p_2^2 \sum_{e \in E_L} ((1+\epsilon/3) t^+(e))^2 \leq 92 p_2^2 T^{4/3}
\]
where the first inequality holds because $Var(X_e|S_t,S,N) \leq E(X_e^2 |S_t,S,N) \leq p^2_2 ((1+\epsilon/3) t^+(e))^2$ 
and the second inequality holds by \Cref{lem:my_inequality} and because $(1+\epsilon/3)^2 < 2$ because $\epsilon < 1$. We now bound
\begin{align} 
\sum_{\substack{e_1, e_2 \in E_L\\e_1 \cap e_2 \neq \emptyset}} E(X_{e_1} X_{e_2}|S_t,S,N) &= p_2^3 \sum_{\substack{e_1, e_2 \in E_L\\e_1 \cap e_2 \neq \emptyset}} F_{N,e_1} F_{N,e_2} t(e_1) t(e_2)\\
&\leq 2 p_2^3 \sum_{\substack{e_1, e_2 \in E_L\\e_1 \cap e_2 \neq \emptyset}} \frac{t^+(e_1)}{t(e_1)} \cdot \frac{t^+(e_2)}{t(e_2)} \cdot t(e_1) t(e_2)
\\&= 2 p_2^3 \sum_{e_1 \in E_L} t^+(e_1) \sum_{e_2 \in N(e_1)}t^+(e_2)
\\&\stackrel{\text{w.h.p.}}{\leq} 2 p_2^3 \theta \sum_{e \in E_L} t^+(e) = 2 p_2^3 \theta T_L\\
\end{align} 
where the second inequality holds with probability at least $1-1/n^2$ by \Cref{lem:vertex_is_heavy} since we are assuming both endpoints of $e$ are light.
The first holds for the following reason. As we already mentioned, we are assuming that for $N$, it holds $F_{N, e} = (1\pm\epsilon/3)\frac{t^+(e)}{t(e)}$ for any edge $e$ (this happens with probability at least $1-O(1/n)$). We thus have
\[
F_{N, e_1}F_{N, e_2} \leq (1+\epsilon/3)^2 \frac{t^+(e_1)}{t(e_1)} \cdot \frac{t^+(e_2)}{t(e_2)} \leq 2 \frac{t^+(e_1)}{t(e_1)} \cdot \frac{t^+(e_2)}{t(e_2)}
\]
Together, this gives us a bound
\[
Var(E(A_v | S,S_v,S_t,N)|S,S_t,N) \leq 92 p^2_2 T^{4/3} + 2p_2^3 \theta T
\]

We now bound the expectation of variance. Let $Y_e$ be the increment to $A_v$ contributed by an edge $e \in E(G[S_v])$
. It holds that
\begin{align}
Var(A_v | S,S_v,S_t,N) &= \sum_{e \in E(G[S_v])} Var(Y_e | S,S_t,N) \\
&= \sum_{\substack{e \in E(G[S_v])\\ t^+_{S_t}(e) < 162 \log n}} Var(Y_e | S,S_t,N) + \sum_{\substack{e \in E(G[S_v])\\ t^+_{S_t}(e) \geq 162 \log n}} Var(Y_e | S,S_t,N) \\
&\leq \sum_{\substack{e \in E(G[S_v])\\ t^+_{S_t}(e) < 162 \log n}} \frac{\theta^3}{\epsilon^2 \tilde{T}^2} \cdot t^+(e)t(e) + \sum_{\substack{e \in E(G[S_v])\\ t^+_{S_t}(e) \geq 162 \log n}} F_{N,e} t(e)^2 \\
&\stackrel{\text{w.h.p.}}{\leq} 205\, \frac{\theta^2 \log n}{\epsilon^2 \tilde{T}} \sum_{\substack{e \in E(G[S_v])\\ t^+_{S_t}(e) < 162 \log n}} t(e) + (1+\epsilon/3) \sum_{\substack{e \in E(G[S_v])\\ t^+_{S_t}(e) \geq 162 \log n}} t^+(e) t(e) 
\end{align}
where the first inequality holds because $Y_e = t(e) \theta^3/(\epsilon^2 \tilde{T})$ with probability $t^+(e)/t(e) \cdot \epsilon^2 \tilde{T}^2/\theta^3$ and $Y_e = 0$ otherwise. The second inequality holds because by the Chernoff and the union bound, it holds that, with probability at least $1-1/n$. for any edge $e$ such that $t^+_{S_t}(e) < 162 \log n$, it holds that $t^+(e) < 205 \tilde{T} \log n/\theta$.
It, therefore, holds by the linearity of expectation that
\begin{align} 
E(Var(A_v | S,S_t,S_v,N)S,S_t,N) &\leq 205\, \frac{\theta^2 \log n}{\epsilon^2 \tilde{T}} \sum_{\substack{e \in E\\ t^+_{S_t}(e) < 162 \log n}} p_2^2 t(e) + (1+\epsilon/3) \sum_{\substack{e \in E\\ t^+_{S_t}(e) \geq 162 \log n}} p_2^2t^+(e) t(e)\\
&\leq 615\, \frac{p_2^2 \theta^2 T \log n}{\epsilon^2 \tilde{T}} + 62 p_2^2 T^{4/3} \nonumber
\end{align} 
By the law of total variance, it holds
\begin{align}
Var(A_v|S,S_t,N) &= E(Var(A_v | S,S_t,S_v,N)|S,S_t,N) + Var(E(A_v|S,S_t,S_v,N)|S,S_t,N) \\&\leq 615\, \frac{p_2^2 \theta^2 T \log n}{\epsilon^2 \tilde{T}} + 62 p_2^2 T^{4/3} + 92 p^2_2 T^{4/3} + 2p_2^3 \theta T \\&
\leq \frac{1}{16}(\epsilon p^2 T)^2 + \frac{1}{160}(\epsilon p^2 T)^2 + \frac{1}{100}(\epsilon p^2 T)^2 + \frac{1}{50}(\epsilon p^2 T)^2 < \frac{1}{8}(\epsilon p_2^2 T)^2
\end{align}
where the second inequality holds because of the way we set $p_2,\theta$ and because $T < m^3/n^3$. We thus have $Var(\hat{T}_L | S) = Var(A_v|S)/p_2^4 = \frac{1}{8} (\epsilon T)^2$.
\end{proof}

\begin{lemma} \label{lem:heavy_subgraph_exp_and_var}
It holds $E(\hat{T}_H | S,S_t,N) = T_H$ and $E(Var(T_H | S,S_t,N)|S_t,N) \leq \frac{1}{12} (\epsilon^2 T^2)^2$
\end{lemma}
\begin{proof}
We now focus on $A_e$. Let $\Delta_i A_e$ be the $i$-th increment of $A_e$. It holds
\[
E(\Delta_i A_e | S,S_t,N) = \frac{1}{m}\sum_{e, e \cap V_H \neq \emptyset} \frac{t^+(e)}{t(e)} t(e) = \frac{1}{m}\sum_{e, e \cap V_H \neq \emptyset} t^+(e) = T_H/m
\]
Therefore, $E(A_e | S,S_t,N) = E(\sum_{i=1}^k \Delta_i A_e |S,S_t,N) = k T_H/m$ and thus $E(\hat{T}_H | S,S_t,N) = T_H$. Let $Y$ be the number of heavy vertices with respect to $S$.
We have
\[
Var(\Delta_i A_e|S,S_t,N) \leq \frac{1}{m} \sum_{\substack{e \in E\\e \cap V_H \neq \emptyset}} \frac{t^+(e)}{t(e)} t(e)^2 = \sum_{\substack{e \in E\\e \cap V_H \neq \emptyset}} t^+(e) t(e) \leq 6 (2 + \log n) Y T / m
\]
where the last inequality holds by \Cref{lem:my_inequality2}. We thus have $Var(T_H | S) = m^2 Var(\Delta_1 A_e|S)/k \leq 6 m (2 + \log n) Y T/k$. It thus holds by \Cref{lem:not_many_heavy_vertices} that
\begin{align}
E(Var(T_H | S,S_t,N)|S_t,N) \leq \frac{36 (2 + \log n) m T^2}{\theta k} = \frac{1}{12} (\epsilon^2 T^2)
\end{align}
where the equality holds by our choice of $k$.
\end{proof}

We are now in position to state and prove the main lemma. 
\begin{lemma}
\Cref{alg:tc_vertex_sampling} returns an unbiased estimate $\hat{T}$ of $T$ and has expected query complexity of $O(\min(\frac{m}{\epsilon^2 \tilde{T}^{2/3}},\frac{\sqrt{n m}\log n}{\epsilon^2 \sqrt{\tilde{T}}}))$. Moreover, if $\tilde{T} \leq T$, then with probability at least $2/3$, it holds $\hat{T} = (1\pm\epsilon)T$.
\end{lemma}
\begin{proof}
We first argue the query complexity. We then argue the last part of the lemma, namely the deviation bounds on $\hat{T}$. Finally, we then prove that the returned estimate is unbiased.

\medskip \noindent
\emph{We now argue the complexity of the algorithm.}
In the case $\tilde{T} \geq m^3 /(n^3 \log^6 n)$, the dominant branch of the $\min$ is $m/(\epsilon^2 \tilde{T}^{2/3})$. In this case, we execute \Cref{alg:sampling_edges_with_advice} and by \Cref{lem:tc_edge_sampling}, the complexity is $O(\frac{m}{\epsilon^2 \tilde{T}^{2/3}})$ as desired.

We now consider the case $\tilde{T} < m^3 /(n^3 \log^6 n) $. On \cref{line:sample_S}, we perform in expectation $n p_1 = O(n \log n/\sqrt{\theta}) \subseteq O(\frac{\sqrt{n m} \log n}{\sqrt{\tilde{T}}})$ queries. On \cref{line:sample_Sv,line:sample_St,line:sample_N}, we perform in expectation $O(n \theta \log n/ (\epsilon^2 \tilde{T})) = O(\frac{\sqrt{n m} \log n}{\epsilon^2 \sqrt{\tilde{T}}})$ queries. We now calculate the expected number of queries performed on \cref{line:estimating_edge_from_incuced_subgraph_few_triangles}. Each of the $m$ edges is in $G[S_v]$ with probability $p_2^2 = \Theta(\theta^2/(\epsilon^4 \tilde{T}^2))$. For each such edge, we make a query only when $B=1$, which happens with probability $\epsilon^2 \tilde{T}^2/\theta^3$. This gives us expectation of (up to a constant factor)
\[
m\cdot \theta^2/(\epsilon^4 \tilde{T}^2) \cdot \epsilon^2 \tilde{T}^2/\theta^3 = \frac{m}{\epsilon^2 \theta} = \frac{\sqrt{n m}}{\epsilon^2 \sqrt{\tilde{T}}}
\]
On line \cref{line:estimating_edge_from_incuced_subgraph_many_triangles}, we only access vertices that have previously been queried, as $w \in N$. This therefore does not increase the query complexity.
On \cref{line:sample_edges_incident_to_heavy}, we sample $k = O(\frac{m \log n}{\epsilon^2 \theta}) = O(\frac{\sqrt{nm}}{\epsilon^2 \sqrt{T}})$ edges. On \cref{line:sampling_third_vertex}, we sample at most one vertex for each of these sampled edges, thus not increasing the asymptotic complexity. Putting all the query complexities together, the total query complexity is as claimed.

\medskip \noindent

\emph{We now prove the deviation bounds.} We condition on the high-probability events from \Cref{lem:light_subgraph_exp,lem:light_subgraph_var}. We do not make this conditioning explicit in the rest of the proof. We have by \Cref{lem:light_subgraph_exp} that $E(\hat{T}_L | S,S_t,N) = (1\pm\epsilon/3)T_L$ and by \Cref{lem:heavy_subgraph_exp_and_var} that $E(\hat{T}_H | S,S_t,N) = T_H$. Therefore, $E(\hat{T} | S,S_t,N) = E(\hat{T}_L + \hat{T}_H | S,S_t,N) =  (1+\epsilon/3)T$. By \Cref{lem:light_subgraph_var}, it holds $Var(\hat{T}_L | S,S_t,N) \leq \frac{1}{8} (\epsilon T)^2$ and by \Cref{lem:heavy_subgraph_exp_and_var}, we have $Var(\hat{T}_H | S,S_t,N) \leq \frac{1}{8} (\epsilon T)^2$. The random variables $\hat{T}_L,\hat{T}_H$ can be easily seen to be independent conditionally on $S$. Since $\hat{T}_H$ is independent of $S_t$ and $N$, it holds that $\hat{T}_L$ and $\hat{T}_H$ are also independent conditionally on $S,S_t,N$. It thus holds that $Var(\hat{T} | S,S_t,N) = Var(\hat{T}_L | S,S_t,N) + Var(\hat{T}_H | S,S_t,N)$.
We now use the law of total variance, conditioning on $S$, to bound $Var(\hat{T} | S_t,N)$. Since $E(\hat{T} | S,S_t,N) = (1\pm \epsilon/3)T$, it holds $Var(E(\hat{T} | S,S_t,N) | S_t,N) \leq (2\epsilon T/3)^2 / 4 = 1/9 (\epsilon T)^2$. It then holds
\begin{align}
Var(\hat{T} | S_t,N) &= E(Var(\hat{T} | S,S_t,N) | S_t,N) + Var(E(\hat{T} | S,S_t,N)| S_t,N) \\ &\leq \frac{1}{8} (\epsilon T)^2 + \frac{1}{12} (\epsilon T)^2 + \frac{1}{9} (\epsilon T)^2 \leq \frac{72}{23} (\epsilon T)^2
\end{align}
where the first of the three terms of the bound comes from \Cref{lem:light_subgraph_var}, the second comes from \Cref{lem:heavy_subgraph_exp_and_var}, and the third from the above bound on $Var(E(\hat{T} | S,S_t,N) | S_t,N)$. Therefore, by the Chebyshev inequality, we have that
\[
P(|\hat{T} - T| > \epsilon T) \leq \frac{Var(\hat{T})}{(\epsilon T)^2} \leq 72/23
\]
Adding probability of $O(1/n)$ of the complement of the events we are conditioning on, we have that $P(|\hat{T} - T| > \epsilon T) \leq 1/3$ for $n$ large enough.

\medskip \noindent
\emph{We now argue unbiasedness.}
We now argue that the algorithm gives an unbiased estimate of $T$. It holds
\[
E(\hat{T}) = E(E(\hat{T}_L + \hat{T}_H| S)) = T_L + T_H = T 
\]
%
%
\end{proof}
By standard advice removal (as we discussed in \Cref{sec:removing_advice}), we get the following theorem.
\begin{theorem}
There is an algorithm in the full neighborhood access model with random vertex and edge queries, that returns $\hat{T}$ such that, with probability at least $2/3$, it holds $\hat{T} = (1\pm \epsilon) T$ and has expected query complexity of $O(\min(\frac{m}{\epsilon^2 T^{2/3}},\frac{\sqrt{n m} \log n}{\epsilon^2 \sqrt{T}}))$.
\end{theorem}

We may modify the algorithm so that \Cref{alg:sampling_edges_with_advice} is called on line 2 whenever $\hat{T} \geq m^3/(\epsilon^3 n^3)$. We simulate the random edge queries using \Cref{alg:sampling_with_with_replacement}. The number of random edge queries in the case $\hat{T} \geq m^3/(\epsilon^3 n^3)$ is then $s = m/(\epsilon^2 T^{2/3})$ and the complexity is thus $O(n \sqrt{s/m} + s) = O(\frac{n}{\epsilon T^{1/3}})$. In the case $\hat{T} < m^3/(\epsilon^3 n^3)$, we make $s = O(\frac{\sqrt{n m} \log n}{\epsilon^2 \sqrt{T}})$ queries and the complexity is thus $O(n \sqrt{s/m} + s) = O(\frac{\sqrt{n m} \log n}{\epsilon^2 \sqrt{T}})$. This gives total complexity of $O(\frac{n}{\epsilon T^{1/3}} + \frac{\sqrt{n m} \log n}{\epsilon^2 \sqrt{T}})$:
%
\begin{theorem}
There is an algorithm in the full neighborhood access model with random vertex queries, that returns $\hat{T}$ such that, with probability at least $2/3$, it holds $\hat{T} = (1\pm \epsilon) T$ and has query complexity of $O(\frac{n}{\epsilon T^{1/3}} + \frac{\sqrt{n m} \log n}{\epsilon^2 \sqrt{T}}))$.
\end{theorem}

\subsection{Lower bound}
\subsubsection{Proving \texorpdfstring{$\Omega(m/T^{2/3})$}{Omega(m/T\string^2/3)} and \texorpdfstring{$\Omega(n/T^{1/3})$}{Omega(n/T\string^1/3)}}
\begin{theorem}
Any algorithm that returns $\hat{T}$ such that $\hat{T} = (1\pm \epsilon) T$ with probability at least $2/3$ and uses (1) only random vertex queries and solves, (2) only random edge queries, (3) both random vertex and edge queries has to have query complexity at least (1) $\Omega(n/T^{1/3})$, (2) $\Omega(m/T^{2/3})$, (3) $\Omega(\min(n/T^{1/3}, m/T^{2/3}))$.
\end{theorem}
\begin{proof}
Given $n$ and $m$, we construct two graphs with $\Theta(n)$ vertices and $\Theta(m)$ edges with the same number of vertices and edges. We than prove that the two graphs are hard to distinguish using $o(n/T^{1/3})$ random vertex samples and $o(m/T^{2/3})$ random edges. The three lower bounds follow. Let $H$ be a triangle-free graph on $n$ vertices and $m$ edges. We let $G_1$ be a disjoint union of $H$ with a clique of size $k = \Theta(T^{1/3})$ that has at least $T$ triangles and $\ell$ edges. Let $G_2$ be the disjoint union of $H$ with a bipartite graph on $k$ vertices with $\ell$ edges. This means that $G_1$ has $\geq T$ triangle while $G_2$ is triangle-free.

Consider sampling a pair of vertices $v_1\in G_1,v_2 \in G_2$ as follows. With probability $n/(n+k)$, we sample a vertex $v$ uniformly from $H$ and set $v_1 = v_2 = v$. Otherwise, we sample $v_1$ independently uniformly from $G_1 \setminus H$ and $v_2$ from $G_2 \setminus H$. One can easily verify that $v_1$ is sampled uniformly from $G_1$ and $v_2$ from $G_2$. It holds $P(v_1 \neq v_2) \leq k/(n+k) = \Theta(T^{1/3}/n)$ (the equality holds because $k = \Theta(T^3) \leq O(n)$). 

Similarly, we sample $e_1 \in G_1$ and $e_2 \in G_2$ as follows. We sample with probability $m/(m+\ell)$ an edge $e$ uniformly from $H$ and set $e_1 = e_2 = e$. Otherwise, we sample $e_1,e_2$ independently from $G_1 \setminus H, G_2 \setminus H$, respectively. It holds $P(e_1 \neq e_2) \leq \ell/(n+\ell) = \Theta(T^{2/3}/m)$ (the equality holds because $\ell = \Theta(T^{2/3}) \leq O(m)$). Making $o(n/T^{1/3})$ vertex samples and $o(m/T^{2/3})$ edge samples from the above couplings, it holds by the union bound that, with probability $1-o(1)$, the samples are equal. This means that any algorithm executed (with the same randomness) on these two samples, has to give the same answer on $G_1$ and $G_2$ with probability at least $1-o(1)$. This is in contradiction with the algorithm being correct with probability $2/3$.
\end{proof}

\begin{figure}[t]
\includegraphics[width=0.8\textwidth]{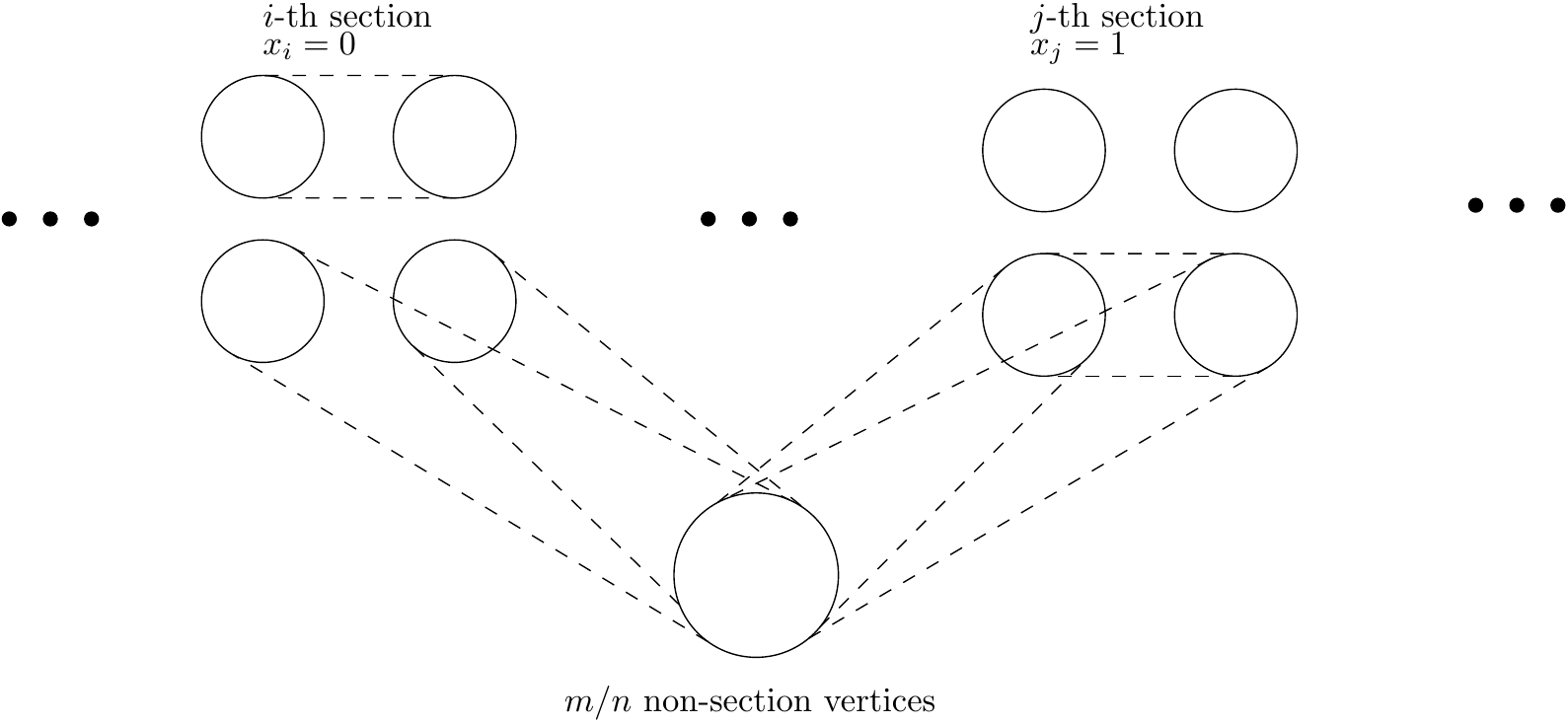}
\centering
\caption{The hard instance used in the proof of \Cref{thm:lb_tc}.} \label{fig:illustration_of_lb}
\end{figure}

\subsubsection{Proving \texorpdfstring{$\Omega(\sqrt{n m/T})$}{sqrt(nm/T)}}
The proof is via reduction from the OR problem: given $(x_1, \cdots, x_n) \in \{0,1\}^n$, the OR problem asks for the value $\bigvee_{i=1}^n x_i$. It is known that any algorithm that solves the OR problem with probability at least $2/3$ has complexity $\Omega(n)$. See for example \cite{Assadi_notes_or_problem} for a proof of this.

\begin{theorem} \label{thm:lb_tc}
Assume $T \leq m^3/n^3$. Any algorithm that returns $\hat{T}$ such that $\hat{T} = (1\pm \epsilon) T$ with probability at least $2/3$ and uses both random vertex and edge queries must have query complexity at least $\Omega(\sqrt{n m/T})$.
\end{theorem}
\begin{proof}
The proof is by reduction from the OR problem of size $\ell = \sqrt{n m/T}$. Specifically, we show that if we have an algorithm that solves the triangle counting problem in expectation using $Q$ full neighborhood queries, then we can also solve the OR problem of size $\ell$ in $Q$ queries. This implies the bound as one needs $\Omega(\ell) = \Omega(\sqrt{n m/T})$ queries to solve the OR problem of size $\ell$.

Given a vector $x = x_1, \cdots, x_\ell$, we define a graph $G_x = (V, E_x)$. 
The graph consists of $\ell$ \emph{sections} $s_1, \cdots s_\ell$ and $m/n$ \emph{non-section vertices}, where each section consists of $4\sqrt{nT/m}$ \emph{section vertices}. Each section is divided into four \emph{subsections}, with each subsection having $\sqrt{nT/m}$ vertices. We call the subsections top-left, top-right, bottom-left, and bottom-right. We order vertices within a section such that vertices in the top-left subsection come first, then top-right, bottom-left, and bottom-right in this order. We order the sections arbitrarily. Together with the orders on the vertices within a section, this induces an order on all section vertices. We put the non-section vertices at the end of this order.

We have specified the vertex set; we now specify the edges of $G_x$. There is an edge between each right subsection vertex and each non-section vertex. Consider the $i$-th section in the ordering. If $x_i = 0$, then there is an edge between each top-left vertex and top-right vertex. That is, for $x_i=0$, the top subsections induce a complete bipartite graph and the bottom subsections induce an independent set. If $x_i=1$, the top subsections induce an independent set and there is an edge between each bottom-left and bottom-right vertex. We represent each edge $e$ consisting of vertices $u$ and $v$ as a pair $(u,v)$ such that $u$ comes in the order before $v$. We order edges contained within a section lexicographically, with edges incident to the non-section vertices coming at the end.


In the rest of the proof, we argue that (1) the graph $G_x$ has $\Theta(n)$ vertices, $\Theta(m)$ edges (2) $\bigvee_{i = 1}^\ell x_i = 0$ iff $G_x$ is triangle-free and $\bigvee_{i = 1}^\ell x_i = 1$ iff $G_x$ has at least $T$ triangles, and (3) that we can simulate a query on $G_x$ by a single query on $x$. Proving these three things means that any algorithm that solves the approximate triangle counting problem can be used to give an algorithm for the OR problem of size $\ell$ with the same query complexity. This then implies the lower bound.

The number of vertices is $m/n + \sqrt{\frac{n m}{T}} \cdot 4\sqrt{n T/m} = \Theta(n)$. The number of edges is $m/n \cdot \sqrt{\frac{n m}{T}} \cdot \sqrt{n T/m} + \sqrt{\frac{n m}{T}} \cdot n T/m = \Theta(m) + \frac{n^{3/2}T^{1/2}}{m^{1/2}} \leq \Theta(m)$ where the last inequality is true by assumption $T \leq m^3/n^3$. This proves property $(1)$.

Consider the case $\bigvee_{i = 1}^\ell x_i = 0$. Each section induces a disjoint union of a complete bipartite graph and an independent set. Each section only neighbors the non-section vertices. Moreover, the non-section vertices induce an independent set. Thus, any triangle has to contain one non-section vertex and two section vertices. Since the non-section vertices only neighbor the bottom subsections, these two section vertices have to be in the bottom subsections. If $x_i = 0$, there are no edges between the bottom subsections of the $i$-th section. Therefore, these two vertices cannot be from the $i$-th section. If $x_i = 0$ for all $i \in [\ell]$, there cannot be any triangles in the graph.

If $\bigvee_{i = 1}^\ell x_i = 1$, there exists $i$ such that $x_i = 1$. Consider the $i$-th section. There is a triangle for each triplet of vertices $u,v,w$ for $u$ in the bottom-left subsection, $v$ being in the bottom-right subsection and $w$ being a non-section vertex from the same unit. This means that the graph contains at least $(\sqrt{n T / m})^2 \cdot m/n = T$ triangles. This proves property $(2)$.


We finally argue that any full neighborhood query on $G$ can be simulated using a query on $x_j$ for some $j \in [\ell]$. Specifically, we prove that, given $i \in [|V|]$ ($i \in [|E_x|]$), the neighborhood of $v_i$ (of $v,w$ for $e_i = vw$) only depends on the value of $x_j$ for some $j$ (respectively).
The neighborhood of any non-section vertex is the same regardless of $x$ (and the query can thus be answered without making any queries on $x$). The neighborhood of a vertex in the $j$-th section is determined by the value of $x_i$. Vertex queries can thus be simulated by a single query on $x$. We now argue that this is also the case for edge queries. The edges are ordered such that for $j < j'$, all edges incident to the $j$-th section are in the ordering before those incident to the $j'$-th section.
The number of edges within a section is independent of $x$. 
We call this number $k$.
The $i$-th edge is then the $(i-k\lfloor (i-1)/k \rfloor)$-th vertex in the $(\lfloor (i-1)/k \rfloor + 1)$-th section and it thus only depends on $x_{\lfloor (i-1)/k \rfloor + 1}$. This proves the property $(3)$.
\end{proof}


\fi 

\ifconference
\else
\bibliographystyle{plainnat}
\fi
\bibliography{literature}
\end{document}